\documentclass{birkjour}

\usepackage{subcaption}
\usepackage[colorinlistoftodos]{todonotes}

\numberwithin{equation}{section}

\newtheorem{theorem}{Theorem}[section]

\newtheorem{lemma}[theorem]{Lemma}
\newtheorem{proposition}[theorem]{Proposition}

\newtheorem{assumption}[theorem]{Assumption}

\theoremstyle{definition}

\newtheorem{remark}[theorem]{Remark}

\newcommand{\rset}{\mathbb{R}}
\newcommand{\cset}{\mathbb{C}}
\newcommand{\Rd}{\,\mathrm{d}}
\newcommand{\tPsi}{\tilde\Psi} %
\newcommand{\tlambda}{\tilde{\lambda}} %
\newcommand{\tLambda}{\tilde\Lambda}  %
\newcommand{\tH}{\tilde H } %
\newcommand{\tA}{\tilde A } %
\newcommand{\bLambda}{\Lambda} %
\newcommand{\bH}{\bar{H}}
\newcommand{\bB}{\bar{B}} %
\newcommand{\bA}{\bar{A}} %
\newcommand{\brA}{\breve{A}} %
\newcommand{\by}{\bar{y}} %
\newcommand{\TR}{\mathrm{Tr}\,} %
\newcommand{\CSP}{\cset^{d\times d}}

\def\EXP#1{e^{#1}}
\def\OPER#1{\hat{#1}}
\def\OPERW#1{\widehat{#1}}
\def\PERIOD{\,.}
\def\COMMA{\,,}
\def\BIGO{\mathcal{O}}
\def\Id{\mathrm{I}}
\def\IU{\mathrm{i}}
\def\NORMFAC{\left(\frac{\sqrt M}{2\pi}\right)}
\def\WEYL#1{({#1})^{\OPERW{}}\,\,}
\def\HSPACE{\mathcal{H}}
\def\SSPACE{\mathcal{S}}
\def\FT{\mathcal{F}}

\def\MP{{\#}}
\def\XYH{\tfrac{1}{2}(x+y)}
\def\XXH#1#2{\tfrac{1}{2}({#1}+{#2})}

\def\ZBSTAR{\tilde{z}_b}
\def\HBSTAR{\tilde{H}_b}
\def\HSSTAR{\tilde{H}_s}

\title[Quantum observables approximated by molecular dynamics]%
{Canonical quantum observables \\
 for molecular systems approximated \\
 by ab initio molecular dynamics}
\author{Aku Kammonen}
\address{Institutionen f\"or Matematik, Kungl. Tekniska H\"ogskolan, \\100 44 Stockholm, Sweden}
\email{kammo@kth.se}
\author{Petr Plech\'{a}\v{c}}
\address{      Department of Mathematical Sciences, 
         University of Delaware, \\
         Newark, DE 19716, USA}
\email{plechac@udel.edu}
\author{Mattias Sandberg}
\address{Institutionen f\"or Matematik, Kungl. Tekniska H\"ogskolan,\\ 100 44 Stockholm, Sweden}
\email{msandb@kth.se}
\author{Anders Szepessy}
\address{Institutionen f\"or Matematik, Kungl. Tekniska H\"ogskolan,\\ 100 44 Stockholm, Sweden}
\email{szepessy@kth.se}
\thanks{The authors are grateful to Caroline Lasser for valuable help on this work.
The research was supported by
Swedish Research Council 621-2014-4776 and the Swedish e-Science Research Center. The research of P.P. was supported by ARO MURI Award No.  W911NF-14-024}
\subjclass{82C10, 81Q20}
\date{} 
\begin{document}

\begin{abstract}
It is known that ab initio molecular dynamics based on the electron ground state eigenvalue
can be used to approximate quantum observables in the canonical ensemble when the temperature is low compared to
the first electron eigenvalue gap. 
This work proves that a certain weighted average of the different ab initio dynamics,  corresponding to each electron eigenvalue, approximates quantum observables for any temperature.
The proof uses the semiclassical Weyl law to show that
canonical quantum observables of nuclei-electron systems, based on matrix valued Hamiltonian symbols, can be 
approximated by ab initio molecular dynamics with the error proportional to the 
electron-nuclei mass ratio. The result
covers observables that depend on time-correlations. 
A combination of the Hilbert-Schmidt inner product for quantum operators and Weyl's law
shows that the error estimate holds 
for observables and Hamiltonian symbols  that have three and five bounded derivatives, respectively, 
provided the electron eigenvalues are distinct for any nuclei position
and the observables are in the diagonal form with respect to the electron eigenstates.
\end{abstract}

\maketitle
\tableofcontents

%
%

\section{Introduction} 

\subsection{Background}
Given a quantum system defined by the Hamiltonian $\OPER{H}(x,-i\hbar \nabla)$ acting on $L^2(\rset^{N})$ the quantum canonical ensemble at the inverse temperature $\beta = 1/(k_B T)$ is described by the density operator $\OPER{\rho} = \EXP{-\beta\OPER{H}}$. A  quantum observable is defined by a Hermitian, densely defined operator $\OPER{A}$ on $L^2(\rset^{N})$ and the quantum canonical ensemble average is obtained from the normalized trace of the product as 
 $\frac{\TR({\OPER{\rho}\OPER{A}})}{\TR{\OPER{\rho}}}\PERIOD$
The Weyl quantization establishes a connection between the operator Hamiltonian $\OPER{H}(x,-i\hbar\nabla)$ and its  real-valued symbol function $H(x,p)$ defined on the classical phase space $\rset^{N}\times\rset^{N}$. The semiclassical analysis for $\hbar\to 0$ shows that the quantum observables 
can be approximated by the classical Gibbs ensemble average
\begin{equation}\label{classical_observable}
\frac{\int_{\rset^{N}}\int_{\rset^{N}} A(x,p) \EXP{-\beta H(x,p)}\, dx dp}{\int_{\rset^{N}}\int_{\rset^{N}} \EXP{-\beta H(x,p)}\, dx dp} \COMMA
\end{equation}
where the function $A(\cdot,\cdot)$ is the symbol of the 
Weyl quantized operator $\OPER{A}$.

The first mathematical result of such a semiclassical limit was obtained by Wigner \cite{wigner}. Wigner introduced
the ``Wigner''-function, based on solutions to the Schr\"odinger equation with scalar potentials,
and made an expansion in the Planck constant $\hbar$ to relate the canonical quantum observable 
to the classical Gibbs phase-space average.

We study molecular dynamics approximations of canonical ensemble averages  for quantum observables of the nuclei-electrons system. The role of the semiclassical parameter is played by the mass ratio $m_e/m_n$ between
electrons mass $m_e$ (the light particles) and the nuclei mass $m_n$ (the heavy particles). We assume the atomic units (a.u.) in which $\hbar=1$ and the mass of electrons $m_e = 1$, thus our small parameter is $1/M \ll 1$
where $M$ is the mass of nuclei in atomic units. Note that in atomic units the proton mass is $m_p\approx  1836$. The Boltzmann constant is $k_B=1$ in atomic units hence $\beta = 1/T$. 

The Hamiltonian of this system consists of the kinetic energy of nuclei and the electronic kinetic energy operator together with the operator describing interaction between electrons and nuclei
$$
\OPER{H} = -\frac{1}{2M} \Delta_{x} - \frac{1}{2} \Delta_{x_e} +\OPER{V}_e(x,x_e)\PERIOD
$$
In this work we treat the electronic kinetic energy operator and the interaction operator, $\OPER{H}_e = -\frac{1}{2}\Delta_{x_e}+\OPER{V}_e(x,x_e)$ as a matrix valued potential $V:\rset^N\to \cset^d$ defining the new Hamiltonian
\begin{equation}\label{ham_defin}
\OPER{H} = -\frac{1}{2M} \Id \otimes \Delta + \OPER{V}(x)\COMMA
\end{equation}
where $\Id$ is the $d\times d$ identity matrix. 
 The matrix-valued potential $V$ is then obtained by approximating $\OPER{H}_e$ on a finite dimensional subspace of electronic states. We assume that this approximation results in the Hermitian matrix valued confining potential $V$ with non-degenerate eigenvalues.
By including the electron part  as a matrix-valued operator
one can derive the limit as the electron-nuclei mass ratio $1/M$ tends to zero, see \cite{Teufel_matrix} which, in Section~6, includes an overview of previous results.
 This limit can then be approximated by
 ab initio molecular dynamics simulations for nuclei, with the potential generated by the electron eigenvalue problem, see \cite{lebris,marx_hutter}, based on the nuclei and electron scale separation using the Born-Oppenheimer approximation \cite{panati_spohn_teufel1, panati_spohn_teufel2}. Such molecular dynamics simulations have the benefit to require 
 less computational effort than to solve the Schr\"odinger equation with  time dependent electrons.

If the temperature is small in comparison to the minimal difference of the second and first eigenvalue of the electron potential $\OPER{V}$, the probability for the quantum system to be in excited electron states is negligible and the canonical ab initio molecular dynamics based on the electron ground state yields accurate approximations. When the temperature is not small compared to this electron eigenvalue gap, the probability to be in excited states is substantial and the molecular dynamics associated with the electronic ground state energy will not yield accurate approximation of quantum observables.
\subsection{Overview of results}
We address an important question 
which seems mathematically open:
How to modify the canonical {\it ab initio} molecular dynamics in order to accurately approximate quantum  observables
based on matrix valued potentials and all temperatures? 

We derive molecular dynamics methods that accurately, in mathematical sense,
approximate a quantum observable also in the case where the temperature can be large compared to the first spectral gap.
The approximation consists of a weighted sum 
 of molecular dynamics observables  for 
 the scalar Hamiltonians which are 
 the eigenvalues of the original matrix-valued Hamiltonian symbol.
 Furthermore, the weights, which are the probabilities to be in the corresponding electron states, are determined precisely 
as molecular dynamics observables. For instance, molecular systems with light atoms  and applications with laser heating
have been simulated more accurately by taking several electron states into account,
see e.g.  \cite{parinelo_frenkel}, \cite{ring-poly}
and Section 5.3.4 in \cite{marx_hutter}.

Section~\ref{sec_gibbs}  presents analysis
with quantum observables not depending on time. In Section~\ref{sec_corr} we study observables that are correlations in time or depending on time correlations.
The main result is 
that the weighted sums of molecular dynamics observables approximate
canonical quantum observables, based on Schr\"odinger Hamiltonians with matrix-valued potentials for any positive temperature.  
The approximation error is bounded by the square root of
electron-nuclei mass ratio, $M^{-1/2}$, times a constant,
provided 
the observable symbols are in the diagonal form with respect to the electron eigen-states and
have up to three derivatives bounded
in $L^2(\rset^{N}\times\rset^{N})$, the Hamiltonian symbol has five derivatives 
and the electron eigenvalues are distinct for any nuclei position. 
An improved approximation error $\mathcal{O}(M^{-1})$ holds with additional assumptions.
%
%
The main mathematical tool is the semiclassical Weyl law, as formulated e.g. in \cite{zworski,Teufel_matrix}.

The semiclassical Weyl law has been used before, see \cite{Teufel_matrix}, to approximate canonical quantum observables, based on general matrix valued Hamiltonians, including infinite dimension $d=\infty$, by phase-space averages. 
 For general Hamiltonian operators the classical approximation with $\BIGO(M^{-1})$ accuracy 
includes $\BIGO(M^{-1/2})$ perturbations of the leading order Hamiltonian dynamics, as described in \cite{Teufel_matrix} by insightful analysis of the Hamiltonian flow based on perturbed Hamiltonian functions and symplectic forms.

\medskip

With the aim to obtain sharper error estimates,
we focus on a simpler case from the analysis perspective, when the nuclei kinetic energy contribution to the Hamiltonian operator is defined by the Laplacian, as in \eqref{ham_defin},  and $d$ is finite. We exploit this simplified structure and the assumption of non crossing electron eigenvalues, to construct global projections to the electron states related to the adiabatic approximation.
We show how these projections, in the special case of the Hamiltonian \eqref{ham_defin}, can be determined by a nonlinear eigenvalue problem. 
This new approach with a non-linear eigenvalue problem has an advantage in providing more precise error estimates. 

Another novelty of our approach is that 
 our error estimates are not using  the Calderon-Vaillancourt theorem that bounds the $L^2$ operator norm
 of a symbol by $L^\infty$ estimates of derivatives of order $N$ of the symbol. 
 Instead we take advantage of the Hamiltonian form \eqref{ham_defin} to derive error estimates that combine  the Hilbert-Schmidt norm of quantum observables and Weyl's law
 to obtain new bounds in terms of the $L^2(\rset^{N}\times\rset^{N})$ norm of remainder symbols based on three and five derivatives of the observables and the Hamiltonian, respectively. The new bounds of the remainder symbols given in Lemmas \ref{moyal_lemma} and \ref{comp_lemma} are based on Hermitian properties of the Moyal composition.
 The constant in our $\BIGO(M^{-1})$ approximation result in Theorem~\ref{gibbs_corr_thm_analytic} depends e.g. on
 the $L^2(\rset^{N}\times\rset^{N})$ norm of three derivatives of the observable, while estimates based on the Calderon-Vaillancourt theorem
 require similar bounds on derivatives of order $N\gg 1$, which in practise can be very large 
 so that $\BIGO(M^{-1})$ would not lead to useful 
 error estimates for computational approximations with realistic values of $M$. 
 
In conclusion, we make a step in the direction of precisely estimating constants for semiclassical approximations. There are two main new mathematical ideas that take advantage of the particular form of the Hamiltonian \eqref{ham_defin}: a non-linear eigenvalue problem that avoids the traditional asymptotic expansion and an error estimate of remainder terms that use up to fourth order derivatives only. We think these two ideas can be useful to further trace constants in situations with e.g.\ avoided crossings, degenerate and crossing electron eigenvalues, vanishing temperature and more general Hamiltonians.

At the core of the canonical ensemble is the Gibbs distribution,
which has the important property
that it is the marginal distribution for  
a subsystem weakly coupled to a heat bath, where the
composite system is assumed to have the microcanonical distribution, see e.g. \cite{feynman}.
In Appendix~\ref{Gibbs_section} we present a variant of this property, assuming 
instead that the marginal distribution of the subsystem
is determined by the subsystem Hamiltonian. 
Another basic property is that the Gibbs density is a time independent solution of the 
Liouville-von Neumann equation. 
Appendix~\ref{Gibbs_section} also includes a comparison of the classical and quantum Gibbs 
densities with respect to the classical and quantum Liouville equations. 
In Section~\ref{sec_comp_demo} and Appendix~\ref{numerical_tests} we present numerical results of simple model problems where the quantum and molecular dynamics observables are compared.

%
%
\section{The Schr\"odinger equation and Gibbs ensembles}\label{sec_gibbs}

\subsection{Problem formulation and Weyl quantization}
We consider the matrix valued Schr\"odinger operator
\begin{equation*}\label{ham_defin1}
\OPER{H} = -\frac{1}{2M} \Id\otimes \Delta + V(x)\COMMA
\end{equation*}
where $V:\mathbb R^N \rightarrow \cset^{d\times d}$ is a  Hermitian matrix valued 
confining potential and $\Id\otimes \Delta $ is the $d\times d$ identity matrix 
times the Laplacian on $\mathbb R^N$, modeling the nuclei kinetic energy. 

Hence,  the quantum model consists of $N/3$ nuclei whose  coordinates are in $\rset^3$ 
and the wave functions $\Phi_n\in L^2(\rset^N,\cset^d)$ are vector-valued with $d$ components. We use the notation $\HSPACE := L^2(\rset^N,\cset^d)\equiv [L^2(\rset^N)]^d$.
Approximation of the electronic part of the Hamiltonian using the electron eigenvalue problem gives the matrix potential (operator) $V(x)$ defined for each nuclei configuration $x\in\rset^N$.
Here $M\gg 1$ is the nuclei-electron mass ratio, assumed to be much larger than one.
The setting with individual nuclei masses and a diagonal mass matrix $M$ 
can be transformed to the form \eqref{schrod_eq}
by introducing the new coordinates $M_1^{1/2}\bar x=M^{1/2}x$, which transforms 
the Hamiltonian into
\[
-\frac{1}{2M_1} \Id\otimes \Delta_{\bar x}  + V(M_1^{1/2}M^{-1/2}\bar x)\PERIOD
\]
 
To handle the spectrum of $\OPER{H}$ 
we assume that the
smallest eigenvalue $\lambda_1(x)$ of $V(x)$ tends to infinity as
$|x|\rightarrow\infty$.
This assumption implies that the spectrum of $\OPER{H}$ 
is discrete, see \cite{dellara}. More precisely,
let $\Phi_n:\mathbb R^N\rightarrow \cset^d$ and $E_n\in\rset$ for $n=1,2,3,\ldots$  be solutions  
of the eigenvalue problem
\begin{equation}\label{schrod_eq}
\OPER{H}\Phi_n = E_n\Phi_n\COMMA
\end{equation}
then the set of eigenfunctions $\{\Phi_n\}_{n=1}^\infty$ forms a complete 
basis of the Hilbert space $\HSPACE$. 

To have a complete set of eigenfunctions in $\HSPACE$ is used  for the 
 analysis of the canonical quantum ensemble average in this work, although it is not crucial.  The approach we present 
 is based on the concept of the trace of quantum operators introduced 
 by von Neumann, \cite{neumann}.

\medskip
\noindent{\it Notation and Weyl calculus.}
Since our analysis provides error estimates for approximating quantum observables with classical ones it is natural to use tools of Weyl calculus that defines correspondence between operators on $\HSPACE$ and their symbols in a suitable functional space. Here we introduce the notation used throughout this paper.  

For functions $u,v \in \HSPACE $ we denote the scalar product
\[
\langle v,w\rangle = \int_{\rset^N} 
v^*(x)\cdot w(x)\, \Rd x\COMMA\;\;\;
\mbox{where  $v^*(x)\cdot w(x) := \sum_{j=1}^d v_j^*(x)w_j(x)$} \COMMA
\]
and the corresponding norm $\|u\|^2_\HSPACE = \int_{\rset^N} |u(x)|^2 \Rd x$. The space of smooth rapidly decaying matrix-valued functions, i.e., Schwartz space, is denoted $\SSPACE(\rset^N\times\rset^N,\cset^{d\times d})$ and it is abbreviated as $\SSPACE$. 
We define the Fourier transform $\FT: A(x,p)\mapsto \FT[A](\xi_x,\xi_p)$
\begin{equation}\label{FT:definition}
\FT[A](\xi_x,\xi_p) = \int_{\rset^{N}}\int_{\rset^{N}} A(x',p') \EXP{-\IU (x'\cdot\xi_x + p'\cdot \xi_p)}\Rd x' \Rd p'
\PERIOD
\end{equation}
We emphasize that the Fourier transform of a symbol $A$ is denoted by $\FT[A]$ while $\OPER{A}$ denotes the Weyl quantization of the symbol $A$. We also use the notation $\WEYL{A}$ instead of the simple $\OPER{A}$, in particular in long expressions such as $\WEYL{A\EXP{-\beta H}}$.  
We define the Weyl quantization operator of the matrix-valued symbol $A\in \SSPACE$ as 
the mapping $A\mapsto \OPER{A}$
that assigns to the symbol $A$ the linear operator 
$\OPER{A}:\HSPACE \to \HSPACE$ defined 
for all $\Phi\in\SSPACE(\rset^N,\cset^d)$ by
\begin{equation}\label{Weyl:definition}
  \OPER{A} \Phi(x)=\NORMFAC^N
  \int_{\rset^{N}}\int_{\rset^{N}} \EXP{\IU M^{1/2}(x-y)\cdot p} 
   A(\XYH,p) \Phi(y) \Rd p \Rd y \PERIOD
\end{equation}
For example, the symbol $H(x,p):=\tfrac{1}{2}|p|^2 \Id + V(x)$ yields 
$\OPER{H} = -(2M)^{-1}\Id\otimes \Delta + V(x)$.
The definition \eqref{Weyl:definition} implies  
that any quantisation $\OPER{A}$ is an integral operator
\[
\OPER{A} \Phi(x)= \int_{\rset^{N}} K_A(x,y)\Phi(y) \Rd y\COMMA 
\]
with a matrix-valued, distributional kernel $K_A(x,y)$ 
\begin{equation}\label{Kernel:definition}
   K_A(x,y)= \NORMFAC^N
   \int_{\rset^{N}}\EXP{\IU M^{1/2}(x-y)\cdot p} A(\XYH,p)  \Rd p\PERIOD
\end{equation}
The expression above shows that $K_A$ is the Fourier transform in the second argument of the symbol $A(x,p)$ and consequently
the Weyl quantization is well defined for
symbols in $\SSPACE'$, the space of tempered distributions.

An important property of the Weyl quantization is given by the Moyal product $A\MP B$ of two symbols $A$, $B$ 
\begin{equation}\label{Moyal:definition}
[A\MP B](x,p) = \left.\EXP{\tfrac{\IU}{2\sqrt{M}}(\nabla_{x'}\cdot\nabla_{p} - \nabla_{x}\cdot\nabla_{p'})} A(x,p)B(x',p') \right|_{(x,p)=(x',p')}\PERIOD
\end{equation}
The Moyal product provides correspondence between the algebra of operators and the algebra of their symbols by identifying composition of two operators with the Weyl quantization of the Moyal product of their symbols,  
more precisely
\begin{equation}\label{Composition:rule}
\WEYL{A\MP B} = \WEYL{A}\WEYL{B}\PERIOD
\end{equation}
Further properties and background on the Weyl calculus are presented, e.g., in \cite{zworski,folland2}.

\medskip

The principal idea in this work is to study the trace of operators on $\HSPACE$ with kernels $K_A$ defined by \eqref{Kernel:definition}.
The trace is a composition of the
trace in $L^2(\rset^N)$ and the trace of $d\times d$ matrices. 
The two different traces are defined by
\[
\begin{split}
\TR{\OPER{B}} &:=
\sum_{n=1}^\infty \langle \Phi_n,\OPER{B}\Phi_n\rangle\COMMA \ \mbox{ for an operator $\OPER{B}$ on $\HSPACE$,} \\
\TR{ B}&:=\sum_{j=1}^d B_{jj}\, ,\ \mbox{ for a $d\times d$ matrix } B\, . \\
\end{split}
\]
The analysis here uses the fact that
the 
$\HSPACE$-trace of a Weyl operator based on a $d\times d$ 
matrix valued symbol is equal to
the phase-space average of its symbol trace. Indeed  we have by \eqref{Kernel:definition} for $A\in\SSPACE$
\begin{equation}\label{trace1}
\TR{\OPER{A}} =
\int_{\mathbb R^{N}} \TR K_A(x,x) \Rd x = \NORMFAC^N\int_{\mathbb R^{N}}
\int_{\mathbb R^{N}}
\TR  A(x,p)  \Rd p\, \Rd x\PERIOD
\end{equation}
We introduce the coordinate $z=(x,p)\in\rset^{2N}$ in the phase space and its Lebesgue measure  $\Rd z=\Rd x\Rd p$.
In fact also the composition of two Weyl operators is determined by the phase-space average as follows.
\begin{lemma}\label{composition}
The composition of  two Weyl operators $\OPER A$ and $\OPER B$, with $A\in\SSPACE$ and $B\in\SSPACE$ satisfies
\[
\TR(\OPER{A}\OPER{B})=\NORMFAC^N\int_{\mathbb R^{2N}}  \TR\big(A(z)B(z)\big) \Rd z\COMMA
\]
where $A(z)B(z)$ is the matrix product of the two $d\times d$ matrices $A(z)$ and $B(z)$.
\end{lemma}
Although this result is not new, see \cite{Teufel_matrix}, we include a proof since it is important for the work here.
\begin{proof}
The kernel of the composition is
\[
\begin{split}
   K_{AB}(x,y) &=
    \NORMFAC^{2N}\int_{\mathbb R^{3N}}
      A(\XXH{x}{x'},p)B(\XXH{x'}{y},p')\\
      &\qquad\times \EXP{iM^{1/2}\big((x-x')\cdot p + (x'-y)\cdot p'\big)} \Rd p' \Rd p \Rd x'
\end{split}
\]
so that the trace of the composition becomes
\[
\begin{split}
    &\TR(\OPER{A}\OPER{B}) =\int_{\mathbb R^N}\TR K_{AB}(x,x) \Rd x\\
    &=\NORMFAC^{2N}\int_{\mathbb R^{4N}}
       \TR\big(A(\XXH{x}{x'},p)B(\XXH{x'}{x},p')\big) \\
&\qquad \times \EXP{iM^{1/2}\big((x-x')\cdot p + (x'-x)\cdot p'\big)} \Rd p' \Rd p \Rd x' \Rd x\\
&=\NORMFAC^{2N}\int_{\mathbb R^{4N}}
\TR\big(A(y,p)B(y,p')\big) \EXP{iM^{1/2}y'\cdot (p-p')} \Rd p' \Rd p \Rd y' \Rd y\\
&=\NORMFAC^{N}\int_{\mathbb R^{2N}}
\TR\big(A(y,p)B(y,p)\big)  \Rd p  \Rd y\COMMA\\
\end{split}
\]
using the change of variables $(y,y')=\big((x+x')/2,x-x'\big)$, which verifies the claim.
\end{proof}

The isometry between  Weyl operators with the Hilbert-Schmidt  inner product,  $\TR({\OPER{A}^*\OPER{B}})$, and the corresponding $L^2(\rset^N\times\rset^N,\mathbb C^{d\times d})$ symbols obtained by  Lemma~\ref{composition} shows how to extend from symbols in $\SSPACE$ to $L^2(\rset^N\times\rset^N,\mathbb C^{d\times d})$ by density of $\SSPACE$
in $L^2(\rset^N\times\rset^N,\mathbb C^{d\times d})$, see \cite{Teufel_matrix}. We will use the Hilbert-Schmidt norm 
$\|\OPER{A}\|_{\mathcal{HS}}^2=\TR(\OPER{A}^*\OPER{A})=\TR(\OPER{A}^2)$,
to estimate Weyl operators.

%
%
\subsection{Gibbs density operator and its approximation} \label{sec:gibbs_intro}

The quantum statistical average of the (time-independent) observable $\OPER{A}$ in the canonical ensemble at the inverse temperature $\beta$ is given by
\begin{equation}\label{Gibbs_observable}
\frac{\TR(\OPER{\rho}\OPER{A})}{\TR(\OPER{\rho})}\COMMA\;\;
\mbox{with (non-normalized) density operator 
$\OPER{\rho} = \EXP{-\beta\OPER{H}}$.}
\end{equation}
Similarly the time-dependent or time-correlation  observables becomes
\begin{eqnarray}\label{Gibbs_timeobservable}
\frac{\TR( \OPER{A}_\tau\OPER{B}_0 \OPER{\rho})}{\TR(\OPER{\rho})}\COMMA\;\;\;\mbox{and related}\;\;
\frac{\TR(\OPER{A}_\tau(\OPER{B}_0 \OPER{\rho} + \OPER{\rho} \OPER{B}_0)}{\TR(\OPER{\rho})}\COMMA
\end{eqnarray}
where $\OPER{C}_\tau:=\EXP{{\rm i}\tau M^{1/2}\hat H} \hat C\EXP{-{\rm i}\tau M^{1/2}\hat H}$ is the quantum observable at time $\tau$ with $\OPER{C}_0=\hat C$,  
as presented more precisely  in Section \ref{sec_corr}. 

The Weyl quantization provides a correspondence  between quantum operators $\OPER{A}$ and their classical symbols $A$.
The quantum canonical ensemble is described by the density matrix operator $\OPER{\rho} = \EXP{-\beta \OPER{H}}$ while the classical canonical ensemble is defined by the Gibbs distribution $\mu({\rm d}z) \sim \EXP{-\beta H}\,{\rm d}z$ on the phase space $z=(x,p)$. 

Lemma~\ref{composition} suggests that the correspondence between the quantum and classical Gibbs observables can be achieved if we use as the density matrix operator the Weyl quantization $\OPERW{\EXP{-\beta H}}$. In Section~\ref{sec_corr} we derive error estimates that show that quantum observables such as those in \eqref{Gibbs_timeobservable} can be approximated by a classical Gibbs observables in \eqref{classical_observable}, derived from the Hamiltonian symbol $H$, and thus linked to the classical molecular dynamics. Theorems~\ref{gibbs_corr_thm_analytic}, \ref{gibbs_corr_thm} and \ref{gibbs_corr_thm_unif} prove  error estimates both
if the density operator used in \eqref{Gibbs_timeobservable} is $\OPER{\rho} = \EXP{-\beta\OPER{H}}$ and if it is replaced
with Weyl quantization  $\OPERW{\EXP{-\beta H}}$.


%
%
Replacing the standard density operator $\EXP{-\beta\OPER{H}}$ with the operator $\OPERW{\EXP{-\beta H}}$ raises a question which one should be taken as a reference for error analysis. We show in the proof of Theorem~\ref{gibbs_corr_thm} that the observables based on these two operators differ  by a quantity of order $\BIGO(M^{-1})$ when the number of 
particles, $N$, is small compared to $M$. Thus in this case one can use either of them. However, the standard density operator $\EXP{-\beta \OPER{H}}$ is the stationary solution of the (quantum) Liouville-von Neumann equation while the corresponding symbol is not the time-independent solution of the classical Liouville equation. On the other hand starting with the classical Gibbs density $\EXP{-\beta H}$ the corresponding Weyl quantization gives the proposed density operator $\OPERW{\EXP{-\beta H}}$ which is not a stationary solution of the Liouville-von Neumann equation. We discuss this issue in Appendix~\ref{Gibbs_section}.

\medskip

The trace property of Lemma~\ref{composition} shows that an approximation  of the Gibbs observable,
where the order of the operations of exponentiation and quantization have been reversed, satisfies 
\[
\begin{split}
\TR(\OPERW{\EXP{-\beta H}}\OPER{A}) &= \sum_{n=1}^\infty \langle \Phi_n,\OPERW{\EXP{-\beta H}}\OPER{A}\Phi_n\rangle \\
&=\NORMFAC^{N}\int_{\mathbb R^N} \TR\big( A(z)\EXP{-\beta H(z)}\big) \Rd z
\end{split}
\]
and in the normalized form
\begin{equation}\label{S_def}
\begin{split}
G&:= \frac{ \TR(\OPERW{\EXP{-\beta H}}\OPER{A})}{ \TR(\OPERW{\EXP{-\beta H}})}
=\frac{\int_{\mathbb R^{2N}} \TR (A(z)\EXP{-\beta H(z)}) \Rd z}{
\int_{\mathbb R^{2N}} \TR (\EXP{-\beta H(z)}) \Rd z} \PERIOD
\end{split}
\end{equation}
We will use the diagonalized form of the Hamiltonian symbol and similarly of transformed observables.
More precisely, let  $\tA:\mathbb R^{2N}\rightarrow \CSP$ be a symbol in the Schwartz class
and consider symbols in the matrix product form $A(z)=\tPsi(x)\tA(z)\tPsi^*(x)$,
where $\tPsi(x)$ is the $d\times d$ matrix composed of the eigenvectors to $V(x)$ as columns, i.e.,
\begin{equation}\label{v_lambda_psi}
  \sum_{j=1}^d V_{ij}(x)\tPsi_{jk}(x)=\tlambda_k(x)\tPsi_{ik}(x)\COMMA
\end{equation}
and $\tlambda_k(x), k=1,\ldots, d$, denote the eigenvalues of $V(x)$ in the increasing order.
Here $\tPsi^*(x)$ is the Hermitian adjoint of $\tPsi(x)$ in $\cset^d$ and 
$\delta_{jk}$ is the Kronecker delta.
Then $\tPsi(x)$ diagonalizes $H(z)$ and it and  $A(z)$ satisfy:
\begin{equation}\label{tilde_H}
\begin{split}
H(z) &=\tPsi(x)\tH(z)\tPsi^*(x),\qquad \mbox{ where } \tH_{jk}(z)=\delta_{jk} 
\big(\frac{|p|^2}{2} + \tlambda_j(x)\big)\COMMA\\
A(z)&=\tPsi(x)\tA(z)\tPsi^*(x), \qquad \mbox{ where } 
\tA(z)=\tPsi^*(x)A(z)\tPsi(x)\PERIOD %
\end{split}
\end{equation}
The diagonal property of $\EXP{-\beta \tH}$ shows that the trace satisfies
\[
\begin{split}
\TR(A(z)\EXP{-\beta H(z)}) &=\TR (\tA(z)\EXP{-\beta \tH(z)})
=\sum_{j=1}^d \tA_{jj}(z)\EXP{-\beta \tH_{jj}(z)}\COMMA
\end{split}
\]
which only requires the diagonal part of $\tA$.
We obtain by \eqref{S_def} the approximate Gibbs quantum observable as a sum 
\[
G=\frac{\sum_{j=1}^d\int_{\mathbb R^{2N}} \tA_{jj}(z) \EXP{-\beta \tH_{jj}(z)} \Rd z}{
\sum_{j=1}^d\int_{\mathbb R^{2N}}  \EXP{-\beta \tH_{jj}(z)} \Rd z} 
\]
and each term can be written in canonical ensemble form:

\begin{lemma}\label{thm_stat} The approximate canonical ensemble average satisfies, for $\beta>0$,
\begin{equation}\label{S_def_q}
    \frac{ \TR(\OPERW{\EXP{-\beta H}}\OPER{A})}{ \TR(\OPERW{\EXP{-\beta H}})}
    =\sum_{j=1}^d q_j 
    \int_{\mathbb R^{2N}} \tA_{jj}(z) \frac{\EXP{-\beta \tH_{jj}(z)}\Rd z}{\int_{\mathbb R^{2N}} \EXP{-\beta \tH_{jj}(z')} \Rd z'} 
\end{equation}
with the weights given by respective probability to be in the state $j$ 
\begin{equation}\label{qh_def}
   q_j=q_j(\tH):=\frac{\int_{\mathbb R^{2N}}\EXP{-\beta \tH_{jj}(z)} \Rd z}{
   \sum_{k=1}^d\int_{\mathbb R^{2N}}\EXP{-\beta \tH_{kk}(z')} \Rd z'}\COMMA\quad j=1,\ldots, d \PERIOD
\end{equation}
\end{lemma}

In Section~\ref{sec_corr} we analyze the trace based
on  the time independent Gibbs density operator 
$\EXP{-\beta \OPER{H}}$ and on $\widehat{\EXP{-\beta H}}$, using instead the transformation
$\OPER{A}=\OPER{\Psi}\OPER{\bar A}\OPER{\Psi}^*$ for a diagonal symbol 
$\bar A:[0,\infty)\times\rset^{2N}\rightarrow \CSP$
with an orthogonal matrix $\Psi:\rset^{N}\rightarrow\CSP$ diagonalizing 
a non-linear perturbation of the eigenvalue problem \eqref{v_lambda_psi}. 
\begin{remark}\label{projections} 
Since $\Psi_{\cdot j}(\Psi_{\cdot j})^*$ is the projection to the electron state $j$, a method to handle projections to electron states 
is  to normalize with respect to
that state and use
\begin{equation*}\label{indep}
\begin{split}
\frac{\sum_{n=1}^\infty \langle \Phi_n, \big( \Psi_{\cdot j}\tA_{jj}(\Psi_{\cdot j})^*\big)^{\widehat{}} 
\ \big(\EXP{-\beta H}\big)^{\widehat{}} \ \Phi_n\rangle}{
\sum_{n=1}^\infty \langle \Phi_n, 
\big(\Psi_{\cdot j} (\Psi_{\cdot j})^*\big)^{\widehat{}}\  \big(\EXP{-\beta H}\big)^{\widehat{}} \ \Phi_n\rangle
}=   \int_{\rset^{2N}} \tA_{jj}(z)
\frac{\EXP{-\beta \tH_{jj}(z)}}{\int_{\rset^{2N}}\EXP{-\beta \tH_{jj}(\bar z)} d\bar z} \Rd z\PERIOD
\end{split}
\end{equation*}$\square$
\end{remark}

\subsection{Computational demonstration}\label{sec_comp_demo}
In order to demonstrate computational consequences of the presented analysis we formulate a simple model problem for comparing quantum observables with observables obtained from molecular dynamics. 
We consider a model Hamiltonian $\hat H = -M^{-1}\Id\otimes\Delta + V(x)$, where $V: \mathbb{R}\to \mathbb{R}^{2\times 2}$ and $\Id$ is the $2\times 2$ identity matrix. The corresponding Schr\"odinger eigenvalue equation $\OPER{H} \Phi_n = E_n\Phi_n$, where $E_n \in \mathbb{R}$ and $\Phi_n:\mathbb{R}\rightarrow\mathbb{R}^2$ represents a model with
a heavy particle with coordinate $x\in\rset$ and a two state light electron particle.
The details of the example,  numerical implementation, and specific values of the parameters are described in Appendix~\ref{numerical_tests}.

\subsubsection{Equilibrium observables}\label{sec_equilib1}
First we focus on equilibrium position observables.
We demonstrate numerically an estimate of the difference of the quantum and classical canonical observables, 
given by \eqref{Gibbs_observable} and the right hand side in \eqref{S_def_q}, respectively. 
We view the averaging for the quantum, 
$\int_\rset A(x)\,\mu_{\mathrm{qc}}(x)\Rd x$, and classical, 
$\int_\rset A(x)\,\mu_{\mathrm{cl}}(x) \Rd x$, position observables as averaging with respect to the measures with the densities
\begin{eqnarray}
\mu_{\mathrm{qc}}(x) &=&  \frac{\sum_{n}  |\Phi_n(x)|^2  \EXP{-\beta E_n}}{\int_{\rset}\sum_{n} |\Phi_n( x')|^2 \EXP{-\beta E_n}\Rd x'} \COMMA \label{qc_density}\\
\mu_{\mathrm{cl}}(x) &=&  \sum_{k=1}^2 q_k \frac{\EXP{-\beta \lambda_k(x)}}{\int_{\rset} \EXP{-\beta \lambda_k(x')}\Rd x'} \COMMA\;\;
q_k = \frac{\int_{\rset} \EXP{-\beta \lambda_{k}(x)} \Rd x}{\sum_{j = 1}^{2}\int_{\mathbb{R}} \EXP{-\beta \lambda_{j}(x)} \Rd x}
\label{cl_density}\PERIOD
\end{eqnarray}
Here $(E_n,\Phi_n(x))$ denotes an eigen-pair of the operator $\OPER{H}$ and $\lambda_{k}(x)$, $k=1,2$  is the $k$-th eigenvalue of the matrix-valued function $V(x)$, such that $\lambda_{1}(x)\leq \lambda_{2}(x)$.

\begin{figure}
 \includegraphics[height=3.5cm]{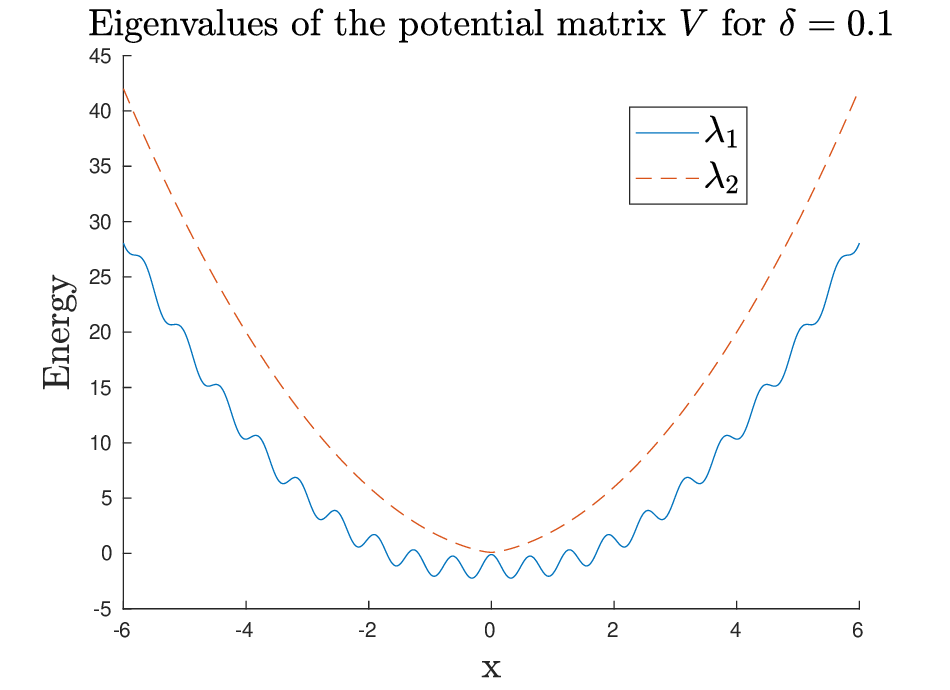}
    \caption{The eigenvalue functions $\lambda_1(x)$ and $\lambda_2(x)$ of $V(x)$.}
    \label{image:examplePotentialMatrix}
 \end{figure}

Figure~\ref{image:twoMDdensities} shows the classical density $\mu_{\mathrm{cl}}$ from molecular dynamics on the electron ground state, corresponding to $q_1=1$ and $q_2=0$, and the density given by the values of $q_k$ in \eqref{cl_density}. For the choice of parameters these probabilities are $q_1=0.8$ and $q_2=0.2$.
Figure~\ref{image:density} shows the Schr\"odinger and classical densities $\mu_\mathrm{qc}(x)$ and $\mu_\mathrm{cl}(x)$, computed using \eqref{qc_density} and \eqref{cl_density}. We note that the classical density computed using \eqref{cl_density} approximates the quantum density quite well, whereas the Figure~\ref{image:twoMDdensities} shows that the classical density computed using only the electron ground state is not close to the quantum density.
The potential $V(x)$ that was used in these numerical tests is described in 
Appendix~\ref{numerical_tests}, and has eigenvalues $\lambda_k(x)$ depicted in Figure~\ref{image:examplePotentialMatrix}. 
%
%
\begin{figure}
  \begin{subfigure}{0.49\textwidth}
    \includegraphics[width=0.9\linewidth]{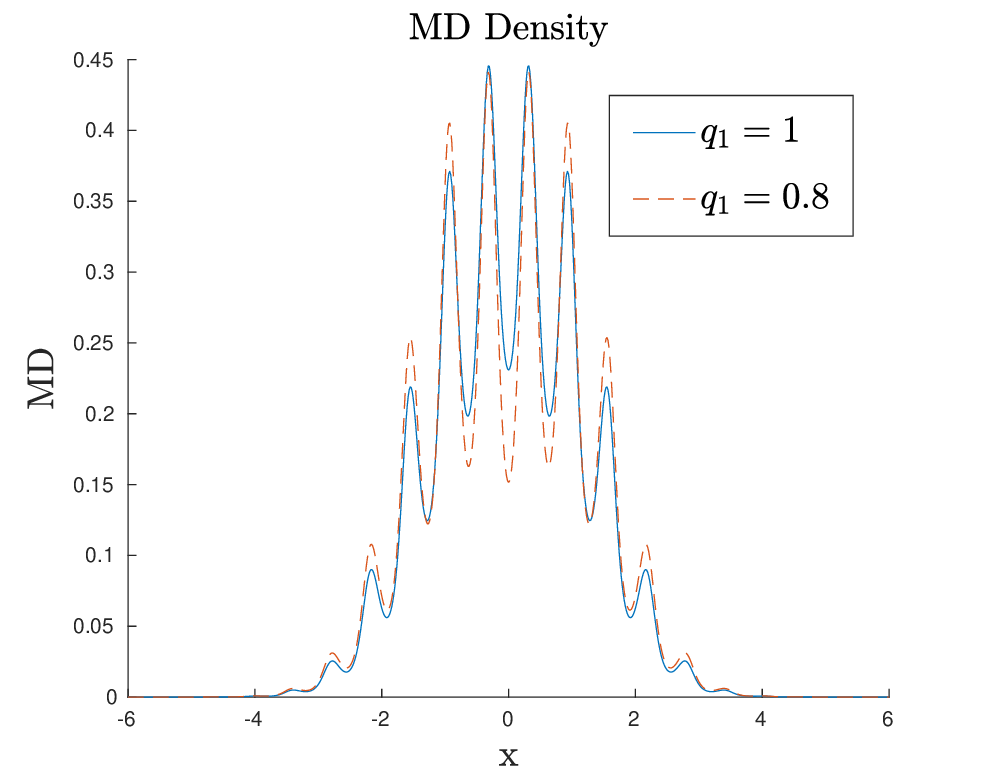}
    \caption{}
    \label{image:twoMDdensities}
  \end{subfigure}
  \begin{subfigure}{0.49\textwidth}
     \includegraphics[width=0.9\linewidth]{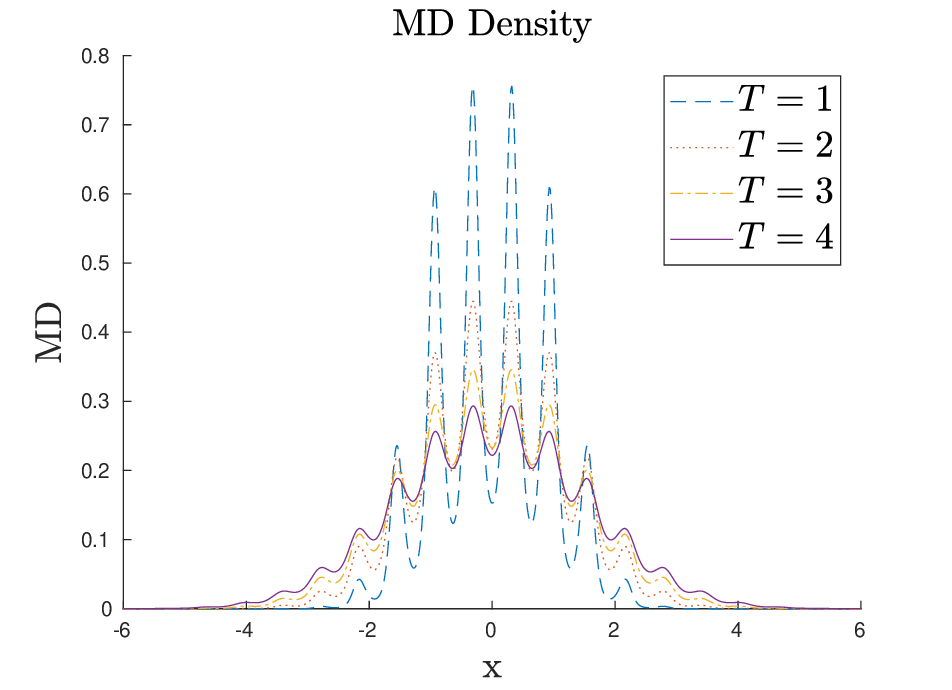}
     \caption{}
     \label{fig:mddenT}
  \end{subfigure}
  \caption{(a) The classical density $\mu_\mathrm{cl}$ using $q_1=1$, corresponding to molecular dynamics on the canonical ground state at $T=1.9947$ and $M=1000$ compared to the classical density with $q_1=0.8$, from \eqref{cl_density}.
   (b) Molecular dynamics density \eqref{cl_density} for different temperatures $T$ and $M=1000$ 
    }
  \label{fig:eigenvalues-1d-3s}
\end{figure}

\begin{figure}
\begin{subfigure}{0.49\textwidth}
    \includegraphics[width=0.9\linewidth]{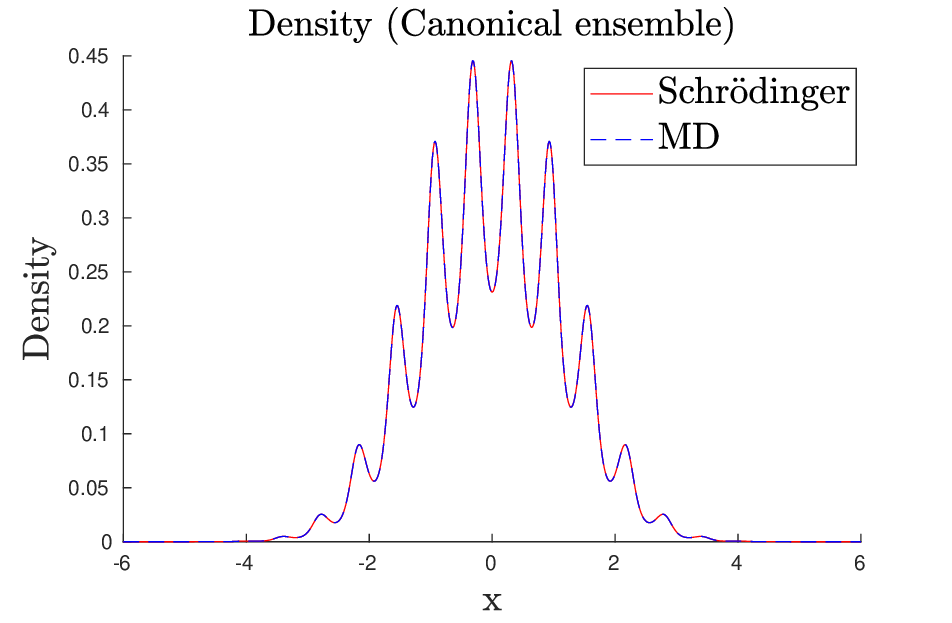}
    \caption{}
    \label{image:density}
  \end{subfigure}
  \begin{subfigure}{0.49\textwidth}
     \includegraphics[width=0.9\linewidth]{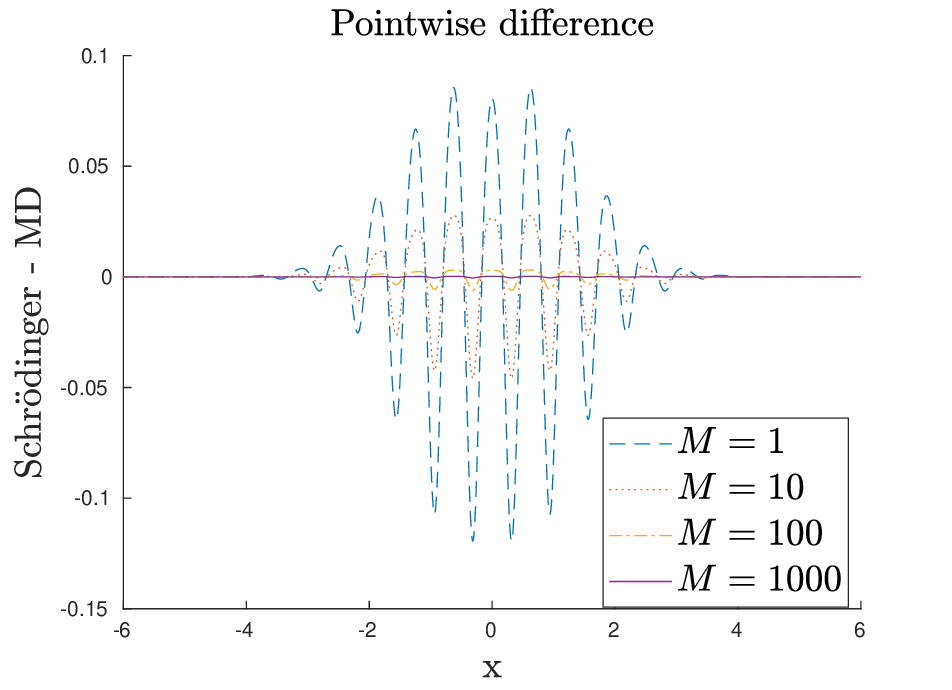}
     \caption{}
     \label{fig:25b}
  \end{subfigure}
\caption{(a) The quantum and classical densities in \eqref{qc_density} and \eqref{cl_density} plotted for the same parameters $M=1000$ and $T=1.9947$. 
(b) Difference of the classical and quantum pointwise densities, given in
\eqref{cl_density} and \eqref{qc_density},  for different mass ratio $M$ at $T=1.9947$ (right).}\label{fig:25}
\end{figure}



Results from numerical simulations demonstrate the error analysis of Section~\ref{sec_corr}, in particular that the $L^1$ as well as $L^\infty$ difference between $\mu_\mathrm{qc}$ and $\mu_\mathrm{cl}$ decreases as $\BIGO(M^{-1})$. 
We note the striking agreement between the Schr\"odinger and molecular dynamics equilibrium densities for $M=1000$ in Figure~\ref{image:density}.
The pointwise difference for $M=1,10,100,1000$ is given in Figure~\ref{fig:25b}. 
Figure~\ref{fig:mddenT} shows the increasing variation of the molecular dynamics density as the temperature decreases. 
We conclude that although, the densities vary substantially for different temperature, the molecular dynamics density approximates the canonical Schr\"odinger density well, with the error $\mathcal O(M^{-1})$ shown also in Figure~\ref{fig:1D-delta-pE}. Molecular dynamics approximation of quantum observables in the micro canonical ensemble typically have larger errors, see \cite{bhkpss}.

Figure~\ref{fig:1D-delta-pE} depicts the $\BIGO{(M^{-1})}$ decrease of the error for different values of the spectral gap between the ground state $\lambda_1(x)$ and the next state $\lambda_2(x)$ at a fixed temperature.
 Both errors are inverse proportional to the mass ratio $M$. The temperature $T$ and the spectral gap (controlled by the parameter $\delta$) are chosen so that the weight $q_1 = 0.8$, in other words the excited state contributes non-trivially to the average. In neither of the norms, 
 in Figure~\ref{fig:1D-delta-pE}, can we see a $\delta$ dependency in the error while in our derivation the constant in $\BIGO(M^{-1})$ of \eqref{eqn:analyticalCorrelation}  depends on the norm of the derivatives up to order five 
 of the eigenvectors and eigenvalues of $V$. 
 In this example the eigenvector derivative of the order $n$ is $\BIGO\left(\delta^{-n}\right)$.

\begin{figure}
\begin{subfigure}{0.49\textwidth}
  \includegraphics[width=0.9\linewidth]{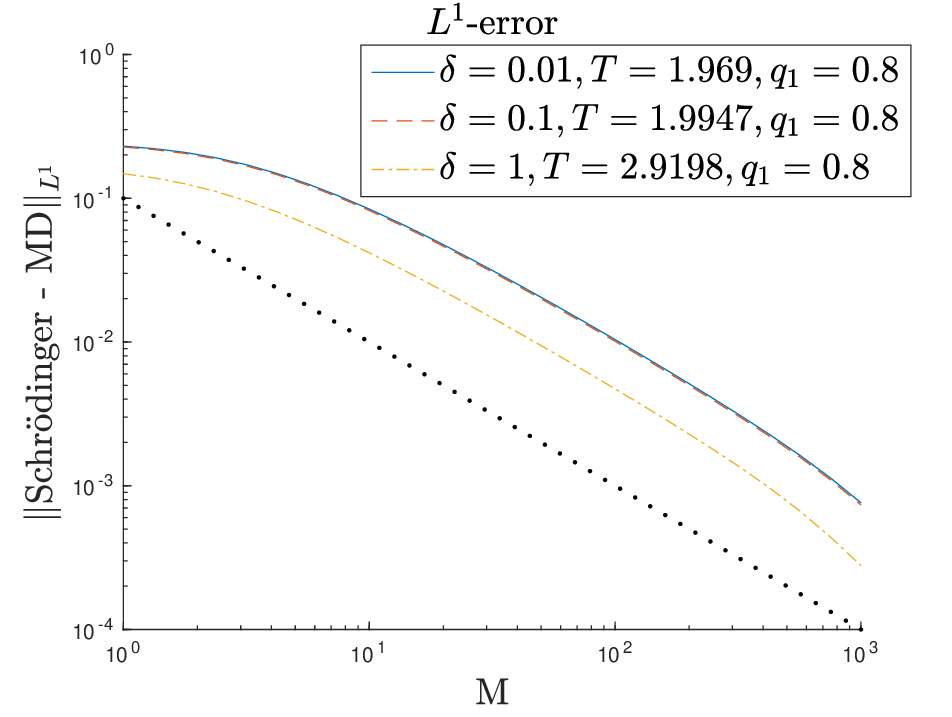}
  \caption{$L^1$ error}
\end{subfigure}
\begin{subfigure}{0.49\textwidth}
  \includegraphics[width=0.9\linewidth]{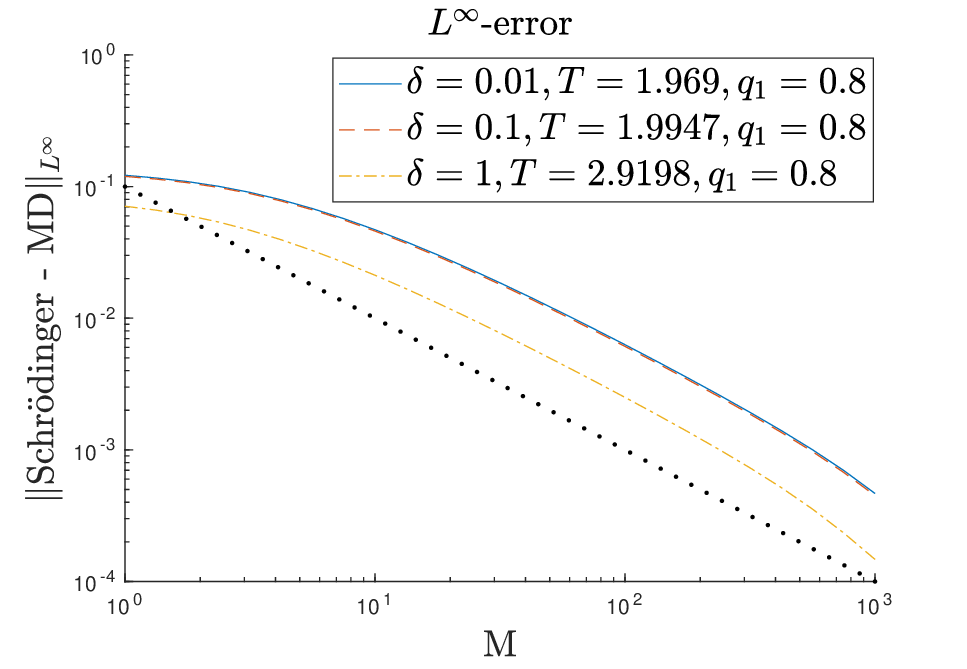}
  \caption{$L^\infty$ error}
\end{subfigure}
\caption{Dependence on the mass $M$ of the error between quantum and molecular dynamics densities $\mu_\mathrm{qc}$ and $\mu_{\mathrm{cl}}$ respectively, shown in log-log scale.  
The dotted lines show the reference slope $-1$.
See Appendix~\ref{numerical_tests} for details on the choice of the parameter $\delta$ and the temperature $T$.
}\label{fig:1D-delta-pE}
\end{figure}

\subsubsection{Time-correlated observables}
Next we demonstrate that the derived method and analysis is also applicable to observables depending on time correlations. In particular, we test the position observable $x_0$ at time $t=0$ with the position observable $x_\tau$ at time $t=\tau$. The time evolution of the position observable operator is given, in Heisenberg representation, as 
\begin{equation*}
\OPER{x}_{\tau} := \EXP{\IU\tau\sqrt{M}\OPER{H}}\OPER{x}_0 \EXP{-\IU\tau\sqrt{M}\OPER{H}}.
\end{equation*}
Applied to the special case in this example, with $\hat H$ as in Section \ref{sec_equilib1}, 
Theorems \ref{gibbs_corr_thm_analytic}, \ref{gibbs_corr_thm} and \ref{gibbs_corr_thm_unif} show that
\begin{equation}\label{eqn:analyticalCorrelation}
\begin{aligned}
&\frac{\frac{1}{2}\TR\big(\OPER{x}_{\tau}(\OPER{x}_{0}\EXP{-\beta\OPER{H}}+\EXP{-\beta\OPER{H}}\OPER{x}_{0})\big)}{\TR(\EXP{-\beta\OPER{H}})}\\ &= \sum_{j = 1}^{2}\int_{\mathbb{R}^2}q_j x_{\tau}^{j}(z_0)x_{0}^{j}(z_0)\frac{\EXP{-\beta (\frac{|p_{0}|^2}{2} + \lambda_j(x_0))}}{\int_{\mathbb{R}^2}\EXP{-\beta (\frac{|p|^2}{2} + \lambda_j(x))}\Rd z}\Rd z_0 + \BIGO(M^{-1}),
\end{aligned}
\end{equation}
where $z_{\tau}^{j} = (x_{\tau}^{j}, p_{\tau}^{j})$, $j=1,2$ solve the Hamiltonian dynamics
\begin{equation}\label{eqn:HamiltonDynamics}
\begin{cases}
\dot{x}_{\tau}^{j} = p_{\tau}^{j} \\
\dot{p}_{\tau}^{j} = -\frac{\Rd}{\Rd x_{\tau}^{j}}\lambda_j(x_{\tau}^{j}),\;\;\;\; \tau > 0
\end{cases}
\end{equation}
with the initial condition $z_{0}^{j}=(x_0,p_0)=z_{0}$, and $\lambda_j$ as in Figure \ref{image:examplePotentialMatrix}, precisely defined in Appendix~\ref{sec:equilibriumposobs}. 

\begin{figure}[htbp]
 \begin{subfigure}{0.49\textwidth}
    \includegraphics[width=0.9\linewidth]{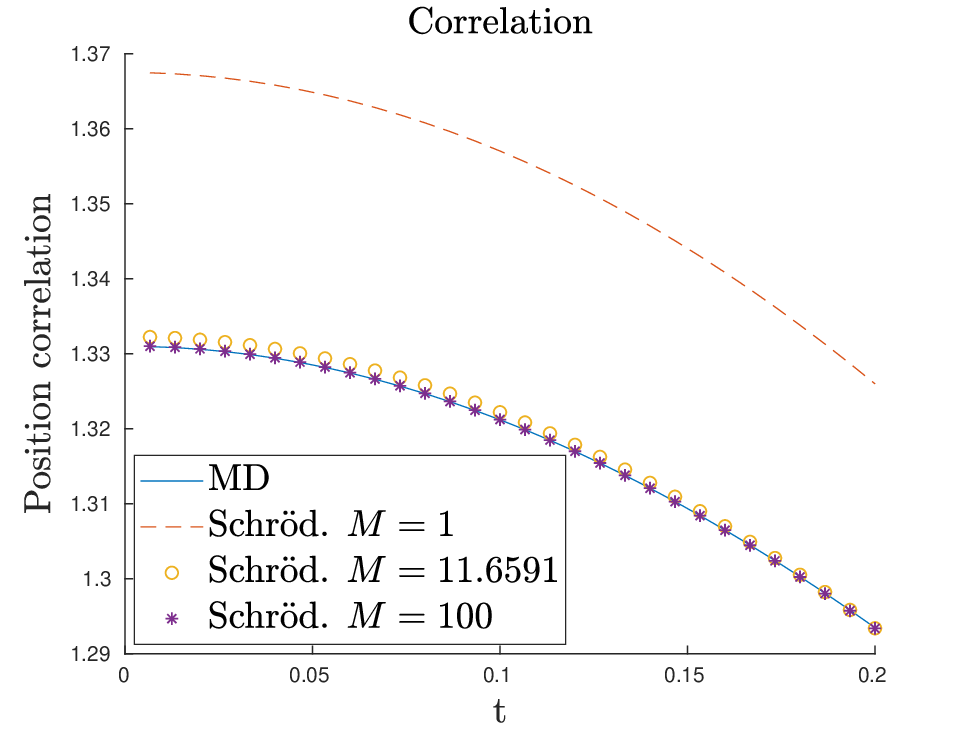}
    \caption{Correlation observable}
    \label{fig:corr3m}
 \end{subfigure}
  \begin{subfigure}{0.49\textwidth}
    \includegraphics[width=0.9\linewidth]{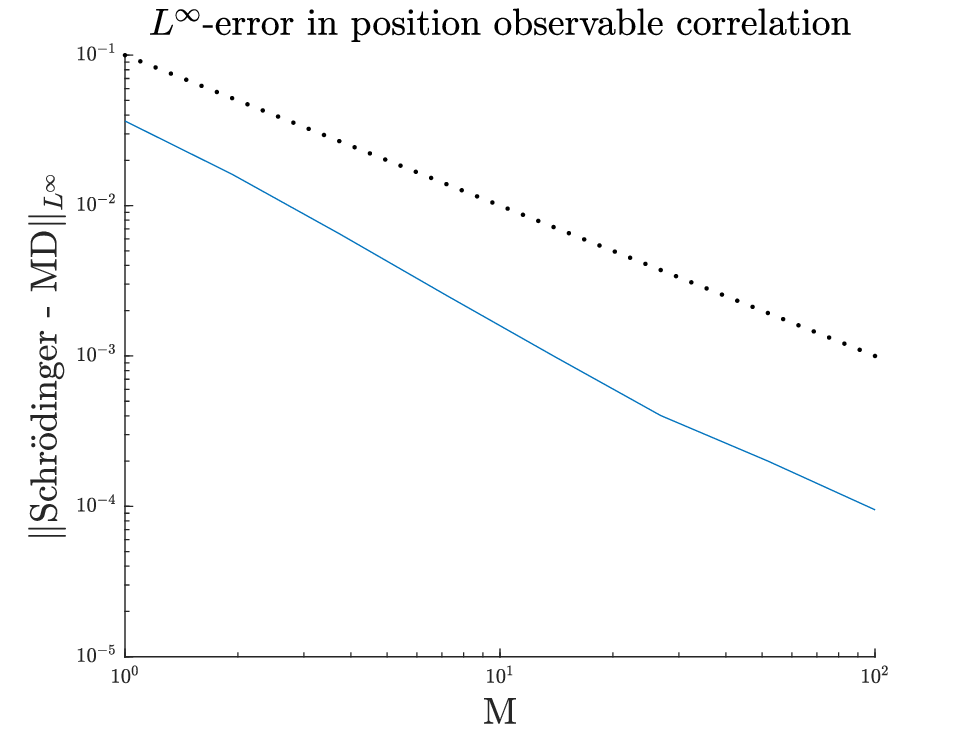}
   \caption{Error in correlation observable}
   \label{fig:corrError}
  \end{subfigure}
  \caption{(a) Molecular dynamics position correlation observable shown together with its Schr\"odinger counterpart. for three values of the mass ratio $M$. (b) The $L^{\infty}$-error in the molecular dynamics position correlation observable approximation for $\tau=0.2$, shown in log-log scale. The dotted line shows the slope $-1$ for reference. 
  }
  \label{correlation_figure}
\end{figure}
With the same values for $T$ and $\delta$ as in the case of equilibrium observables we show the position correlation observable for 
different correlation times $\tau$, and three different mass ratios $M$, in
Figure~\ref{fig:corr3m}. 

Figure~\ref{fig:corrError} demonstrates that the $L^{\infty}$-error is inverse proportional to the mass ratio $M$ in agreement with the error analysis of Section~\ref{sec_corr}.


%
%

\section{Time correlated observables}\label{sec_corr}
In this section we study canonical quantum observables for correlations in time, namely
\[
\TR( \OPER{A}_\tau\OPER{B}_0 {\EXP{-\beta \OPER{H}}})= 
\sum_{n=1}^\infty \langle \Phi_n, \OPER{A}_\tau\OPER{B}_0 {\EXP{-\beta \OPER{H}}}\Phi_n\rangle\COMMA
\]
and the  related variant 
$\TR\big( \OPER{A}_\tau(\OPER{B}_0 {\EXP{-\beta \OPER{H}}}+ {\EXP{-\beta \OPER{H}}}\OPER{B}_0)\big)$,
based on the time dependent operator $\OPER{A}_\tau$, which for $\tau\in \rset$ is defined by
\begin{equation}\label{a_t_def}
\OPER{A}_\tau:=\EXP{\IU\tau M^{1/2}\hat H}\OPER{ A}_0\EXP{-\IU\tau M^{1/2}\hat H}\COMMA
\end{equation}
with a  matrix valued symbol $A_0:\rset^{2N} \rightarrow \CSP$ in the Schwartz class.

\noindent{\it Example.}
For instance, the observable for the diffusion constant 
\[
\frac{1}{6\tau}\frac{3}{N}\sum_{k=1}^{N/3} |x_k(\tau)-x_k(0)|^2=\frac{1}{2N\tau}\big(
|x({\tau})|^2 +|x(0)|^2-2x(\tau)\cdot x(0)\big)
\]
uses the time-correlation $\hat x(\tau)\cdot \hat x(0)$ where  $\OPER{A}_\tau=\hat x_\tau
\rm I$ and $\OPER{B}_0=\hat x_0\rm I$ and
\[
\hat x_\tau\cdot \hat x_0=\sum_{n=1}^{N/3}\sum_{j=1}^3\EXP{\IU\tau M^{1/2}\hat H}\hat x_{n_j} \EXP{-\IU\tau M^{1/2}\hat H}  \hat x_{n_j}\PERIOD
\]
\medskip

To analyze the time evolution of $\OPER{A}_\tau$ we use transformed variables:
assume 
that $\Psi:\rset^{N}\rightarrow \CSP$ and $\Psi(x)$ is an orthogonal matrix  with the Hermitian transpose
 $\Psi^*(x)$ 
  and define
$ {\bA}:[0,\infty)\times\rset^{2N}\rightarrow \CSP$ 
by
\begin{equation}\label{tilde_a_t}
\OPERW{\bA(\tau,z)} = \hat \Psi^*(x) \OPER{ A}_\tau\hat\Psi(x)\COMMA\quad \tau\ge 0\PERIOD
\end{equation}
The matrix $\Psi(x)$ we will use is defined precisely below and it approximates the matrix $\tPsi(x)$ that diagonalizes the potential matrix $V(x)$ in the sense that $\|\Psi(x)-\tPsi(x)\|=\BIGO(M^{-1})$.
We also assume that %
\begin{equation*}\label{tildeAB}
\begin{split}
\OPER{B_0}&=\hat \Psi \OPERW{\bB_0}\hat\Psi^*\COMMA\\
\end{split}
\end{equation*}
and we restrict our study to the case where the $d\times d$ matrix  symbols $\bA(0,\cdot)=\bA_0$ and $\bB_0$ are diagonal, as motivated in \eqref{diag_A}.
Let $\alpha$ be any complex number and define for $t\in\mathbb R$ the exponential 
\[
\OPERW{\by_t}:=\hat \Psi^* \EXP{t\alpha\hat H}\hat\Psi.
\]
Differentiation shows that
\begin{equation}\label{eq:y_t_diffeq}
\partial_t \OPERW{\by_t} = \alpha \hat \Psi^*\hat H\hat\Psi \hat \Psi^* \EXP{t\alpha\hat H}\hat\Psi
= \alpha \hat \Psi^*\hat H\hat\Psi\OPERW{\by_t}
\end{equation}
and we conclude that 
\begin{equation}\label{y_t_exp}
\OPERW{\by_t} =\EXP{t\alpha \hat \Psi^*\hat H\hat\Psi}\PERIOD
\end{equation}
The composition rule 
\begin{equation}\label{comp_rule}
\begin{split}
{A\MP B}(x,p) &:= \EXP{\frac{\IU}{2M^{1/2}} (\nabla_{x'}\cdot \nabla_{p} -\nabla_{x}\cdot \nabla_{p'})}A(x,p)B(x',p')\Big|_{(x,p)=(x',p')}\COMMA\\
\OPER{A}\OPER{B} &=\OPERW{A\MP B} 
\end{split}
\end{equation}
 see Theorems 4.11 and 4.12 in \cite{zworski}, implies
\begin{equation}\label{psi_H_psi}
\hat \Psi^*\hat H\hat\Psi= \WEYL{\Psi^*\MP  H\MP \Psi}\PERIOD
\end{equation}
A power expansion of the exponential  in \eqref{comp_rule}, see \cite{zworski}, yields the semiclassical expansion
\[\sum_{n=0}^m\frac{1}{n!} 
(\IU\frac{\nabla_{x'}\cdot \nabla_{p} -\nabla_{x}\cdot \nabla_{p'}}{2M^{1/2}})^n A(x,p)B(x',p')\Big|_{(x,p)=(x',p')} + \mathcal O(M^{-(m+1)})\, .\] This expansion is defined for symbols in the Schwartz class and Lemma \ref{moyal_lemma} provides an extension to a larger set of symbols. 
Due to the special from of the Hamiltonian studied here, namely $H(x,p)=\frac{|p|^2}{2}\Id + V(x)$, all 
terms of the semiclassical expansion for \eqref{psi_H_psi}  with higher powers than $M^{-1}$ drop out and we obtain a simple sum of only two terms:
\begin{lemma}\label{H_bar} Any  orthogonal two times differentiable matrix valued mapping 
$\Psi:\mathbb R^N\to \CSP$ satisfies
\[
 \Psi^*\MP  H\MP \Psi(x,p)
=\Psi^*(x)H(x,p)\Psi(x) + \frac{1}{4M}\nabla\Psi^*(x)\cdot\nabla\Psi(x)\PERIOD
\]
\end{lemma}
\begin{proof}
The composition in \eqref{comp_rule} shows that
 \[
 \begin{split}
& \Psi^*\MP  H\MP \Psi(x,p)\\
 &=\Psi^*(x)\MP  \big(H(x,p)\Psi(x) +\frac{\IU M^{-1/2}}{2} p\cdot \nabla\Psi(x) -\frac{M^{-1}}{4}\Delta\Psi(x)\big)\\
 &=\Psi^*(x)H(x,p)\Psi(x)+\frac{\IU M^{-1/2}}{2}p\cdot \nabla\Psi^*(x)\Psi(x)-\frac{M^{-1}}{4}\Delta\Psi^*(x)\Psi(x)\\
&\quad +\frac{\IU}{2M^{1/2}} \Psi^*(x) p\cdot \nabla\Psi(x)
 -\frac{1}{4M}\nabla\Psi^*(x)\cdot\nabla\Psi(x)-\frac{1}{4M}\Psi^*(x)\Delta\Psi(x)\\
  \end{split}
 \]
 which by the property $\Psi^*\Psi=\Id$ can be written
 \[
 \begin{split}
 \Psi^*\MP  H\MP \Psi(x,p)&=\Psi^*(x)H(x,p)\Psi(x)+\frac{\IU}{2M^{1/2}}p\cdot\nabla\big(\Psi^*(x)\Psi(x)\big)\\
&\qquad  -\frac{1}{4M}\Big(\Delta \big(\Psi^*(x)\Psi(x)\big) -\nabla\Psi^*(x)\cdot\nabla\Psi(x)\Big)\\
  &=\Psi^*(x)H(x,p)\Psi(x) + \frac{1}{4M}\nabla\Psi^*(x)\cdot\nabla\Psi(x)\PERIOD
 \end{split}
 \]
\end{proof}

Define
\[
\bH(x,p):= \Psi^*\MP  H\MP \Psi=\Psi^*(x)H(x,p)\Psi(x) + \frac{1}{4M}\nabla\Psi^*(x)\cdot\nabla\Psi(x)\COMMA
\]
which by \eqref{eq:y_t_diffeq} and \eqref{y_t_exp} implies $\hat\Psi^*\EXP{t\alpha\hat H}\hat\Psi=\EXP{t\alpha\OPER{\bH}}$
and we obtain by \eqref{a_t_def} and \eqref{tilde_a_t} that
\[
\partial_t\OPER{\bA}_t=\IU  M^{1/2}[\OPER{\bH},\OPER{\bA}_t]\PERIOD
\]
The next step is to determine $\Psi$ so that
$
\bH=\Psi^*\MP H\MP \Psi
$
is diagonal (or approximately diagonal in \eqref{diag_comp} ), in order to make $\bH\bA_t- \bA_t\bH$ small since it appears 
in the expansion of the compositions in 
\[
\partial_t {\bA}_t=\IU M^{1/2}( \bH\MP  \bA_t- {\bA}_t\MP  \bar  H)\PERIOD
\]
To have $\bH\bA_t- \bA_t\bH$ small then also requires $\bA_t$ to be diagonal (or almost diagonal).
In the case when $\bH$ is diagonal, the quantization  $\OPER{\bH}$ is diagonal
and then $\OPER{\bA}$ remains diagonal if it initially is diagonal, since then
\begin{equation}\label{diag_A}
\begin{split}
\frac{{\mathrm d}}{{\mathrm d}t} \OPER{\bA}_{jk}(t)
& = \IU M^{1/2}\big(\bH_{jj}\OPER{\bA}_{jk}(t) -
\OPER{\bA}_{jk}(t)\bH_{kk}\big)=0\COMMA \mbox{ for } j\ne k\PERIOD\\
\end{split}
\end{equation}
Consequently we restrict our study to observables where $\bA_0=\bA(0)$ are diagonal.

We have
\[
\bH(x,p)= \Psi^*(x)\Big(\frac{|p|^2}{2} \Id + V(x) + \frac{1}{4M} \Psi(x)\nabla\Psi^*(x)
\cdot \nabla\Psi(x)\Psi^*(x)\Big)\Psi(x)\PERIOD
\]
Therefore the aim is to choose the orthogonal matrix $\Psi$  so that it is a solution or an approximate solution to the non linear eigenvalue problem  
\begin{equation}\label{eq:nonlineareigenvalueproblem}
\big(V + \frac{1}{4M} \Psi\nabla\Psi^*\cdot \nabla\Psi\Psi^*\big)\Psi= \Psi\bLambda
\end{equation}
where $\bLambda$ is diagonal with $\bLambda_{jj}=:\lambda_j$. Such a transformation $\Psi$ is an $\BIGO(M^{-1})$ perturbation of the eigenvectors to $V(x)$,
provided the eigenvalues $\tlambda_j(x)$ of $V(x)$ do not cross and $M$ is sufficiently large. Section \ref{perturb}  shows that \eqref{eq:nonlineareigenvalueproblem} 
has an approximate solution to any accuracy
\begin{equation}\label{eigen_nonlin2}
\big(V + \frac{1}{4M} \Psi\nabla\Psi^*\cdot \nabla\Psi\Psi^*\big)\Psi= \Psi\Lambda +\mathcal O(M^{-\kappa})\COMMA
\end{equation}
with $\Lambda$ diagonal and any $\kappa>0$.   
\subsection{Approximate solution of the nonlinear eigenvalue problem}\label{perturb} 

The following fixed point iterations yields an approximate solution \eqref{eigen_nonlin2} to the nonlinear eigenvalue problem \eqref{eq:nonlineareigenvalueproblem}.
 Let $\mathcal S(C)$ denote an orthogonal matrix of eigenvectors to a $d\times d$ Hermitian matrix $C$, with the columns in the order of the eigenvalues, so that e.g. $\mathcal S(V(x))=\left[\tPsi_1(x) \ \tPsi_2(x) \ \ldots \ \tPsi_d(x)\right]$ as in \eqref{v_lambda_psi}. Let $\Psi[1]=\mathcal S(V(x))$ and define
\begin{equation}\label{psi_kappa}
\Psi[j+1]=\mathcal S\big((V + \frac{1}{4M} \Psi[j]\nabla\Psi^*[j]\cdot \nabla\Psi[j]\Psi^*[j])(x)\big)\PERIOD
\end{equation}
Assume that the eigenvalues $\tlambda_j$ of $V$ are distinct and $V\in [\mathcal C^m(\rset^{N})]^{d^2}$. 
The regular perturbation theory of real symmetric matrices in \cite{kato}
shows that for sufficiently large $M$  and any $k\le m$ %
\begin{equation}\label{induct1}
\max_{|\gamma|\le k-1} \|\partial^\gamma(\Psi[2]-\Psi[1])\|_{L^\infty}=\BIGO_k(M^{-1})\COMMA\end{equation}
where $\BIGO_k$ denotes an order relation that is allowed to depend on $k$.
Then induction in $j$, for $j\le k$, shows that
\begin{equation}\label{induction}
\max_{|\gamma|\le k-j} \|\partial^\gamma(\Psi[j+1]-\Psi[j])\|_{L^\infty}\le C_{k,j} M^{-j}\COMMA
\end{equation}
as follows: we have 
\[
\begin{split}
&V + \frac{1}{4M} \Psi[j+1]\nabla\Psi^*[j+1]\cdot \nabla\Psi[j+1]\Psi^*[j+1] \\
&= 
V + \frac{1}{4M} \Psi[j]\nabla\Psi^*[j]\cdot \nabla\Psi[j]\Psi^*[j] + \BIGO_{1}(M^{-(j+1)}) 
\end{split}
\]
so that regular perturbation theory implies that the left hand side is diagonalized by
\begin{equation}\label{induct2}
\Psi[j+2]=\Psi[j+1] +\BIGO(M^{-(j+1)})
\end{equation}
and there is a constant $K_{k,j}$ such that
\[C_{k,j+1}\le K_{k,j} C_{k,j}\PERIOD\]

The choice $\Psi=\Psi[\kappa]$  implies
\begin{equation}\label{diag_comp}
\begin{split}
&\bH(x,p) =  \Psi^*[\kappa]\MP  H\MP  \Psi[\kappa]\\
&=
\underbrace{\Psi^*[\kappa]\big(\frac{|p|^2}{2} \Id + V + \frac{1}{4M} \Psi[\kappa-1]\nabla\Psi^*[\kappa-1]\cdot \nabla\Psi[\kappa-1]\Psi^*[\kappa-1]\big)\Psi[\kappa]}_{=:\frac{|p|^2}{2} \Id  +  \bLambda(x)}\\
&\quad +
\frac{1}{4M}\Psi^*[\kappa]\Big(
  \Psi[\kappa]\nabla\Psi^*[\kappa]\cdot \nabla\Psi[\kappa]\Psi^*[\kappa]\\
&\qquad  -\Psi[\kappa-1]\nabla\Psi^*[  \kappa-1]\cdot \nabla\Psi[\kappa-1]\Psi^*[\kappa-1]
\Big)\Psi[\kappa]\\
&=:\underbrace{\frac{|p|^2}{2} \Id  +  \bLambda(x)}_{=:\bH_0} + \underbrace{r_0}_{\BIGO(M^{-\kappa})}\COMMA
\end{split}
\end{equation}
 where \[
 \bLambda(x)=\tLambda(x)+\BIGO(M^{-1})\] is diagonal with $\tLambda_{jj}=\tlambda_j$
 as in \eqref{v_lambda_psi} and $r_0$ is the term with the factor $\frac{1}{4M}\Psi^*[\kappa](\ldots )\Psi[\kappa]$,
 which only depends on the $x$-coordinate.
 Here $\kappa\le m$, where $V\in [\mathcal C^m(\rset^{N})]^{d^2}$ and we remind that
 \[
 \bH_0 = \frac{|p|^2}{2} \Id  +  \bLambda(x)=\underbrace{\frac{|p|^2}{2} \Id  +  \tLambda(x)}_{=\tH} + \BIGO(M^{-1}) .
 \]
 We see that
\[
\partial_t\OPER{\bA}_t  = \IU M^{1/2}[\OPER{\bH}, \OPER{\bA}_t]
= \IU M^{1/2}[\OPER{\bH}_0, \OPER{\bA}_t] + \IU M^{1/2}[\OPER{ r}_0, \OPER{\bA}_t]
\]
consists of a diagonal part and a small coupling $\BIGO(M^{-\kappa+1/2})$ part.
This asymptotic recursion for $\Psi[j]$ is typically not convergent,
therefore the error term $r_0=\BIGO(M^{-\kappa})$ may be large if $\kappa$ is large,
unless $M$ is very large, which is a reason to avoid  large values of $\kappa$.
\subsection{Approximation in time of observables}\label{classical_observ}
 The transformation to diagonal $\bH$ and $\bA_0$ yields restrictions to the set of observables $A_0$
 that we can analyse. 
 The aim of this section is to describe  these restrictions and 
 present the time dependent molecular dynamics observable that will be used to approximate $\bA_t$.
 
 \subsubsection{Time dependent molecular dynamics} The symbol $\brA:[0,\infty)\times\rset^{2N}\rightarrow \CSP$
that satisfies
 \begin{equation}\label{poissonbracket}
 \partial_t \brA_t=\{\bH_0,\brA_t\}\COMMA\quad t>0,\ \brA_0=\bA_0\COMMA
 \end{equation}
 approximates $\bA_t$, as we shall see in Lemma \ref{lemma12} below.
By writing the Poisson bracket in the right hand side explicitly,  we see that 
equation \eqref{poissonbracket} is a scalar linear hyperbolic partial differential equation for each component:
 \begin{equation}\label{hyp_pde}
 \partial_t \brA_{jj}(t,x,p) = (\nabla_p\cdot \nabla_{x'}-\nabla_x\cdot \nabla_{p'})\big((\bH_0)_{jj}(x,p)\brA_{jj}(t,x',p')\big)\big|_{(x',p')=(x,p)} \PERIOD
 \end{equation}
 This partial differential equation can be solved by the method of characteristics,
 which generates molecular dynamics paths as follows.
Let $\brA_{jj}$ be constant along  the characteristic
\begin{equation}\label{a_char}
\brA_{jj}(t,z_0) := \brA_{jj}\big(0,z^j_t(z_0)\big)
\end{equation}
where the characteristic path $z^j_t=(x_t,p_t)$ solves the Hamiltonian system
\begin{equation}\label{HS}
\begin{split}
\dot x_t &= p_t\\
\dot p_t &= -\nabla \lambda_j(x_t), \quad t>0
\end{split}
\end{equation}
with initial data $(x_0,p_0)=z_0$ and the Hamiltonian
$(\bH_0(x,p))_{jj}=(\frac{|p|^2}{2}\Id + \bLambda(x))_{jj}=\frac{|p|^2}{2}+\lambda_j(x)$. For each $j$  we have
\begin{equation}\label{A_brav_char}
\underbrace{\partial_t\brA_{jj}\big(0,z^j_t(z_0)\big)}_{=\partial_s\brA_{jj}(0,z^j_t(z_s))\big|_{s=0}}  =\{(\bH_0)_{jj}(z_0),\brA_{jj}(0,z^j_t(z_0))\} \COMMA
\end{equation}
where the equality in the left hand side holds
because the Hamiltonian system is autonomous
and the Poisson bracket in the right hand side is obtained from the chain rule differentiation at $s=0$. We conclude that \eqref{poissonbracket} holds for $\brA$ constructed by \eqref{a_char}. %

 
\subsubsection{The set of allowed observables}\label{A_set}

Our approximation of canonical quantum observables  becomes implicit in the following sense. Given a Hamiltonian symbol $H$ we can determine electron states $\Psi$ so that $\bH=\Psi^*\MP  H\MP  \Psi$ is diagonal.
 For any diagonal initial symbols $\bA_0$ and $\bB_0$, we will for instance show that 
 the molecular dynamics observable 
 \[\int_{\rset^{2N}} \TR(\brA_\tau(z) \bB_0(z) \EXP{-\beta \bH(z)}) \Rd z\]
 approximates the quantum observable $\TR (\OPER{A}_\tau \OPER{B}_0 \EXP{-\beta \OPER{H}})$, where 
 \[A_0=\Psi\MP  \bA_0\MP \Psi^* \ \mbox{ and } \ B_0=\Psi\MP  \bB_0\MP \Psi^*\, .\]  To approximate a given quantum observable
  \[\TR( \OPER{A}_\tau \OPER{B}_0 \EXP{-\beta\OPER{H}})\] is in this formulation possible only if
$ \bA_0 = \Psi^*\MP  A_0\MP \Psi$ and $\bB_0 = \Psi^*\MP  B_0\MP \Psi$
are diagonal (or almost diagonal). 
From a given diagonal $\bA_0$ it is therefore direct to determine $A_0$
but the opposite to first choose $A_0$  then requires to verify if $\bA_0$ is diagonal.
The set of such $A_0$ is in the special case where $\bA_0$ only depends on $x$
given by $A_0= \Psi \bA_0 \Psi^*$, with $\bA_0$ any diagonal matrix. If in addition $\bA_0(x)=a(x)\Id$, for any scalar function $a:\rset^N\to \rset$, we obtain $A_0(x)= \bA_0(x)=a(x)\Id$ and similarly for $B_0$.

We note that if eigenvalue surfaces cross, i.e. if 
$\lambda_j(x)=\lambda_{j+1}(x)$ for some $j$ and $x$,
then $\nabla\Psi(x)$ may not be in $L^2(\rset^{N})$.
We have assumed that the observable symbols $\bA_0$ and $\bB_0$ are diagonal
in the same coordinate transformation $\Psi[\kappa]$ that approximately
diagonalizes  the Hamiltonian, in the composition way \eqref{diag_comp}.
Example of observables that cannot be diagonalized by the same transformation as
the Hamiltonian are $A_0(z)=x\Psi_{\cdot 1}(x)(\Psi_{\cdot 2})^*(x)$ and $B_0(z)=x\Psi_{\cdot 1}(x)(\Psi_{\cdot 2})^*(x)$
and the correlation based on these observables are then not applicable in Theorems \ref{gibbs_corr_thm_analytic}, \ref{gibbs_corr_thm} and \ref{gibbs_corr_thm_unif}, in contrast to the projections $\Psi_{\cdot j}(\Psi_{\cdot j})^*$ in Remark \ref{projections}.
 
\subsection{Assumptions and theorems}
This section states our main results on molecular dynamics approximation of quantum observables in the canonical ensemble. The molecular dynamics is based on the Hamiltonian system formulated in (\ref{a_char}--\ref{A_brav_char}).
The observable $\bA_{jj}(0,z^j_t (z_0))$, 
for each $j\in\{1,\ldots, d\}$, 
is constant along the molecular dynamics path and provides the classical approximation of the corresponding quantum observable as we shall see
in Theorems \ref{gibbs_corr_thm_analytic}, \ref{gibbs_corr_thm} and \ref{gibbs_corr_thm_unif}. 
By assuming that the potential $V\in [\mathcal C^3(\rset^N)]^{d^2}$, the approximate nonlinear eigenvalue problem \eqref{eigen_nonlin2} has an approximate  solution
matrix $\Psi$ of eigenvectors with eigenvalues $\Lambda_{jj}=\lambda_j$
that yields an almost  diagonal  Hamiltonian $\bH=\bH_0+ r_0 $ with small remainder $r_0$, as defined in \eqref{diag_comp}. 
The observables $\OPER{A}_\tau=\OPER{\Psi}\OPER{\bA}_\tau\OPER{\Psi}^*$ and $\OPER{B}=\OPER{\Psi}\OPER{\bB}_0
\OPER{\Psi}^*$ 
satisfy

\begin{theorem}\label{gibbs_corr_thm_analytic}
 Assume that $V\in \mathcal [C^{3}(\rset^N)]^{d^2}$,  $\bA_0$ and $\bB_0$ are diagonal, the $d\times d$ matrix valued Hamiltonian $H$ has distinct eigenvalues, and that there is a constant $C$ such that
\begin{equation*}\label{R77}
\begin{split}
\sum_{|\alpha|\le 2}\|\partial^\alpha_x \psi_k\|_{L^\infty(\rset^{3N})} &\le C\COMMA\quad k=1,\ldots,d\COMMA\\
\max_i\sum_{|\alpha|\le 3}\|\partial^\alpha_x \partial_{x_i}\lambda_j\|_{L^\infty(\rset^{N})} &\le C\COMMA\\
\sum_{|\alpha|\le 3}
\|\partial_z^\alpha\bA_{jj}(0,\cdot)\|_{L^2(\rset^{2N})}  &\le C\COMMA\\
\|\EXP{-\beta \bH(z)}\|_{L^2(\rset^{2N})} +\|\bB_0(z)\EXP{-\beta \bH(z)}\|_{L^2(\rset^{2N})} &\le C\COMMA\\
\end{split}
\end{equation*} 
hold, then there is a constant $c$, depending on $C$, such that the canonical ensemble average satisfies
\begin{equation*}\label{G_corr_unif1}
\begin{split}
&\Big|\frac{\TR\big(\OPER{A}_\tau \hat\Psi{ \WEYL{{\bB}_0 \EXP{-\beta \bH}}} \hat\Psi^*\big)}{
\TR(\hat\Psi\OPERW{\EXP{-\beta \bH}}\hat\Psi^*)}
 - \sum_{j=1}^d \int_{\rset^{2N}} 
\frac{\bA_{jj}(0,z^j_\tau (z_0))\bB_{jj}(z_0) \EXP{-\beta \bar{H}_{jj}(z_0)}}{\sum_{k=1}^d\int_{\rset^{2N}}\EXP{-\beta \bar{H}_{kk}(z)} \Rd z} \Rd z_0\Big|\\
&\le  cM^{-1}\COMMA
\end{split}
\end{equation*}
where $z^j_\tau$ solves \eqref{HS}.
%
\end{theorem}

We note that
\[
\frac{\TR\big(\OPER{A}_\tau \hat\Psi{ \WEYL{{\bB}_0 \EXP{-\beta \bH}}} \hat\Psi^*\big)}{
\TR(\hat\Psi\OPERW{\EXP{-\beta \bH}}\hat\Psi^*)}
=\frac{\TR\big(\OPER{\bA}_\tau { \WEYL{{\bB}_0\EXP{-\beta \bH}}} \big)}{
\TR(\OPERW{\EXP{-\beta \bH}})}\, .
\]
Here we have used the quantization
of the classical density, $\OPERW{\EXP{-\beta\bar H}}$, as discussed in Section \ref{sec:gibbs_intro}. 
Our related results
for the density $\EXP{-\beta\hat H}$
require similar but more assumptions
on the regularity of the data given below.
We note that all assumptions are based on a fixed number of derivatives not depending on $N$.

\begin{assumption}\label{lem3.4} Assume $\kappa=1$ or $\kappa=2$ in \eqref{diag_comp} defines the remainder $r_0$, the eigenvalues $\Lambda$, the Hamiltonian $\bar H_0$ and that
there is a constant $C$ such that
\begin{equation}\label{R0a}
\|r_0\|_{L^\infty(\rset^{N})} \le C M^{-\kappa}\COMMA\\
\end{equation}
\begin{equation}\label{R1a}
\begin{split}
&\||\nabla\bLambda|^2\|_{L^\infty(\rset^{N})}+\|\Delta\bLambda\|_{L^\infty(\rset^{N})}+\max_{j,k}\|\partial_{x_j}\partial_{x_k}\bLambda\|_{L^\infty(\rset^{N})} \le C\COMMA\\
& \sum_{jk}\sup_{t\in [0,1/T]}
\|\EXP{-t\bH_0}{p_j}\partial_{p_k}\bB_0\|_{L^2(\rset^{2N})} \\
&\quad +\sum_{jk}\sup_{t\in [0,1/T]}
\|\EXP{-t\bH_0}\partial_{p_j}\partial_{p_k}\bB_0\|_{L^2(\rset^{2N})} 
\le C\COMMA\\
& \sum_{jk}\sup_{t\in [0,1/T]}
\|\EXP{-t\bH_0}{p_j}{p_k}\bB_0\|_{L^2(\rset^{2N})} 
\le C\COMMA \\
&\sup_{t\in [0,1/T]}
\|\EXP{-t\bH_0}\Delta_x\bB_0\|_{L^2(\rset^{2N})}
+\sup_{t\in [0,1/T]}
\|\EXP{-t\bH_0}\nabla\bLambda\cdot\nabla_x\bB_0\|_{L^2(\rset^{2N})}
\le C\COMMA\\
&\sup_{t\in [0,1/T]}
\|\EXP{-t\bH_0}\bB_0\|_{L^2(\rset^{2N})}
\le C\COMMA\\
&\sup_{t\in [0,1/T]}\|\EXP{-t\bH_0}(p\cdot\nabla_x\bB_0
-\nabla\bLambda\cdot \nabla_p\bB_0)\|_{L^2(\rset^{2N})} \le C \COMMA\\
\end{split}
\end{equation}
\begin{equation}\label{R2a}
 \|\bA_0\|_{L^2(\rset^{2N})}\le C\COMMA
 \end{equation}
 \begin{equation*}\label{R100}
\sup_{t\in [0,1/T]}
\|\EXP{-t\bH_0}\|_{L^2(\rset^{2N})}\le C\PERIOD
\end{equation*}
\end{assumption}

\begin{assumption}\label{lem3.5}
Assume there is a constant $C$ such that
 \begin{equation*}\label{R1000}
\sup_{t\in [0,1/T]}
\|\EXP{-t\bH_0}\|_{L^2(\rset^{2N})}\le C\COMMA
\end{equation*}

\begin{equation}\label{R7a}
\begin{split}
\max_i\sum_{|\alpha|\le 3}\|\partial^\alpha_x \partial_{x_i}\lambda_j\|_{L^\infty(\rset^{N})} &\le C\COMMA\\
\sum_{|\alpha|\le 3}
\|\partial_z^\alpha\bA_{jj}(0,\cdot)\|_{L^2(\rset^{2N})}  &\le C\COMMA\\
\end{split}
\end{equation}

\begin{equation*}\label{R60}
\begin{split}
\|\bB_0(z)\EXP{-\beta \bH_0(z)}\|_{L^2(\rset^{2N})} &\le C\PERIOD\\
\end{split}
\end{equation*}
\end{assumption}

\begin{assumption}\label{lem3.6}
Assume $\kappa=2$ 
 in \eqref{diag_comp} defines the remainder $r_0$, the eigenvalues $\Lambda$, the Hamiltonian $\bar H_0$ and that there is a constant $C$ such that
\begin{equation*}\label{R00}
\|r_0\|_{L^\infty(\rset^{N})} \le CM^{-2}\COMMA\\
\end{equation*}
\begin{equation}\label{R8a}
\begin{split}
\frac{1}{(2\pi)^{2N}}\|\FT (\partial_{z_n} \EXP{-\bH_0})\|_{L^1(\rset^{2N})}&\le C\COMMA\\
\sum_n\|\partial_{z_n}(p\cdot\bB_0-\nabla\bLambda\cdot \bB_0)\|_{L^2(\rset^{2N})}&\le C\COMMA\\
\end{split}
\end{equation}

\begin{equation}\label{R9a}
\begin{split}
\frac{1}{(2\pi)^{2N}}\|\FT (p\cdot\bB_0-\nabla\bLambda\cdot \bB_0)\|_{L^1(\rset^{2N})}\le C\, .\\
\end{split}
\end{equation}
\end{assumption}

Based on these assumptions we can formulate
two additional  theorems
for the observables $\OPER{A}_\tau=\OPER{\Psi}\OPER{\bA}_\tau\OPER{\Psi}^*$ and $\OPER{B}=\OPER{\Psi}\OPER{\bB}_0
\OPER{\Psi}^*$, with $\Psi=\Psi[\kappa], \ \kappa\in\{1,2\}$, defined in \eqref{psi_kappa} and the diagonal Hamiltonian $\tH=\bH_0 + \BIGO(M^{-1})$ in \eqref{tilde_H}. %
\begin{theorem}\label{gibbs_corr_thm}
 Assume that $\bA_0$ and $\bB_0$ are diagonal, the $d\times d$ matrix valued Hamiltonian $H$ has distinct eigenvalues, 
 %
 and that 
 Assumptions \ref{lem3.4} and \ref{lem3.5} hold,
then the canonical ensemble average satisfies
\begin{equation}\label{G_corr}
\begin{split}
&\frac{\frac{1}{2}\TR\big(\OPER{ A}_\tau ({\OPER{B}_0 {\EXP{-\beta\OPER{H}}} 
+{\EXP{-\beta \OPER{H}}}\OPER{B}_0})\big)}{\TR({\EXP{-\beta \OPER{H}}})}\\
&=  \sum_{j=1}^d \int_{\rset^{2N}} q_j\bA_{jj}(0,z^j_\tau (z_0))\bB_{jj}(z_0) 
\frac{\EXP{-\beta \tH_{jj}(z_0)}}{\int_{\rset^{2N}}\EXP{-\beta \tH_{jj}(z)} \Rd z} \Rd z_0
+\BIGO(M^{-1/2})\COMMA
\end{split}
\end{equation}
where $z^j_\tau=(x_\tau,p_\tau)$ solves 
\begin{equation*}\label{HS*}
\begin{split}
\dot x_t &= p_t\\
\dot p_t &= -\nabla\tilde {\lambda}_j(x_t), \quad t>0,
\end{split}
\end{equation*}
and $q_j=q_j(\tH)$ is defined by \eqref{qh_def}.
If in addition Assumption \ref{lem3.6} holds then the estimate in \eqref{G_corr} holds with the more accurate bound $\BIGO(M^{-1})$ replacing $\BIGO(M^{-1/2})$.
\end{theorem}

By comparing instead to the Weyl quantized classical Gibbs density
$\OPERW{\EXP{-\beta \bH_0}}$ we have the following  more accurate error estimate, that only requires Assumption \ref{lem3.5}.

\begin{theorem}\label{gibbs_corr_thm_unif}
 Assume that $\bA_0$ and $\bB_0$ are diagonal, the $d\times d$ matrix valued Hamiltonian $H$ has distinct eigenvalues, that %
 $ \Psi=\Psi[2]$ and that the Assumption \ref{lem3.5},\eqref{R0a} and \eqref{R2a}
 hold,
then the canonical ensemble average satisfies
\begin{equation}\label{G_corr_unif}
\begin{split}
&\frac{\TR\big(\OPER{ A}_\tau \hat\Psi \WEYL{\bB_0 \EXP{-\beta \bH_0}}  \hat\Psi^*\big)}{
\TR(\hat\Psi\OPERW{\EXP{-\beta {\bH}_0}}\hat\Psi^*)}
=\frac{\TR\big(\OPER{ \bA}_\tau { \WEYL{{\bB}_0\EXP{-\beta {\bH}_0}}} \big)}{
\TR(\OPERW{\EXP{-\beta {\bH}_0}})}\\
&=  \sum_{j=1}^d \int_{\rset^{2N}} q_j\bA_{jj}(0,z^j_\tau (z_0))\bB_{jj}(z_0) 
\frac{\EXP{-\beta ({\bar{H}_0})_{jj}(z_0)}}{\int_{\rset^{2N}}\EXP{-\beta ({\bar{H}_0})_{jj}(z)} \Rd z} \Rd z_0
+\BIGO(M^{-1})\COMMA
\end{split}
\end{equation}
where $z^j_\tau$ solves \eqref{HS} and  $q_j=q_j(\bH_0)$ is defined by \eqref{qh_def}.
%
\end{theorem}

\subsection{Structure of the proofs}
It is useful to  split our estimation into two parts
 \begin{equation}\label{steps}
 \begin{split}
&\TR\big(\OPER{\bA}_\tau \frac{1}{2}(\OPER{\bB}_0\EXP{-\beta\OPER{\bH}} +  \EXP{-\beta\OPER{\bH}}\OPER{\bB}_0) 
- \OPERW{\brA(\tau,z)} \OPERW{\bB_0\EXP{-\beta H_0}} \big)\\
&= 
 \frac{1}{2}\TR\Big(\OPER{\bA}_\tau\big(\OPER{\bB}_0\EXP{-\beta\OPER{\bH}} +  \EXP{-\beta\OPER{\bH}}\OPER{\bB}_0 
-(\bB_0\EXP{-\beta H_0})^{\OPERW{ }}  - (\EXP{-\beta H_0}{\bB}_0)^{\OPERW{ }} \ \big)\Big)\\
&\quad +
\TR\big((\OPER{\bA}_\tau- \OPERW{\brA(\tau,z)})(\bB_0\EXP{-\beta H_0})^{\OPERW{ }} \ \big)\\
\end{split}
\end{equation}
where the first part is the approximation error of the Gibbs density operator, which is estimated  in Lemma \ref{lemma1}, 
and  the second part is the approximation error of the dynamics of the observable, which is estimated in Lemma \ref{lemma12}. 

Theorems \ref{gibbs_corr_thm_unif}  and \ref{gibbs_corr_thm_analytic} use only Lemma \ref{lemma12}, while
Theorem \ref{gibbs_corr_thm} uses both Lemma \ref{lemma1} and Lemma \ref{lemma12}.  A third Lemma \ref{lemma13} improves the $\BIGO(M^{-1/2})$ bound in Lemma \ref{lemma1} to $\BIGO(M^{-1})$ under additional assumptions.

\begin{lemma}\label{lemma1}
Assume that  the $d\times d$ matrix symbols $\bA_0(z)$ and $\bB_0(z)$ are diagonal
and
the bounds in Assumption \ref{lem3.4}
hold, then
\begin{equation*}\label{lemma1_eq}
\begin{split}
&\frac{\TR\Big(\OPER{\bA}_\tau\big(\OPER{\bB}_0\EXP{-\beta\OPER{\bH}} 
-(\bB_0\EXP{-\beta H_0})^{\OPERW{ }}  \ \big)\Big)}{
\TR (\EXP{-\beta \OPER{\bH}})}=\BIGO(M^{-\min(1/2,\kappa-1/2)})\PERIOD
\end{split}
%
\end{equation*}
\end{lemma}
\begin{lemma}\label{lemma12}
Assume that  the $d\times d$ matrix symbols $\bA_0(z)$ and $\bB_0(z)$ are diagonal
and
the bounds in Assumption \ref{lem3.5} 
hold, then
\[
\frac{\TR\big((\OPER{\bA}_\tau- \OPERW{\brA(\tau,z)})
(\bB_0\EXP{-\beta H_0})^{\OPERW{ }} \ \big)}{\TR (\EXP{-\beta{\OPER{\bH}}_0})}
=\BIGO(M^{-1})\PERIOD
\]
\end{lemma}
\begin{lemma}\label{lemma13}
Assume that  the $d\times d$ matrix symbols $\bA_0(z)$ and $\bB_0(z)$ are diagonal
and
the bounds in Assumtion \ref{lem3.6} 
hold, then
\begin{equation}\label{lemma13_eq11}
\begin{split}
&\frac{\TR\Big(\OPER{\bA}_\tau\big(\OPER{\bB}_0\EXP{-\beta\OPER{\bH}} +  \EXP{-\beta\OPER{\bH}}\OPER{\bB}_0 
-(\bB_0\EXP{-\beta H_0})^{\OPERW{ }}  - (\EXP{-\beta H_0}{\bB}_0)^{\OPERW{ }} \ \big)\Big)}{
\TR (\EXP{-\beta \OPER{\bH}})}\\
&=\BIGO(M^{-\min(1,\kappa-1/2)})\PERIOD
\end{split}
\end{equation}
\end{lemma}

The results in Lemmas \ref{lemma1} -- \ref{lemma13} have clear  limitations since the error estimate
of the approximation of the Gibbs density operator in Lemmas \ref{lemma1} and \ref{lemma13} is not uniform in $N$ and $T^{-1}$
and the approximation error of the observable dynamics in Lemma \ref{lemma12} is not uniform in $\tau$.
This means that many particles and low temperature yields a large approximation error of the density operator.
The approximation error of the observable dynamics depends exponentially on time, $\EXP{c\tau}$, but $c$ is
uniform in $N$ provided the assumptions in Theorem \ref{gibbs_corr_thm_analytic} hold uniformly in $N$.
 In conclusion by combining \eqref{S_def_q}, \eqref{steps} and Lemmas \ref{lemma1} -- \ref{lemma13},
 we  obtain the theorems.
 
 The three proofs of the lemmas, in Section \ref{sec:proof}, have the same structure with three steps -
 find an error representation, estimate remainder terms in Moyal expansions and
 evaluate the trace - described roughly as follows.
 \subsubsection{Error representation}
 In the case of Lemma \ref{lemma12}, we  compare the classical dynamics
 \[
\partial_t y(t,z)=\{\bH_0(z),y(t,z)\} ,\quad t>0\COMMA\quad  y(0,\cdot)= {\bA}_0\COMMA
\]
with the quantum dynamics \[
\OPERW{\bA_t}= \EXP{\IU t M^{1/2} \OPER{\bH}} 
{\bA}_0\EXP{-\IU t M^{1/2} \OPER{\bH}}=:\OPERW{\by(t,z)}\COMMA\]
that satisfies 
\[\partial_t\OPERW{\by(t,z)} = \IU M^{1/2}[\OPER{\bH}, \OPERW{\by(t,z)}]\COMMA \quad t>0\COMMA\quad  \OPERW{\by(0,\cdot)}
=\OPER{ \bA}_0\PERIOD
\]
The definition of the composition rule \eqref{comp_rule} yields
\[
\partial_t \by(t,z)= \IU M^{1/2} \big(\bH(z) \MP  \by(t,z) - \by(t,z)\MP  \bH(z)\big)
\COMMA \quad t>0\COMMA\quad  {\by(0,\cdot)}
={\bA}_0\COMMA
\]
and Duhamel's principle applied to $\OPERW{ y(t,z)}-\OPERW{\by(t,z)}$
implies the error representation
\begin{equation}\label{error_rep23}
\begin{split}
&\OPERW{ y(t,z)}-\OPERW{\by(t,z)} \\
&=\int_0^t \EXP{\IU(t-s)M^{1/2}\OPER{\bH}}
\Big(\{\bH_0(z),y(s,z)\} - \IU M^{1/2} \big(\bH \MP  y(s,z) - y(s,z)\MP  \bH\big) \Big)^{\OPERW{}} \\
&\qquad \times \EXP{-\IU(t-s)M^{1/2}\OPER{\bH}} \Rd  s
%
\PERIOD\\
\end{split}
\end{equation}

For the other two lemmas, related representations are obtained using $\EXP{-t\OPER{\bar H}}$ instead of $\EXP{\IU t M^{1/2}\OPER{\bar H}}$.


%
%
%
%
%
%

\subsubsection{Estimation of remainder terms and evaluation of the trace}
We will use the composition rule \eqref{comp_rule} to estimate remainder terms in the error representation \eqref{error_rep23}. 
Expansion of the exponential in the
 composition rule leads to the so called Moyal expansion. The usual estimates of the remainder terms in Moyal expansions determine the $L^2(\rset^N)$ operator norm from $L^\infty$ norm estimates of order $N$ derivatives of the remainder symbol, using the  Calderon-Vaillancourt theorem, see 
 \cite[Theorem 4.23]{zworski}. To avoid derivatives of high order, if $N$ is large,
we instead estimate the remainder terms %
in the form $\TR(\OPER{R}\OPER{C})$,
for Hermitian operators $\OPER{R}$ and $\OPER{C}$ on $L^2(\rset^N)$,
by the $L^2$ norms of their symbols. We use the Hilbert-Schmidt inner product, $\TR(\OPER{R}^*\OPER{C})$, and the corresponding  Hilbert-Schmidt norm, 
$\|\OPER{R}\|_{\mathcal{HS}}^2=\TR(\OPER{R}^*\OPER{ R})=\TR(\OPER{R}^2)$, and Lemma~\ref{composition} as follows
\begin{equation}\label{trace_hs}
\begin{split}
|\TR(\hat R^*\hat C) |^2 &\le \TR(\hat R^*\hat R)\TR(\hat C^*\hat C)\\
&=\NORMFAC^{2N}\int_{\rset^{2N}} \TR\, R^2(z) \Rd z \int_{\rset^{2N}} \TR \, C^2(z) \Rd z\PERIOD
\end{split}
\end{equation}

Lemmas~\ref{moyal_lemma} and \ref{comp_lemma} estimate the $L^2(\rset^{2N})$ norm of the remainder terms  in Moyal compositions
by integration by parts, roughly as follows
\begin{equation}\label{normhs}
\begin{split}
&\int_{\rset^{2N}} |\EXP{\frac{\IU}{2M^{1/2}} ( \nabla_{x'}\cdot\nabla_{p} -  \nabla_{x}\cdot\nabla_{p'})} 
\underbrace{A(x,p)B(x',p')}_{=:r(x,p,x',p')}|^2_{z'=z} \Rd z \\
&= \int_{\rset^{4N}} \big(\EXP{\frac{\IU}{2M^{1/2}} ( \nabla_{x'}\cdot\nabla_{p} -  \nabla_{x}\cdot\nabla_{p'}) }r(z,z') \big)^*\\
&\qquad\times \big(\EXP{\frac{\IU}{2M^{1/2}} ( \nabla_{x'}\cdot\nabla_{p} -  \nabla_{x}\cdot\nabla_{p'})} r(z,z')\big)\delta(z-z')\Rd z\Rd z' \\
&= \int_{\rset^{4N}} r^*(z,z')\\
&\quad\times {\EXP{\frac{-\IU}{2M^{1/2}}( \nabla_{x'}\cdot\nabla_{p} -  \nabla_{x}\cdot\nabla_{p'}) }\Big(\big(
\EXP{\frac{\IU}{2M^{1/2}} ( \nabla_{x'}\cdot\nabla_{p} -  \nabla_{x}\cdot\nabla_{p'})}} r(z,z')\big)\delta(z-z')\Big) \Rd z\Rd z' \\
&=\ldots\le \frac{1}{(2\pi)^{4N}} \|\FT  A\|^2_{L^1(\rset^{2N})}\|B\|^2_{L^2(\rset^{2N})} ,
\end{split}
\end{equation}
where the last steps indicated by $"=\ldots\le "$ 
are explained in the proof of Lemmas~\ref{moyal_lemma} and \ref{comp_lemma} using the Fourier transform of the Dirac measure $\delta(z-z')$,  
that $\IU\nabla_{x}\cdot\nabla_{p}$ is anti-Hermitian (so that $(\EXP{\frac{\IU}{2M^{1/2}} \nabla_{x}\cdot\nabla_{p} })^*=
\EXP{-\frac{\IU}{2M^{1/2}} \nabla_{x}\cdot\nabla_{p} }$  is unitary)  and applying Young's inequality to convolutions of Fourier transforms. Here $\FT $ denotes the Fourier transform \eqref{FT:definition}. 
In our proof of Theorems~\ref{gibbs_corr_thm_unif} and \ref{gibbs_corr_thm_analytic} and to prove \eqref{G_corr} in Theorem~\ref{gibbs_corr_thm} 
we need this estimate only in the
special case where one function depends only on the 
$x$ coordinate, i.e. $A(x)$, and then the right hand side becomes $ \| A\|^2_{L^\infty(\rset^{N})}\|B\|^2_{L^2(\rset^{2N})} $. It is a substantial difference using 
$ \| A\|^2_{L^\infty(\rset^{N})}$ since  a bound on $\|\FT  A\|^2_{L^1(\rset^{2N})}$  is related to $2(N+1)$ derivatives  of $A$ in $L^2(\rset^{2N})$, see \eqref{2N_der}.

The special form of the Hamiltonian, namely $H(x,p)=\frac{|p|^2}{2}{\rm I} + V(x)$, is essential to only obtain
the case $\|A(x)\MP  C(x,p)\|_{L^2(\rset^{2N})}$ 
in our analysis for the $\BIGO(M^{-1/2})$ bound in \eqref{G_corr},  which by 
Lemma \ref{comp_lemma} is bounded by $\|A\|_{L^\infty(\rset^N)}\|C\|_{L^2(\rset^{2N})}$,
and not \[\|C(x,p)\MP  D(x,p)\|_{L^2(\rset^{2N})}\le  \frac{1}{(2\pi)^{2N}}\|\FT  C\|_{L^1(\rset^{2N})}\|D\|_{L^2(\rset^{2N})}\, \]
 from \eqref{CD_fourier}.
An $L^1$-bound on the Fourier transform of a function, which is required in \eqref{R8a} and \eqref{R9a}  
to obtain the accuracy $\BIGO(M^{-1})$ in \eqref{G_corr}, is more demanding on regularity than the $L^\infty$-norm
of the function. For instance, we have
\begin{equation}\label{2N_der}
\begin{split}
\int_{\rset^{2N}} |\FT  g(\xi)|\Rd \xi 
&= \int_{\rset^{2N}}(1+|\xi|^2)^{(N+1)} |\FT  g(\xi)|(1+|\xi|^2)^{-(N+1)}\Rd \xi\\
&\le \|(1+|\xi|^2)^{(N+1)} \FT  g(\xi)\|_{L^2(\rset^{2N})}\|(1+|\xi|^2)^{-(N+1)}\|_{L^2(\rset^{2N})}\\
&\le \mathrm{const.}\,\|(1+\Delta)^{(N+1)}g\|_{L^2(\rset^{2N})}\PERIOD
\end{split}
\end{equation}

The eigenvalue functions $\tlambda_j(x)$ and their Laplacian  are typically proportional  to the number of 
particles, since the Hamiltonian is the energy of the system.
Therefore, the corresponding estimates in the first row of \eqref{R1a} are bounded by
a constant proportional to $N$, %
while \eqref{R7a} can be uniform with respect to $N$. Also the remainder term $r_0$, related to $4M^{-1}\nabla\Psi^*\cdot\nabla\Psi$ in \eqref{diag_comp}, may be proportional to $N$. Therefore
also the estimate $\|\lambda_j-\tilde\lambda_j\|_{L^\infty(\rset^N)}=\BIGO(M^{-1})$,
obtained from \eqref{diag_comp} with $\kappa=1$, may have a constant proportional to $N$.

\subsection{Remainder terms in the Moyal composition}
The following two lemmas estimate remainder terms in the Moyal expansions of the compositions
that we will use below. Here $\FT [C]$ denotes the standard Fourier transform of $C$, see \ref{FT:definition}. 
We also use the abbreviation $z \equiv (x,p)\in\rset^{2N}$ and $\zeta\equiv (\xi_x,\xi_p)\in\rset^{2N}$ 
in which case we write
$\FT[C](\zeta) = \int_{\rset^{2N}} C(z') \EXP{-\IU z'\cdot \zeta}\Rd z'$.
%
\begin{lemma}\label{moyal_lemma}
Assume $C:\rset^{2N}\rightarrow \cset^{d\times d}$ and $D:\rset^{2N}\rightarrow \cset^{d\times d}$
and that there exist constants $M_\gamma$, $N_\gamma$ such that for integer multi-indices $\gamma = (\gamma_1,\dots,\gamma_{2N})$
\[
\begin{split}
 \|\partial^\gamma_z C\|_{L^2(\rset^{2N})} +\|\partial^\gamma_z D\|_{L^2(\rset^{2N})} &\leq M_\gamma\COMMA\quad \mbox{for all $|\gamma|\leq m+1$}\, , \\
 \| \FT[\partial^\gamma_z C]\|_{L^1(\rset^{2N})} &\leq N_\gamma\COMMA\quad \mbox{for all $|\gamma|\leq m+1$} \COMMA 
 \end{split}
\]
then the composition has the expansion
\[
\begin{split}
&C\MP D(x,p)\\
&= 
\sum_{n=0}^m\frac{1}{n!} 
\left(\IU\frac{\nabla_{x'}\cdot \nabla_{p} -\nabla_{x}\cdot \nabla_{p'}}{2M^{1/2}}\right)^n C(x,p)D(x',p')\Big|_{(x,p)=(x',p')} + r(z)\COMMA\\
\end{split}
\]
where the remainder $r\in L^2(\rset^{2N})$ satisfies
\[
\begin{split}
&\|r\|_{L^2(\rset^{2N})}^2
\le 
\left(\frac{1}{4M}\right)^{m+1}\frac{1}{(m!)^2(2m+1)}\\
&\qquad\times \|\int_0^1 \EXP{\frac{\IU s}{M^{1/2}}(\nabla_{x'}\cdot \nabla_{p} -\nabla_{x}\cdot \nabla_{p'})}\\
&\qquad\times (\nabla_{x'}\cdot \nabla_{p} -\nabla_{x}\cdot \nabla_{p'})^{m+1} C(x,p)D(x',p')\Rd s\Big|_{z=z'}
\|^2_{L^2(\rset^{2N})} \\
&=\BIGO(M^{-(m+1)})\PERIOD
\end{split}
\]
If $C(x,p)=A(x)$ %
depends only on the $x$-coordinate and 
\[
\|\partial^\gamma_x A\|_{L^\infty(\rset^{N})} \leq M_\gamma\COMMA\;\;\mbox{for all $|\gamma|\leq m+1$}
\]
then
%
%
%
\begin{equation}\label{1rxp}
\begin{split}
&r(x,p) = \left(\frac{1}{2M^{1/2}}\right)^{m+1}\\
&\quad \times \int_0^1  \EXP{-\frac{\IU s}{2}M^{-1/2}\nabla_{x}\cdot\nabla_{p}}
(-\IU\nabla_{x}\cdot\nabla_{p})^{m+1}{A(x)D(x',p)}  \frac{(1-s)^m}{ m! }\Rd s\Big|_{x'=x}\COMMA\\
\end{split}
\end{equation}
or if $C(x,p)=A(p)$ %
depends only on the $p$-coordinate and 
$
\|\partial^\gamma_p A\|_{L^\infty(\rset^{N})} \leq M_\gamma\COMMA\;\;\mbox{for all $|\gamma|\leq m+1$}$
then
\begin{equation}\label{1rxp2}
\begin{split}
&r(x,p) 
= \left(\frac{1}{2M^{1/2}}\right)^{m+1}\\
&\quad\times\int_0^1  \EXP{\frac{\IU s}{2}M^{-1/2}\nabla_{x}\cdot\nabla_{p}}
(\IU\nabla_{x}\cdot\nabla_{p})^{m+1}{A(p) D(x,p')}  \frac{(1-s)^m}{ m! }\Rd s\Big|_{p'=p}\PERIOD\\
\end{split}
\end{equation}
\end{lemma}

\begin{lemma}\label{comp_lemma}
Assume  $C:\rset^{2N}\rightarrow \cset^{d\times d}$ and $D:\rset^{2N}\rightarrow \cset^{d\times d}$
belong to $L^2(\rset^{2N})$
and in addition one of these functions has  its Fourier transform in $L^1(\rset^{2N})$, then
\begin{equation}\label{CD_fourier}
\begin{split}
&\|C\MP D\|_{L^2(\rset^{2N})}\\
&\le  \frac{1}{(2\pi)^{2N}}\min(\|C\|_{L^2(\rset^{2N})} \|\FT  D\|_{L^1(\rset^{2N})},
\|\FT  C\|_{L^1(\rset^{2N})} \| D\|_{L^2(\rset^{2N})})\COMMA
\end{split}
\end{equation}
and if $A:\rset^{N}\rightarrow \cset^{d\times d}$ depends only on the $x$-coordinate (or only on the $p$-coordinate) and is 
bounded in $L^\infty(\rset^{N})$ then
\begin{equation}\label{CD}
\begin{split}
\|A\MP D\|_{L^2(\rset^{2N})}&\le \|A\|_{L^\infty(\rset^N)}\|D\|_{L^2(\rset^N)}\COMMA\\
\|D\MP  A\|_{L^2(\rset^{2N})}&\le \|A\|_{L^\infty(\rset^N)}\|D\|_{L^2(\rset^N)}\PERIOD\\
\end{split}
\end{equation}
\end{lemma}
The proofs of Lemmas \ref{moyal_lemma} and \ref{comp_lemma} are in Section~\ref{sec_weyl}.

\section{Proofs}\label{sec:proof} 
 \subsection{Proof of Lemma \ref{lemma1}}
 \begin{proof}
 The proof has three steps: error representations, estimation of remainder terms and evaluation of the trace.

{\it Step 1. Error representations.} %
Let $y:[0,\infty)\times \mathbb R^{2N}\rightarrow \CSP$
and $\by_i:[0,\infty)\times \mathbb R^{2N}\rightarrow \CSP,\ i=0,1,2$ 
be defined by 
\[
\begin{split}
y(t,z)&=\EXP{-t\bH_0(z)}\bB_0\COMMA\\
\OPERW{\by_0(t,z)}& =\EXP{-t\OPER{\bH}}\COMMA\\
\OPERW{\by_1(t,z)}& =\OPER{\bB}_0\EXP{-t\OPER{\bH}}\COMMA\\
\OPERW{\by_2(t,z)}& =\EXP{-t\OPER{\bH}}\OPER{\bB}_0\PERIOD
\end{split}
\] 
Differentiation yields the linear ordinary differential equation, with $z$ as a parameter,
\[
\partial_t y(t,z)=-\bH_0(z)y(t,z),\quad t>0\COMMA\quad  y(0,\cdot)=\bB_0.
\]
The dynamics of $\by_1$ satisfies
\[
\partial_t \OPERW{\by_1(t,z)}=\OPER{\bB}_0\frac{d}{dt}\EXP{-t\OPER{\bH}}
=-\OPER{\bB}_0\EXP{-t\OPER{\bH}}\OPER{\bH}
=-\OPERW{\by_1(t,z)}\OPER{\bH}
\]
 and the corresponding evolution equation for
$\by_1$ is the linear partial differential equation
\[
\partial_t \by_1(t,z)=- \by_1(t,z)\MP  \bH(z) \COMMA \quad t>0,\quad  \by_1(0,\cdot)=\bB_0,
\]
with time-independent generator. Analogously  we obtain the equations
\[
\begin{split}
\partial_t \by_2(t,z) &=-  \bH(z)\MP \by_2(t,z) \COMMA \quad t>0,\quad  \by_2(0,\cdot)=\bB_0\COMMA\\
\partial_t \by_0(t,z) &=-  \bH(z)\MP \by_0(t,z) \COMMA \quad t>0,\quad  \by_0(0,\cdot)=\Id\PERIOD\\
\end{split}
\]

We have the two linear equations
\[
\begin{split}
\partial_t \OPERW{\by_1(t,z)} + {\OPERW{ \by_1(t,z)\MP  {\bH}(z)}} &=0\COMMA\quad t>0\COMMA\\
\partial_t \OPERW{ y(t,z)} + \OPERW{  y(t,z)\MP  \bH(z)} &= \OPERW{  y(t,z)\MP  \bH(z)} - \OPERW{ y(t,z) \bH_0(z)} \COMMA\quad t>0\COMMA\\
\end{split}
\]
and Duhamel's principle implies
\begin{equation}\label{duhamel}
\begin{split}
\hat y(t,z)-\OPERW{\by_1(t,z)} 
&=\int_0^t \big(y(s,z)\MP  \bH(z)-y(s,z)\bH_0(z)\big)^{\OPERW{ }}\, \OPERW{\by_0(t-s,z)} \Rd  s\\
&= \int_0^t \underbrace{\big(y(s,z)\MP  \bH(z)-y(s,z)\bH_0(z)\big)^{\OPERW{ }}}_{=:\hat R_{1s}}\, 
\EXP{-(t-s)\OPER{\bH}} \Rd  s\PERIOD\\
\end{split}
\end{equation}
Similarly we have
\[
\begin{split}
\partial_t \OPERW{\by_2(t,z)} +  \OPERW{ {\bH}(z) \MP  \by_2(t,z) } &=0\COMMA\quad t>0\COMMA\\
\partial_t \OPERW{ y(t,z)} + \OPERW{ {\bH}(z)\MP  y(t,z) }&=  \OPERW{  {\bH}(z)\MP  y(t,z)} - \OPERW{ y(t,z) \bH_0(z)} \COMMA\quad t>0\COMMA\\
\end{split}
\]
and 
\begin{equation}\label{duhamel_22}
\begin{split}
\hat y(t,z)-\OPERW{\by_2(t,z)} 
&=\int_0^t \, \OPERW{\by_0(t-s,z)}\big( \bH(z) \MP y(s,z)-y(s,z)\bH_0(z)\big)^{\OPERW{ }} \Rd  s\\
&= \int_0^t \EXP{-(t-s)\OPER{\bH}}\underbrace{\big(\bH(z)\MP y(s,z) -y(s,z)\bH_0(z)\big)^{\OPERW{ }}}_{=:\hat R_{2s}}\,  \Rd  s\PERIOD\\
\end{split}
\end{equation}

The remainder terms satisfy by \eqref{comp_rule} 
\begin{equation}\label{R_est}
\begin{split}
R_{1s}&=(\EXP{-s\bH_0}\bB_0)\MP  \bH-\EXP{-s\bH_0}\bB_0 \bH_0\\
&= 
\big((\EXP{-s\bH_0}\bB_0)\MP  \bH_0  -\EXP{-s\bH_0}\bB_0 \bH_0\big)+ (\EXP{-s\bH_0}\bB_0)\MP  r_0\\
&=\big( -\frac{\IU M^{-1/2}}{2}\{\EXP{-s\bH_0}\bB_0 ,\bH_0\} + r_{2}\big) + (\EXP{-s\bH_0}\bB_0)\MP  r_0\\
&=\big(-\frac{\IU}{2M^{1/2}}\bB_0\underbrace{ {\{\EXP{-s\bH_0} ,\bH_0\}}}_{=0} 
-\frac{\IU}{2M^{1/2}}\EXP{-s\bH_0} {\{\bB_0 ,\bH_0\}}  + r_{2}\big)\\
&\quad + (\EXP{-s\bH_0}\bB_0)\MP  r_0\\
&=\big(-\frac{\IU}{2M^{1/2}}\EXP{-s\bH_0} {\{\bB_0 ,\bH_0\}} + r_{2}\big)+ (\EXP{-s\bH_0}\bB_0)\MP  r_0\\
\end{split}
\end{equation}
and analogously
\begin{equation}\label{R_est2}
\begin{split}
R_{2s}&=(\bH\MP  \EXP{-s\bH_0}\bB_0)-\EXP{-s\bH_0}\bB_0 \bH_0\\
&=(\bH_0\MP  \EXP{-s\bH_0}\bB_0)-\EXP{-s\bH_0}\bB_0 \bH_0+r_0\MP (\EXP{-s\bH_0}\bB_0)\\
&=\big(-\frac{\IU}{2M^{1/2}} {\{\bH_0,\EXP{-s\bH_0}\bB_0 \}}+ r'_{2}\big)+ r_0\MP (\EXP{-s\bH_0}\bB_0)\\
&=\big(-\frac{\IU}{2M^{1/2}}\EXP{-s\bH_0} {\{\bH_0,\bB_0 \}}+ r'_{2}\big)+ r_0\MP (\EXP{-s\bH_0}\bB_0)\PERIOD\\
\end{split}
\end{equation}
We have $r_{2}=(r_2')^*$ and the next step shows that 
\[
\|r_2\|_{L^2(\rset^{2N})}=\BIGO(M^{-1})\PERIOD\]

{\it Step 2. Estimation of remainder terms.}
The remainder representations \eqref{1rxp} and \eqref{1rxp2} applied to the $x$ and $p$ dependent terms in $\bH_0$ separately implies %
\[
\begin{split}
&r_2(x,p,s)\\
&=\frac{-1}{4M} \int_0^1 \EXP{\frac{\IU\sigma}{2M^{1/2}} (\nabla_x\cdot\nabla_p)}
\Big((\nabla_x\cdot\nabla_p)^2\big(\frac{|p|^2}{2}\Id \EXP{-s \bH_0(x,p')}\bB_0(x,p')\big)\\
&\qquad 
+(\nabla_x\cdot\nabla_p)^2\big(\bLambda(x) \EXP{-s \bH_0(x',p)}\bB_0(x',p)\big)\Big)\Big|_{
{\tiny \begin{array}{c}
 x=x'\\
p=p'
\end{array}}}
(1-\sigma)\Rd  \sigma\\
\end{split}
\]
which can be written
\begin{equation}\label{r_2_uppskattn}
\begin{split}
&r_2(x,p,s)\\
&=\frac{-1}{4M} \int_0^1 \EXP{\frac{\IU\sigma}{2M^{1/2}} (\nabla_x\cdot\nabla_p)}
\Big(\EXP{-s \bH_0(x,p')}\bB_0(x,p')(s^2 \nabla\bLambda(x)\cdot\nabla\bLambda(x) 
-s \Delta\bLambda(x)) \\
&\qquad -2s \EXP{-s \bH_0(x,p')}\nabla\bLambda(x)\cdot \nabla_x\bB_0(x,p')
+\EXP{-s \bH_0(x,p')}\Delta_x\bB_0(x,p') \\
&\qquad 
-s \EXP{-s \bH_0(x',p)}\bB_0(x',p)\Delta\bLambda(x)\\
&\qquad +s^2 \EXP{-s \bH_0(x',p)}\bB_0(x',p)\sum_{jk}p_jp_k\partial_{x_j}\partial_{x_k}\bLambda(x)\\
&\qquad -2s\EXP{-s \bH_0(x',p)}\sum_{jk}p_j\partial_{x_j}\partial_{x_k}\bLambda(x)\partial_{p_k}\bB_0(x',p)\\
&\qquad +\EXP{-s \bH_0(x',p)}\sum_{jk} \partial_{x_j}\partial_{x_k}\bLambda(x)\partial_{p_k}\partial_{p_k}\bB_0(x',p)
\Big)_{x=x',p=p'}
(1-\sigma)\Rd  \sigma\PERIOD\\
\end{split}
\end{equation}
We have by \eqref{CD} in Lemma \ref{comp_lemma}, where $A$ is the function of $x$ and $D$ is the function of $(x',p')$ in the estimates of $r_0$ and $r_2$ in \eqref{diag_comp} and \eqref{r_2_uppskattn},
\[
\begin{split}
&\|r_0\MP  \EXP{-s \bH_0}\bB_0\|_{L^2(\rset^{2N})}
\le \|r_0\|_{L^\infty(\rset^{N})}\|\EXP{-s \bH_0}\bB_0\|_{L^2(\rset^{2N})}\COMMA \\
&\frac{\|\EXP{-s \bH_0} \{\bB_0,\bH_0\}\|_{L^2(\rset^{2N})}}{M^{1/2}}
=\frac{\|\EXP{-s \bH_0}(p\cdot\nabla_x\bB_0
-\nabla\bLambda\cdot \nabla_p\bB_0)\|_{L^2(\rset^{2N})}}{M^{1/2}}\COMMA\\
\end{split}
\]
\[
\begin{split}
&4M\|r_{2}\|_{L^2(\rset^{2N})}
\le s \|\EXP{-s \bH_0}\bB_0(x,p)\|_{L^2(\rset^{2N})}\|\Delta\bLambda\|_{L^\infty(\rset^{2N})}\\
&\qquad +s^2\sum_{jk} \|p_jp_k\EXP{-s \bH_0}\bB_0(x,p)\|_{L^2(\rset^{2N})}\|\partial_{x_jx_k}\bLambda\|_{L^\infty(\rset^{2N})}\\
&\qquad +2s\sum_{jk} \|p_j\EXP{-s \bH_0}\partial_{p_k}\bB_0(x,p)\|_{L^2(\rset^{2N})}\|\partial_{x_jx_k}\bLambda\|_{L^\infty(\rset^{2N})}\\
&\qquad +\sum_{jk} \|\EXP{-s \bH_0}\partial_{p_jp_k}\bB_0(x,p)\|_{L^2(\rset^{2N})}\|\partial_{x_jx_k}\bLambda\|_{L^\infty(\rset^{2N})}\\
&\qquad +\|\EXP{-s \bH_0}\bB_0(x,p)(s^2 \nabla\bLambda\cdot\nabla\bLambda-s \Delta\bLambda)\|_{L^2(\rset^{2N})}\\
&\qquad +2s \|\EXP{-s \bH_0}\nabla\bLambda\cdot\nabla_x\bB_0(x,p)\|_{L^2(\rset^{2N})}\\
&\qquad +\|\EXP{-s \bH_0}\Delta_x\bB_0(x,p)\|_{L^2(\rset^{2N})}\PERIOD\\
\end{split}
\]
Here we see that $r_0$ depends only on the $x$-coordinate 
and the composition in $r_2$ has one factor that also depends only on the $x$-coordinate.
Therefore, by Lemma \ref{moyal_lemma} and \eqref{CD} we obtain
\begin{equation}\label{RR_est}
\|R_{1s}\|_{L^2(\rset^{2N})} + \|R_{2s}\|_{L^2(\rset^{2N})} = \BIGO(M^{-1/2})\COMMA
\end{equation}
provided  there holds
\begin{equation}\label{R0}
\|r_0\|_{L^\infty(\rset^{N})} =\BIGO(M^{-1})\COMMA\\
\end{equation}
and
\begin{equation}\label{R1}
\begin{split}
&\||\nabla\bLambda|^2\|_{L^\infty(\rset^{N})}+\|\Delta\bLambda\|_{L^\infty(\rset^{N})}+\max_{j,k}\|\partial_{x_j}\partial_{x_k}\bLambda\|_{L^\infty(\rset^{N})} =\BIGO(1)\COMMA\\
& \sum_{jk}\sup_{t\in [0,1/T]}
\|\EXP{-t\bH_0}{p_j}\partial_{p_k}\bB_0\|_{L^2(\rset^{2N})} \\
&\quad +\sum_{jk}\sup_{t\in [0,1/T]}
\|\EXP{-t\bH_0}\partial_{p_j}\partial_{p_k}\bB_0\|_{L^2(\rset^{2N})} 
=\BIGO(1)\COMMA\\
& \sum_{jk}\sup_{t\in [0,1/T]}
\|\EXP{-t\bH_0}{p_j}{p_k}\bB_0\|_{L^2(\rset^{2N})} 
=\BIGO(1)\COMMA \\
&\sup_{t\in [0,1/T]}
\|\EXP{-t\bH_0}\Delta_x\bB_0\|_{L^2(\rset^{2N})}
+\sup_{t\in [0,1/T]}
\|\EXP{-t\bH_0}\nabla\bLambda\cdot\nabla_x\bB_0\|_{L^2(\rset^{2N})}
=\BIGO(1)\COMMA\\
&\sup_{t\in [0,1/T]}
\|\EXP{-t\bH_0}\bB_0\|_{L^2(\rset^{2N})}
=\BIGO(1)\COMMA\\
&\sup_{t\in [0,1/T]}\|\EXP{-t\bH_0}(p\cdot\nabla_x\bB_0
-\nabla\bLambda\cdot \nabla_p\bB_0)\|_{L^2(\rset^{2N})} =\BIGO(1)\PERIOD\\
\end{split}
\end{equation}

{\it Step 3. Evaluation of the trace.}
The Hilbert-Schmidt inner product, $\TR(\OPER{B}^*\hat C)$,  for symmetric operators on
$L^2(\rset^{N})$, and its Cauchy's inequality imply together with  \eqref{duhamel} 
\begin{equation}\label{A_trace_new}
\begin{split}
&\Big[\TR\Big(\OPER{\bA}_\tau\big(\OPER{\by}(t)-\hat y(t)\big)\Big)\Big]^2 \\
&=\big(\TR(\OPER{\bA}_\tau\int_0^t\hat R_{1s} \EXP{-(t-s)\OPER{\bH}} \Rd s)\big)^2 \\
&=\big( \int_0^t \TR(\OPER{\bA}_\tau\hat R_{1s} \EXP{-(t-s)\OPER{\bH}})\Rd s\big)^2\\
&\le  \Big( \int_0^t \big(\TR(\EXP{\IU\tau M^{1/2}\OPER{\bH}}\OPER{\bA}_0^2 \EXP{-\IU\tau M^{1/2}\OPER{\bH}})
\TR(\EXP{-(t-s)\OPER{\bH}} \hat R_{1s}^2   \EXP{-(t-s)\OPER{\bH}} )\big)^{1/2}\Rd s\Big)^2\\
&=  \Big( \int_0^t \big(\TR(\OPER{\bA}_0^2 )
\TR( \hat R_{1s}^2   \EXP{-2(t-s)\OPER{\bH}} )\big)^{1/2}\Rd s\Big)^2\\
&\le  \int_0^t \TR(\OPER{\bA}_0^2) \Rd  s \int_0^t 
\TR((\underbrace{\hat\Psi\hat R_{1s}\hat\Psi^*}_{=:\OPER{\bar R}_s})^2 \EXP{-2(t-s)\hat { H}} )\Rd s\PERIOD\\
\end{split}
\end{equation}
Lemma \ref{composition} establishes
\[
 \int_0^t \TR(\OPER{\bA}_0^2) \Rd  s 
 =t\NORMFAC^{N}\int_{\rset^{2N}} \TR \big(\bA^2(0,z)\big)\Rd  z
 \]
 and we assume that the initial data satisfies
 \begin{equation}\label{R2}
 \|\bA_0\|_{L^2(\rset^{2N})}=\BIGO(1)\PERIOD
 \end{equation}
 The basis $\{ \Phi_n\}_{n=1}^\infty$ of solutions to the Schr\"odinger equation \eqref{schrod_eq} implies  
\[
\begin{split}
\TR( \OPER{\bar R}_s^2 \EXP{-2(t-s)\hat { H}} ) &=
\sum_{n=1}^\infty\langle \Phi_n,  \OPER{\bar R}_s^2\EXP{-2(t-s)\hat { H}} \Phi_n\rangle\\
&=\sum_{n=1}^\infty\langle \Phi_n,  \OPER{\bar R}_s^2 \Phi_n\rangle \EXP{-2(t-s)E_n/T}\\
&\le \EXP{-2(t-s)E_1/T}\sum_{n=1}^\infty\langle \Phi_n,  \OPER{\bar R}_s^2 \Phi_n\rangle\\
&= \EXP{-2(t-s)E_1/T} \TR\big(\hat R_s^*\hat R_s\big)
\, \\
\end{split}
\]
and Lemmas \ref{composition}, \ref{moyal_lemma} and \eqref{CD} combined with \eqref{RR_est} show that
\begin{equation}\label{cr}
\begin{split}
\TR\big( \hat R_s^*\hat R_s\big) 
&= \NORMFAC^{N} \int_{\rset^{2N}} \TR \big(R_s^*R_s\big)\Rd z\, \\
&=\BIGO(M^{-1+N/2})\COMMA\\
\end{split}
\end{equation}
so that
\[
\TR\Big(\OPER{\bA}_\tau\big(\OPER{\by}(t)-\hat y(t)\big)\Big)=\BIGO(M^{-1/2+N/2})\PERIOD
\]
In the special case where $\bA_0=\bB_0=\Id$, we similarly obtain
\[
\TR(\EXP{-\beta\OPER{\bH}}-\OPERW{ \EXP{-\beta \bH_0}})=\BIGO(M^{-1/2+N/2})\, 
\]
provided
\begin{equation}\label{R10}
\sup_{t\in [0,1/T]}
\|\EXP{-t\bH_0}\|_{L^2(\rset^{2N})}= \BIGO(1)\COMMA
\end{equation}
and by \eqref{trace1}  %
\begin{equation}\label{trace_normalize}
\begin{split}
\TR(\widehat {\EXP{-\beta \bH_0}})&= \NORMFAC^{N} \int_{\rset^{2N}} \TR (\EXP{-\beta \bH_0(z)}
)\Rd  z\COMMA\\
\TR(\EXP{-\beta\OPER{\bH}}) &= \NORMFAC^{N} \big(\int_{\rset^{2N}} \TR (\EXP{-\beta \bH_0(z)}
)\Rd  z + \BIGO(M^{-1/2})\big)\PERIOD\\
\end{split}
\end{equation}
In conclusion we have for $M$ sufficiently large
\[
\begin{split}
&|\frac{\TR\Big(\OPER{\bA}_\tau \big((\bB_0{\EXP{-\beta \bH_0}})^{\OPERW{}} 
-\OPER{\bB}_0{\EXP{-\beta \OPER{\bH}}}\big)\Big)}{
\TR ({\EXP{-\beta\OPER{\bH}}})}| 
=\BIGO(M^{-1/2})\PERIOD
\end{split}
\]
%
\end{proof}

\subsection{Proof of Lemma \ref{lemma12}}
\begin{proof}
This proof is analogous to the proof of Lemma \ref{lemma1} and has three similar steps: error representation, estimation of remainder terms, and evaluation of the trace. 

{\it Step 1. Error representation.} We will compare the classical dynamics
$\OPERW{y(t,z_0)}:=\OPERW{{\brA}(t,z_0)}$ 
with the quantum dynamics \[
\OPERW{\bA_t}= \EXP{\IU t M^{1/2} \OPER{\bH}} 
{\bA}_0\EXP{-\IU t M^{1/2} \OPER{\bH}}=:\OPERW{\by(t,z)}\PERIOD\]
The classical dynamics for the symbol $y:[0,\infty)\times \mathbb R^{2N}\rightarrow \CSP$
satisfies by \eqref{poissonbracket}  and \eqref{a_char} the linear partial differential equation 
\[
\partial_t y(t,z)=\{\bH_0(z),y(t,z)\} ,\quad t>0\COMMA\quad  y(0,\cdot)= {\bA}_0\COMMA
\]
that is, we have the diagonal matrix
\[
y_{jk}(t,z_0)=\left\{\begin{array}{cc}
(\bA_0)_{jj}(z_t^j(z_0)) & \mbox{ if } k=j\COMMA\\
0 & \mbox{ if } k\ne j\COMMA\\
\end{array}\right.
\quad \mbox{ for } z_0\in \rset^{2N} \mbox{ and } t\ge 0\PERIOD
\]
The evolution of $\by:[0,\infty)\times \mathbb R^{2N}\rightarrow \CSP$ 
is defined by the quantum dynamics 
\[\partial_t\OPERW{\by(t,z)} = \IU M^{1/2}[\OPER{\bH}, \OPERW{\by(t,z)}]\COMMA \quad t>0\COMMA\quad  \OPERW{\by(0,\cdot)}
=\OPER{ \bA}_0,
\]
which implies
\[
\partial_t \by(t,z)= \IU M^{1/2} \big(\bH(z) \MP  \by(t,z) - \by(t,z)\MP  \bH(z)\big)
\COMMA \quad t>0\COMMA\quad  {\by(0,\cdot)}
={\bA}_0\PERIOD
\]
Duhamel's principle yields
\[
\begin{split}
&\OPERW{ y(t,z)}-\OPERW{\by(t,z)} \\
&=\int_0^t \EXP{\IU(t-s)M^{1/2}\OPER{\bH}}
\Big(\{\bH_0(z),y(s,z)\} - \IU M^{1/2} \big(\bH \MP  y(s,z) - y(s,z)\MP  \bH\big) \Big)^{\OPERW{}} \\
&\qquad \times \EXP{-\IU(t-s)M^{1/2}\OPER{\bH}} \Rd  s\\
&=:\int_0^t 
 \EXP{\IU(t-s)M^{1/2}\OPER{\bH}}\hat R_s\EXP{-\IU(t-s)M^{1/2}\OPER{\bH}} \
\Rd  s \PERIOD\\
\end{split}
\]

{\it Step 2. Estimation of remainder terms.}
Since $y(s,z)$ and $\bH_0(z)$ are diagonal the expansion of the composition in Lemma \ref{moyal_lemma}
and \eqref{CD}
imply %
\begin{equation}\label{R_dyn_est}
\begin{split}
R_s &= \{\bH_0, y(s,\cdot)\} -\IU M^{1/2}\big(\bH\MP  y(s,\cdot)-y(s,\cdot)\MP  \bH\big)\\
&= \underbrace{\{\bH_0, y\} -\IU M^{1/2}(\bH_0\MP  y-y\MP  \bH_0)}_{=r_2=\BIGO(M^{-1})}
-\IU M^{1/2}(r_0\MP  y-y\MP  r_0)\\
&=r_2 -\IU M^{1/2}(r_0\MP  y-y\MP  r_0)= \BIGO(M^{-\min(1,\kappa-1/2)})\COMMA\\
\end{split}
\end{equation}
where by Lemma \ref{moyal_lemma}
\[
r_2=\frac{\IU}{16M} \int_0^1 \EXP{\IU\frac{s}{2M^{1/2}} \nabla_x\cdot\nabla_p} \sum_{|\alpha|=3} \partial^\alpha_x\bLambda(x)\partial^\alpha_p y(x',p)\big|_{x'=x}(1-s)^2\Rd  s\COMMA
\]
and by \eqref{CD}
\begin{equation}\label{R5}
\begin{split}
&\|r_2\|_{L^2(\rset^{2N})}
\le  \frac{1}{48M}\ \sum_{|\alpha|=3} \|\partial^\alpha_x\bLambda\|_{L^\infty(\rset^{N})}
\|\partial_p^\alpha y(t,\cdot)\|_{L^2(\rset^{2N})}\COMMA\\
&\|r_0\MP  y-y\MP r_0\|_{L^2(\rset^{2N})}\le 2\|r_0\|_{L^\infty(\rset^{N})}\|y(t,\cdot)\|_{L^2(\rset^{2N})} \PERIOD
\end{split}
\end{equation}
Let $z(t)=z^j_t$. To estimate $\sum_{|\alpha|\le 3}\|\partial_p^\alpha y(t,\cdot)\|_{L^2(\rset^{2N})}$ we use the first order flow $\nabla_{z_0} z^j_t\big(z_0\big)=:z'(t)$, second order flow 
$z''_{,km}(t)=\partial_{z_k(0) z_m(0)} z(t)$ and third order flow $z'''(t)$, which
are solutions to the system
\[
\begin{split}
\dot z_i(t) &= (J\nabla \bH_0(z_t))_i=:f_i(z_t)\COMMA \\
z'_{i,k}(t) &= \Id_{ik}+ \int_0^t\sum_{k'} f'_{i,k'}(z_s)z'_{k',k}(s)\Rd s\COMMA \quad f'_{i,k'}(z):=\partial_{z_{k'}}f_i(z) \COMMA\\
z''_{i,km}(t) &=  \int_0^t \Big(\sum_{k'}  f'_{i,k'}(z_s)z''_{k',k m}(s) +  \sum_{k'm'} f''_{i,k'm'}(z_s)z'_{k',k}(s)z'_{m',m}(s)\Big) \Rd s\COMMA\\
&\qquad f''_{i,k'm'}(z):=\partial_{z_{k'}z_{m'}}f_i(z)\COMMA\\
 z'''_{i,kmn}(t) &=  \int_0^t \Big(\sum_{k'} f'_{i,k'}(z_s)z'''_{k',k mn}(s) 
+  \sum_{k'm'} f''_{i,k'm'}(z_s)z'_{k',k}(s)z''_{m',mn}(s) \\
&\qquad +\sum_{k'm'} f''_{i,k'm'}(z_s)z''_{k',kn}(s)z'_{m',m}(s)\\
&\qquad +\sum_{k'n'} f''_{i,k'n'}(z_s)z''_{k',k m}(s)z'_{n',n}(s)\\
&\qquad+\sum_{k'm'n'} f'''_{i,k'm'n'}(z_s)z'_{k',k}(s)z'_{m',m}(s)z'_{n',n}(s)\Big)\Rd s
\COMMA\\
\end{split}
\]
where $J$ is the $2N\times 2N$ matrix
\begin{equation}\label{J_def}
J=\left[\begin{array}{cc}
0  & \Id\\
-\Id & 0
\end{array}
\right]
\PERIOD
\end{equation}
By summation and maximization over indices we obtain the integral inequalities 
\begin{equation}\label{z_flow_est}
\begin{split}
\max_{ik}
|z'_{i,k}(t)| &\le 1+ \int_0^t \sum_{k'} |f'_{i,k'}(z_s)| \max_{ik}|z'_{i,k}(s)|\Rd  s\COMMA\\
\max_{ik}\sum_{m}
|z''_{i,km}(t)| &\le \int_0^t \sum_{k'} |f'_{i,k'}(z_s)| \max_{ik}\sum_{m}
|z''_{i,km}(s)|\Rd  s\\
&\qquad + \int_0^t \sum_{k'm'} |f''_{i,k'm'}(z_s)| (\max_{ik}|z'_{i,k}(s)|)^2\Rd  s\COMMA\\
\max_{ik}\sum_{mn}
|z'''_{i,kmn}(t)| &\le \int_0^t \sum_{k'} |f'_{i,k'}(z_s)| \max_{ik}\sum_{mn}
|z'''_{i,kmn}(s)|\Rd  s\\
&\quad + \int_0^t \sum_{k'm'} |f''_{i,k'm'}(z_s)| \max_{ik}|z'_{i,k}(s)| \max_{ik}\sum_m
|z''_{i,km}(s)|\Rd  s\\
&\quad + \int_0^t \sum_{k'm'n'} |f'''_{i,k'm'n'}(z_s)| (\max_{ik}|z'_{i,k}(s)|)^3\Rd  s
\PERIOD\\
\end{split}
\end{equation}
The functions $\max_{ij}\sum_{|\alpha|\le 2} \partial_{z_0}^\alpha \partial_{z_j}z_i(t,z_0)$ can therefore be estimated as in \cite{gronwall} by Gronwall's inequality, which states: if there is a positive constant $K$ and  continuous positive functions $\beta, u:[0,\infty)\rightarrow [0,\infty)$ such that  
\[u(t)\le K + \int_0^t \beta(s)u(s)\Rd  s, \quad \mbox{ for } t>0\COMMA\]
 then 
 \[ u(t)\le K\EXP{\int_0^t \beta(s)\Rd  s}, \quad \mbox{ for } t>0\PERIOD\]
Gronwall's inequality applied to \eqref{z_flow_est} implies
\begin{equation}\label{z_tz_0}
\max_{ij}\sum_{|\alpha|\le 2} \|\partial_{z_0}^\alpha \partial_{z_j}z_i(t,z_0)\|_{L^\infty(\rset^{2N})}=\BIGO(1)
\end{equation}
provided that
\begin{equation}\label{R7'}
\begin{split}
\max_i\sum_{|\alpha|\le 3}\|\partial^\alpha_x \partial_{x_i}\lambda_j\|_{L^\infty(\rset^{N})} &= \BIGO(1)\PERIOD\\
\end{split}
\end{equation}

The flows $z',z'',z'''$ determine the derivatives of the diagonal matrix $y(t,\cdot)$, using $a_j(z):=\bA_{jj}(0,z)$, by
\begin{equation}\label{A_z_t}
\begin{split}
 \partial_{z_k}y_{jj}(t) &= \sum_{k'} a'_{j,k'}(z_t)z'_{k',k}(t)\COMMA\qquad a'_{j,k'}(z):=\partial_{z_{k'}}a_j(z),\\
\partial_{z_kz_m}y_{jj}(t) &=  \sum_{k'} a'_{j,k'}(z_t)z''_{k',k m}(t) +  \sum_{k'm'} a''_{j,k'm'}(z_t)z'_{k',k}(t)z'_{m',m}(t)\COMMA\\
\partial_{z_kz_mz_n}y_{jj}(t)&=  \sum_{k'} a'_{j,k'}(z_t)z'''_{k',k mn}(t) 
+  \sum_{k'm'} a''_{j,k'm'}(z_t)z'_{k',k}(t)z''_{m',mn}(t)\\
&\qquad +\sum_{k'm'} a''_{j,k'm'}(z_t)z''_{k',kn}(t)z'_{m',m}(t)\\
&\qquad +\sum_{k'n'} a''_{j,k'n'}(z_t)z''_{k',k m}(t)z'_{n',n}(t)\\
&\qquad+\sum_{k'm'n'} a'''_{j,k'm'n'}(z_t)z'_{k',k}(t)z'_{m',m}(t)z'_{n',n}(t)
\PERIOD\\
\end{split}
\end{equation}
The constant  in the right hand side of \eqref{z_tz_0} grows typically exponentially with respect to $t$, i.e.
\[
\max_{ij}\sum_{|\alpha|\le 2} \|\partial_{z_0}^\alpha \partial_{z_j}z_i(t,z_0)\|_{L^\infty(\rset^{2N})}
\le \EXP{ct}\COMMA\]
 where $c$ is the positive constant in the right hand side of \eqref{R7'}.
The integration with respect to the initial data measure $\Rd  z_0$ can be replaced by integration with respect
to $\Rd  z_t$ since the phase-space volume is preserved, i.e. the Jacobian determinant 
\[
|{\rm det}\big(\frac{\partial z^j_0}{\partial z^j_t}\big)| =1
\]
 is constant for all time, so that
\[
\begin{split}
\int_{\rset^{2N}} |\partial_{z^j_t}^\alpha a_{j}\big(0,z_t^j(x_0,p_0)\big)|^2 \Rd  z_0
&=
\int_{\rset^{2N}} |\partial_{z^j_t}^\alpha\bA_{jj}\big(0,z_t^j(x_0,p_0)\big)|^2 \Rd  z_0\\
&=\int_{\rset^{2N}} |\partial_{z^j_t}^\alpha\bA_{jj}\big(0,z_t^j(x_0,p_0)\big)|^2
|{\rm det}\big(\frac{\partial z^j_0}{\partial z^j_t}\big)|\Rd  z^j_t\\
&=\int_{\rset^{2N}} |\partial_{z^j_t}^\alpha\bA_{jj}\big(0,z_t^j\big)|^2
\Rd  z^j_t\PERIOD\\
\end{split}
\]
Equation \eqref{A_z_t} and \eqref{z_tz_0} therefore imply
\begin{equation*}\label{Az2}
\begin{split}
&\sum_{|\alpha|=1}\sqrt{\int_{\rset^{2N}} |\partial_{p_0}^\alpha\bA_{jj}\big(0,z_t^j(x_0,p_0)\big)|^2 \Rd  z_0}\\
&\le \max_{ik}\|\partial_{z_k(0)}z_i(t)\|_{L^\infty(\rset^{2N})}\sum_{|\alpha|=1}\|\partial_z^\alpha \bA_{jj}(0,z)\|_{L^2(\rset^{2N})}\COMMA\\
\end{split}
\end{equation*}

\begin{equation*}
\begin{split}
&\sum_{|\alpha|=2}\sqrt{\int_{\rset^{2N}} |\partial_{p_0}^\alpha\bA_{jj}\big(0,z_t^j(x_0,p_0)\big)|^2 \Rd  z_0}\\
&\le \sum_n\max_{ik}\|\partial_{z_n(0)z_k(0)}z_i(t)\|_{L^\infty(\rset^{2N})}\sum_{|\alpha|=1}\|\partial_z^\alpha \bA_{jj}(0,z)\|_{L^2(\rset^{2N})}\\
&\qquad +  \max_{ik}\|\partial_{z_k(0)}z_i(t)\|^2_{L^\infty(\rset^{2N})}\sum_{|\alpha|=2}\|\partial_z^\alpha \bA_{jj}(0,z)\|_{L^2(\rset^{2N})}\COMMA\\
\end{split}
\end{equation*}
and

\begin{equation*}\label{Az2_1}
\begin{split}
&\sum_{|\alpha|=3}\sqrt{\int_{\rset^{2N}} |\partial_{p_0}^\alpha\bA_{jj}\big(0,z_t^j(x_0,p_0)\big)|^2 \Rd  z_0}\\
&\le \sum_{mn}\max_{ik}\|\partial_{z_m(0)z_n(0)z_k(0)} z_i(t)\|_{L^\infty(\rset^{2N})}\sum_{|\alpha|=1}\|\partial_z^\alpha \bA_{jj}(0,z)\|_{L^2(\rset^{2N})}\\
&\qquad +  \sum_n\max_{ik}\|\partial_{z_n(0)z_k(0)} z_i(t)\|_{L^\infty(\rset^{2N})}\max_{ik}\|\partial_{z_k(0)}z_i(t)\|^2_{L^\infty(\rset^{2N})}\\
&\qquad\times \sum_{|\alpha|=2}\|\partial_z^\alpha \bA_{jj}(0,z)\|_{L^2(\rset^{2N})}\\
&\qquad + \max_{ik}\|\partial_{z_k(0)}z_i(t)\|^3_{L^\infty(\rset^{2N})}\sum_{|\alpha|=3}\|\partial_z^\alpha \bA_{jj}(0,z)\|_{L^2(\rset^{2N})}\PERIOD\\
 \end{split}
 \end{equation*}
The estimate \eqref{z_tz_0} implies that these right hand sides are bounded provided that
\begin{equation}\label{R7}
\begin{split}
\max_i\sum_{|\alpha|\le 3}\|\partial^\alpha_x \partial_{x_i}\lambda_j\|_{L^\infty(\rset^{N})} &= \BIGO(1)\COMMA\\
\sum_{|\alpha|\le 3}
\|\partial_z^\alpha\bA_{jj}(0,\cdot)\|_{L^2(\rset^{2N})}  &= \BIGO(1)\PERIOD\\
\end{split}
\end{equation}
We note that the assumption on $\lambda_j$ is compatible with the assumption
 that $\tlambda_j(x)$ tends to infinity
as $|x|\rightarrow\infty$, which is imposed to have a discrete spectrum of $\hat H$.
The choice $\Psi=\Psi[2]$  yields $\lambda_j=\tlambda_j + \BIGO(M^{-1})$.
The eigenvalues $\lambda_j$ have four bounded derivatives if the eigenvalues $\tlambda_j$ of the potential $V$ are distinct and 
also the corresponding eigenvectors $\Psi[1]$ and $\Psi[2]$ have four bounded derivatives, which requires
the fifth order derivatives of the potential $V$ to be bounded.  

{\it Step 3. Evaluation of the trace.}
Define the diagonal matrix \[C(z)=\bB_0(z)\EXP{-\beta \bH_0(z)}\PERIOD\] Duhamel's representation implies 
\[
\begin{split}
&\Big|\TR \Big( \hat C\big({\OPERW{\by(t,z)}} - {
\OPERW{y(t)}}\big)\Big)\Big|\\
&=|\TR(\int_0^t 
 \hat C \EXP{\IU(t-s)M^{1/2}\OPER{\bH}}\hat R_s  \EXP{-\IU(t-s)M^{1/2}\OPER{\bH}}\
 \Rd  s)|\\
 &=|\int_0^t  \TR(
 \hat C \EXP{\IU(t-s)M^{1/2}\OPER{\bH}}\hat R_s  \EXP{-\IU(t-s)M^{1/2}\OPER{\bH}})
 \Rd  s|\\
 &=|\int_0^t  \TR(
  \EXP{-\IU(t-s)M^{1/2}\OPER{\bH}}\hat C \EXP{\IU(t-s)M^{1/2}\OPER{\bH}}\hat R_s )
 \Rd  s|\\
 \end{split}
 \]
 and
 Cauchy's inequality for the trace, and Lemma \ref{composition} establish
\[
\begin{split}
&\Big|\TR \Big( \hat C\big({\OPERW{\by(t,z)}} - {
\OPERW{y(t)}}\big)\Big)\Big|\\
 &\le \int_0^t  \big(\TR(\EXP{-\IU(t-s)M^{1/2}\OPER{\bH}}\hat C^*\hat C \EXP{\IU(t-s)M^{1/2}\OPER{\bH}})
  \TR(\hat R_s^*\hat R_s)\big)^{1/2}\Rd s\\
  &= \int_0^t  \big(\TR(\hat C^*\hat C )
  \TR(\hat R_s^*\hat R_s)\big)^{1/2}\Rd s\\
 &=\NORMFAC^{N}\int_0^t \big(\int_{\rset^{2N}} \TR( C^*C)\Rd z\int_{\rset^{2N}}\TR(R_s^*R_s) {\rm dz}\big)^{1/2}\Rd  s\\
&= \NORMFAC^{N} \int_0^t \|\bB_0(z)\EXP{-\beta \bH_0(z)}\|_{L^2(\rset^{2N})}\|R_s\|_{L^2(\rset^{2N})}\Rd  s\COMMA\\
 \end{split}
\]
which  together with \eqref{trace_normalize} prove  the lemma
provided also
\begin{equation}\label{R6}
\begin{split}
\|\bB_0(z)\EXP{-\beta \bH_0(z)}\|_{L^2(\rset^{2N})} &=\BIGO(1)\PERIOD\\
\end{split}
\end{equation}
\end{proof}

\subsection{Proof of Lemma \ref{lemma13}}
\begin{proof} 
The improved bound \eqref{lemma13_eq11} is based on
\begin{equation}\label{r_acc}
\begin{split}
\OPERW{y(t,\cdot)} - \frac{1}{2} \big(\OPERW{\by_1(t,\cdot)} + \OPERW{\by_2(t,\cdot)}\big)
&= \frac{1}{2}\int_0^t \big(\OPERW{R_{1s}} \OPERW{\by_0(t-s,\cdot)}+ \OPERW{\by_0(t-s,\cdot)}\OPERW{R_{2s}}\big)  \Rd  s\COMMA\\
\end{split}
\end{equation}
which is obtained by  \eqref{duhamel} and \eqref{duhamel_22}.
Let $r_1:=M^{-1/2}\{\bB_0,\bH_0\}\EXP{-s\bH_0}$. The estimate \eqref{R_est}
shows that the dominating term in $R_{1s}$ is $r_1$ and similarly  by \eqref{R_est2}
the dominating term in $R_{2s}$ is $-r_1$. By \eqref{duhamel}, applied with $\bB_0= \Id$, 
we obtain \[
\OPER{\by}_0(s)=\hat y(s)+\int_0^s \hat R_1(s')\OPER{\by}_0(s-s')\Rd  s'\COMMA\]
which replaces 
the remainder term $\hat r_1\OPER{\by}_0$ included in \eqref{duhamel} by the smaller term
\[
(r_1\MP   y-y\MP  r_1)^{\OPERW{}}\COMMA
\]
present in \eqref{r_acc}
and
generates the new remainder terms
\[
\begin{split}
 \int_0^t \hat r_1\hat R_1\OPER{\by}_0\, \Rd  s&= \int_0^t (r_1\MP  R_1)^{\OPERW{}}\ 
 \OPER{\by}_0\, \Rd  s\COMMA\\
\int_0^t\OPER{\by}_0 \hat R_1 \hat r_1\, \Rd  s&= \int_0^t  \OPER{\by}_0(R_1\MP  r_1)^{\OPERW{}}
\, \Rd  s\PERIOD\\
\end{split}
\]
Lemma \ref{moyal_lemma} and \eqref{CD_fourier} imply $\|r_1\MP   y-y\MP  r_1\|_{L^2(\rset^{2N})}=\BIGO(M^{-1})$ provided
\begin{equation}\label{R8}
\begin{split}
\frac{1}{(2\pi)^{2N}}\|\FT (\partial_{z_n} \EXP{-\bH_0})\|_{L^1(\rset^{2N})}&=\BIGO(1)\COMMA\\
\sum_n\|\partial_{z_n}(p\cdot\bB_0-\nabla\bLambda\cdot \bB_0)\|_{L^2(\rset^{2N})}&=\BIGO(1)\COMMA\\
\end{split}
\end{equation}
and there holds $\|r_1\MP  R_1\|_{L^2(\rset^{2N})}+\|R_1\MP  r_1\|_{L^2(\rset^{2N})}=\BIGO(M^{-1})$ if 
\begin{equation}\label{R9}
\begin{split}
\frac{1}{(2\pi)^{2N}}\|\FT (p\cdot\bB_0-\nabla\bLambda\cdot \bB_0)\|_{L^1(\rset^{2N})}=\BIGO(1)\, \\
\end{split}
\end{equation}
and \eqref{RR_est} are satisfied.
The new error terms 
\[
\begin{split}
&\TR\Big(\OPER{\bar A}_\tau\int_0^t \big(r_1(s)\MP y(t-s)-y(t-s)\MP r_1(s)\big)\Rd s\Big)\COMMA\\
&\TR\big(\OPER{\bar A}_\tau\int_0^t(r_1\MP  R_1)_s\OPER{\bar y}_0(t-s) \Rd s\big)\COMMA
\end{split}
\]
can then be estimates as in \eqref{A_trace_new}-\eqref{trace_normalize}.
We obtain that the first term  in the right hand side of \eqref{steps} has the bound $\BIGO(M^{-1})$, by assuming
\eqref{R0}, \eqref{R1}, \eqref{R2}, \eqref{R10}, \eqref{R8} and \eqref{R9}.
%
%
%
%

\end{proof}

\subsection{Weyl quantization estimates: proof of Lemmas \ref{moyal_lemma} and \ref{comp_lemma}}\label{sec_weyl}
The purpose of this section is to prove Lemmas \ref{moyal_lemma} and \ref{comp_lemma}  that
estimate the remainder terms $r_i, \ i=0,2$ in the $L^2(\rset^{2N})$ norm, and avoid derivatives of high order $N$, using that the Hilbert-Schmidt norm of
operators can be bounded based on the $L^2(\rset^{2N})$ norm of their symbols, as illustrated in \eqref{trace_hs} and \eqref{normhs}. The precise estimates of  remainder terms in the Moyal expansion of the composition of two
Weyl quantizations are presented here using Hermitian properties of the operator valued exponential $\EXP{\frac{\IU}{2M^{1/2}}(  \nabla_{x'}\cdot \nabla_{p} -\nabla_{x}\cdot \nabla_{p'}) }$.

The Moyal expansions, see \cite{zworski},
\begin{equation}\label{moyal}
\begin{split}
A(x)\MP D(x,p)& =\EXP{-\frac{\IU}{2M^{1/2}} \nabla_{x'}\cdot \nabla_{p'} }A(x+x')D(x,p+p')\Big|_{x'=p'=0}\COMMA\\
C(x,p)\MP D(x,p) &= \EXP{\frac{\IU}{2M^{1/2}} (\nabla_{x'}\cdot \nabla_{p} -\nabla_{x}\cdot \nabla_{p'})}C(x,p)D(x',p')\Big|_{(x,p)=(x',p')}\COMMA \\
\end{split}
\end{equation}
are well defined for the symbols $A,C$ and $D$ in the Schwartz class, viewing the exponential as a Fourier multiplicator.
We begin with the first expansion.
\subsubsection{The case $A(x)\MP D(x,p)$.}\label{subsecAB} We study the remainder term for the expansion of the exponential using
\[f_{xp}(x',p'):=A(x+x')D(x,p+p')\]
and apply   the  Fourier transform $\FT$  %
defined for $f(x',p')$ by
\begin{equation}\label{fourier_def}
\FT \{f\}(\xi_x, \xi_p):= %
\int_{\rset^{2N}}f(x',p')\EXP{-i(x'\cdot \xi_x +p'\cdot \xi_p)} \Rd x'\Rd p'\PERIOD
\end{equation}
The remainder can be evaluated by the inverse Fourier transform
\[
\begin{split}
&\EXP{-\frac{\IU}{2M^{1/2}} \nabla_{x'}\cdot \nabla_{p'} }f_{xp}(x',p')\Big|_{x'=p'=0}\\
&= %
(\frac{1}{2\pi})^{2N} 
\int_{\rset^{2N}} \FT  f_{xp}(\xi_x,\xi_p) \EXP{\frac{\IU}{2}M^{-1/2}\xi_x\cdot\xi_p} \Rd\xi_x\Rd\xi_p\\
\end{split}
\]
and Taylor expansion of the exponential function  
\begin{equation}\label{exp_def}
\begin{split}
&\EXP{-\frac{\IU}{2M^{1/2}} \nabla_{x'}\cdot \nabla_{p'} }f_{xp}(x',p')\Big|_{x'=p'=0}\\
&= %
(\frac{1}{2\pi})^{2N} 
\int_{\rset^{2N}} \FT  f_{xp}(\xi_x,\xi_p) \EXP{\frac{\IU}{2}M^{-1/2}\xi_x\cdot\xi_p} \Rd\xi_x\Rd\xi_p\\
&= %
(\frac{1}{2\pi})^{2N}
\int_{\rset^{2N}} \FT  f_{xp}(\xi_x,\xi_p) \Big(\sum_{n=0}^m(\frac{\IU\xi_x\cdot\xi_p}{2M^{1/2}})^n \frac{1}{n!} \\
&\qquad +(\frac{\IU\xi_x\cdot\xi_p}{2M^{1/2}})^{m+1} \frac{1}{m!} \int_0^1(1-s)^m \EXP{\frac{\IU s}{2}M^{-1/2}\xi_x\cdot\xi_p} \Rd s\Big) \Rd\xi_x\Rd\xi_p\\
&= 
\sum_{n=0}^m\frac{1}{n!} (-\frac{\IU\nabla_{x'}\cdot\nabla_{p'}}{2M^{1/2}})^n  f_{xp}(x',p')\Big|_{x'=p'=0} \\
&\qquad + (\frac{1}{2M^{1/2}})^{m+1}
\int_0^1  \EXP{-\frac{\IU s}{2}M^{-1/2}\nabla_{x'}\cdot\nabla_{p'}}
(-\IU\nabla_{x'}\cdot\nabla_{p'})^{m+1} \\
&\qquad\times f_{xp}(x',p')  \frac{(1-s)^m}{ m! }\Rd s\Big|_{x'=p'=0}\PERIOD\\
\end{split}
\end{equation}
The remainder is therefore
\begin{equation}\label{rxp}
\begin{split}
&r(x,p)\\
&=
(\frac{1}{2M^{1/2}})^{m+1}
\int_0^1  \EXP{-\frac{\IU s}{2}M^{-1/2}\nabla_{x'}\cdot\nabla_{p'}}
(-\IU\nabla_{x'}\cdot\nabla_{p'})^{m+1}\\
&\qquad\times \underbrace{f_{xp}(x',p')}_{A(x+x')B(x,p+p')}  \frac{(1-s)^m}{ m! }\Rd s\Big|_{x'=p'=0}\\
&= (\frac{1}{2M^{1/2}})^{m+1}
\int_0^1  \EXP{-\frac{\IU s}{2}M^{-1/2}\nabla_{x}\cdot\nabla_{p}}
(-\IU\nabla_{x}\cdot\nabla_{p})^{m+1}\\
&\qquad\times {A(x)B(x',p)}  \frac{(1-s)^m}{ m! }\Rd s\Big|_{x'=x}\PERIOD\\
\end{split}
\end{equation}

\subsubsection{The case $C(x,p)\MP D(x,p)$.} 
As above we obtain
\[
\begin{split}
&C(x,p)\MP D(x,p)\\
&=\EXP{\frac{\IU}{2M^{1/2}} (\nabla_{x'}\cdot \nabla_{p} -\nabla_{x}\cdot \nabla_{p'})}C(x,p)D(x',p')\Big|_{(x,p)=(x',p')}\\
&= 
\sum_{n=0}^m\frac{1}{n!} 
(\IU\frac{\nabla_{x'}\cdot \nabla_{p} -\nabla_{x}\cdot \nabla_{p'}}{2M^{1/2}})^n C(x,p)D(x',p')\Big|_{(x,p)=(x',p')} \\
&\qquad + (\frac{\IU}{2M^{1/2}})^{m+1}
\int_0^1  \EXP{\frac{\IU s}{2}M^{-1/2}(\nabla_{x'}\cdot\nabla_{p}-\nabla_{x}\cdot \nabla_{p'})}
(\nabla_{x'}\cdot \nabla_{p} -\nabla_{x}\cdot \nabla_{p'})^{m+1}   \\
&\qquad\times C(x,p)D(x',p')\Big|_{(x,p)=(x',p')} \frac{(1-s)^m}{ m! }\Rd s\PERIOD\\
\end{split}
\]
We can therefore write the remainder as
\begin{equation*}\label{rxp4}
\begin{split}
&R(x,p) \\
&=(\frac{\IU}{2M^{1/2}})^{m+1}
\int_0^1  \EXP{\frac{\IU s}{2}M^{-1/2}(\nabla_{x'}\cdot\nabla_{p}-\nabla_{x}\cdot \nabla_{p'})}
(\nabla_{x'}\cdot \nabla_{p} -\nabla_{x}\cdot \nabla_{p'})^{m+1}   \\
&\qquad\times C(x,p)D(x',p')\Big|_{(x,p)=(x',p')} \frac{(1-s)^m}{ m! }\Rd s\\
&=\int_0^1 %
 \EXP{\frac{\IU s}{2}M^{-1/2}(\nabla_{x'}\cdot\nabla_{p}-\nabla_{x}\cdot \nabla_{p'})}
 \tilde R(x,p,x',p',s)\big|_{x'=x, p'=p}\Rd s\COMMA
\end{split}
\end{equation*}
with 
\begin{equation}\label{Rx'p'}
\begin{split}
&\tilde R(x,p,x',p',s)\\
 &:= (\frac{\IU}{2M^{1/2}})^{m+1}
(\nabla_{x'}\cdot \nabla_{p} -\nabla_{x}\cdot \nabla_{p'})^{m+1}
C(x,p)D(x',p')
\frac{(1-s)^m}{ m! }\Rd s\PERIOD\\
\end{split}
\end{equation}
Cauchy's inequality implies
\begin{equation}\label{r*_est}
\begin{split}
&\|R(z)\|^2_{L^2(\rset^{2N})}\\
& = \int_{\rset^{2N}} \big| \int_0^1 
\EXP{\frac{\IU s}{2}M^{-1/2}(\nabla_{x'}\cdot\nabla_{p}-\nabla_{x}\cdot \nabla_{p'})}
 \tilde R(z,z',s)\Rd s \big|_{z=z'}^2 \Rd  z\\
&\le \int_{\rset^{2N}} \int_0^1 
|\EXP{\frac{\IU s}{2}M^{-1/2}(\nabla_{x'}\cdot\nabla_{p}-\nabla_{x}\cdot \nabla_{p'})}
 \tilde R(z,z',s)\big|_{z'=z}^2\Rd s \Rd  z\\
& = \int_0^1 \int_{\rset^{2N}} 
\big|\EXP{\frac{\IU s}{2}M^{-1/2}(\nabla_{x'}\cdot\nabla_{p}-\nabla_{x}\cdot \nabla_{p'})}
 \tilde R(z,z',s)\big|_{z'=z}^2 \Rd  z\Rd s\COMMA\\
 \end{split}
 \end{equation}
 which proves Lemma \ref{moyal_lemma}.

\subsubsection{Estimates of $\|C\MP  D\|_{L^2}$ and $\|A\MP  D\|_{L^2}$}
To verify \eqref{CD_fourier} we insert in \eqref{r*_est} %
the Fourier representation 
in the sense of distributions
of
the Dirac delta measure on $\rset^{2N}$, 
\[\delta(z-z')=(2\pi)^{-2N}\int_{\rset^{2N}} \EXP{\IU\omega\cdot (z-z')} \Rd 
\omega\COMMA\]
which implies 
\begin{equation}\label{delta_exp}
\begin{split}
 &\EXP{\frac{\IU s}{2}M^{-1/2}(\nabla_{x'}\cdot\nabla_{p}-\nabla_{x}\cdot \nabla_{p'})}
 \big(\tilde R(z,z',s)\EXP{\IU\omega_x\cdot (x-x')+\IU\omega_p(p-p')}\big)\\
& = \EXP{\IU\omega_x\cdot (x-x')+\IU\omega_p(p-p')}\\
&\qquad\times \EXP{\frac{\IU s}{2}M^{-1/2}\left((\nabla_{x'}-\IU\omega_x)\cdot(\nabla_{p}+\IU\omega_p)-(\nabla_{x}+\IU\omega_x)\cdot 
 (\nabla_{p'}-\IU\omega_p)\right)}
 \tilde R(z,z',s)\\
 &=\EXP{\IU\omega_x\cdot (x-x')+\IU\omega_p(p-p')}\\
&\qquad\times \EXP{\frac{\IU s}{2}M^{-1/2}\left((\nabla_{x'}\cdot\nabla_{p}-\nabla_{x}\cdot \nabla_{p'})
 -\IU\omega_x\cdot (\nabla_p+\nabla_{p'})+\IU\omega_p\cdot(\nabla_x+\nabla_{x'})\right)}
  \tilde R(z,z',s) \COMMA
\end{split}
 \end{equation}
and
 \begin{equation*}\label{isometry1}
 \begin{split}
& \int_{\rset^{2N}} 
|\EXP{\frac{\IU s}{2}M^{-1/2}(\nabla_{x'}\cdot\nabla_{p}-\nabla_{x}\cdot \nabla_{p'})}
 \tilde R(z,z',s)|_{z=z'}^2 \Rd  z\\
 &=
  \int_{\rset^{4N}} \big(\EXP{\frac{-\IU s}{2}M^{-1/2}(\nabla_{x'}\cdot\nabla_{p}-\nabla_{x}\cdot \nabla_{p'})}
 \tilde R^*(z,z',s)\big)\\
&\quad\times \big(\EXP{\frac{\IU s}{2}M^{-1/2}(\nabla_{x'}\cdot\nabla_{p}-\nabla_{x}\cdot \nabla_{p'})}
 \tilde R(z,z',s)\big)
 \delta(z-z')\Rd  z\Rd z'\\
 &=\frac{1}{(2\pi)^{2N}}
  \int_{\rset^{6N}} \big(\EXP{\frac{-\IU s}{2}M^{-1/2}(\nabla_{x'}\cdot\nabla_{p}-\nabla_{x}\cdot \nabla_{p'})}
 \tilde R^*(z,z',s)\big)\\
&\quad\times\big( \EXP{\frac{\IU s}{2}M^{-1/2}(\nabla_{x'}\cdot\nabla_{p}-\nabla_{x}\cdot \nabla_{p'})}
 \tilde R(z,z',s)\big)
  \, \EXP{\IU\omega_x\cdot(x-x')} \EXP{\IU\omega_p\cdot(p-p')}
 \Rd \omega_x\Rd \omega_p \Rd  z\Rd z'. \\
 \end{split}
 \end{equation*}
 Integration by parts, property \eqref{delta_exp} and the fact that  
 the differential operators $\nabla_{x'}\cdot\nabla_{p}$ and $\nabla_{x}\cdot \nabla_{p'}$ commute
and are symmetric in $L^2(\rset^{2N})$ show that
 \begin{equation}\label{isometry}
 \begin{split}
 & \int_{\rset^{2N}} 
|\EXP{\frac{\IU s}{2}M^{-1/2}(\nabla_{x'}\cdot\nabla_{p}-\nabla_{x}\cdot \nabla_{p'})}
 \tilde R(z,z',s)|_{z=z'}^2 \Rd  z\\
  &=\frac{1}{(2\pi)^{2N}}
 \int_{\rset^{6N}}\tilde R^*(z,z',s) \bigg(\EXP{\frac{-\IU s}{2}M^{-1/2}(\nabla_{x'}\cdot\nabla_{p}-\nabla_{x}\cdot \nabla_{p'})}\\
&\quad\times \Big(\big( \EXP{\frac{\IU s}{2\sqrt M}(\nabla_{x'}\cdot\nabla_{p}-\nabla_{x}\cdot \nabla_{p'})}
 \tilde R(z,z',s)\big) \EXP{\IU\omega_x\cdot(x-x')} \EXP{\IU\omega_p\cdot(p-p')}\Big) \bigg) \Rd \omega_x\Rd \omega_p
 \Rd  z\Rd z'\\
 &=\frac{1}{(2\pi)^{2N}}
 \int_{\rset^{6N}}\tilde R^*(z,z',s)\Big(
 \EXP{\frac{-\IU s}{2\sqrt M}
 \left((\nabla_{x'}\cdot\nabla_{p}-\nabla_{x}\cdot \nabla_{p'})
 -\IU\omega_x(\nabla_p+\nabla_{p'})+\IU\omega_p(\nabla_x+\nabla_{x'})\right) }\\
&\qquad\times\big(
  \EXP{\frac{\IU s}{2}M^{-1/2}(\nabla_{x'}\cdot\nabla_{p}-\nabla_{x}\cdot \nabla_{p'})} \tilde R(z,z',s)\big) \Big)\\
  &\qquad\times \EXP{\IU\omega_x\cdot(x-x')} \EXP{\IU\omega_p\cdot(p-p')} \Rd \omega_x\Rd \omega_p
 \Rd  z\Rd z'\\
&=\frac{1}{(2\pi)^{2N}}
 \int_{\rset^{6N}}\tilde R^*(z,z',s)  \big( \EXP{\frac{-\IU s}{2}M^{-1/2}
 \left(
 -\IU\omega_x(\nabla_p+\nabla_{p'})+\IU\omega_p(\nabla_x+\nabla_{x'})\right) }  \tilde R(z,z',s) \big)
 \\
&\qquad\times \EXP{\IU\omega_x\cdot(x-x')} \EXP{\IU\omega_p\cdot(p-p')} \Rd \omega_x\Rd \omega_p
 \Rd  z\Rd z' \PERIOD\\
 \end{split}
 \end{equation}
 Let 
 \[
 \begin{split}
 u &=(z+z')/2\COMMA\\
 v &=(z-z')/2\COMMA
 \end{split}
\]
then by \eqref{Rx'p'} we can expand the derivative \[
(\nabla_{x'}\cdot \nabla_{p} -\nabla_{x}\cdot \nabla_{p'})^{m+1}\] in $\tilde R(z,z',s)$ and collect the
derivatives with respect to $z$ and $z'$ in the functions $\tilde C$ and $\tilde D$
as \[
\tilde R(z,z',s)=\sum_{n=1}^{m+1}\tilde C_n(z)\tilde D_n(z')=\sum_{n=1}^{m+1}\tilde C_n(u+v)\tilde D_n(u-v)=:r(u,v)\PERIOD
\]
The Fourier transform in the $u$-direction is the convolution 
\[
\begin{split}
\hat r(\xi,v)&:=\sum_{n=1}^{m+1}\int_{\rset^{2N}} \tilde C_n(u+v)\tilde D_n(u-v)\EXP{-\IU\xi\cdot u} \Rd  u\\
&=\frac{1}{(2\pi)^{2N}}\sum_{n=1}^{m+1} \int_{\rset^{2N}}
\EXP{\IU\zeta\cdot v}\FT  \tilde C_n(\zeta)  
\EXP{-\IU(\xi-\zeta)\cdot v}\FT  \tilde D_n(\xi-\zeta) \Rd \zeta\PERIOD
\end{split}
\]
The right hand side in \eqref{isometry} becomes
\begin{equation*}\label{hat_r2_2}
\begin{split}
&\frac{1}{(2\pi)^{2N}}
 \int_{\rset^{6N}}\tilde R^*(z,z',s)  \big(\EXP{\frac{-\IU s}{2}M^{-1/2}
 \left(
 -\IU\omega_x(\nabla_p+\nabla_{p'})+\IU\omega_p(\nabla_x+\nabla_{x'})\right) }
   \tilde R(z,z',s)\big) \\
&\qquad\times \EXP{\IU\omega_x\cdot(x-x')} \EXP{\IU\omega_p\cdot(p-p')} \Rd \omega_x\Rd \omega_p
 \Rd  z\Rd z'\\
& =
 \frac{1}{\pi^{2N}}
 \int_{\rset^{6N}}r^*(u,v) \EXP{\frac{-\IU s}{4}M^{-1/2}(
 -\IU J\omega\cdot \nabla_u) + 2\IU \omega\cdot v}\ 
   r(u,v)   \Rd \omega 
 \Rd  u\Rd v\\
 \end{split}
\end{equation*}
which can be written
 \begin{equation}\label{hat_r2}
\begin{split}
&\frac{1}{(2\pi)^{2N}}
 \int_{\rset^{6N}}\tilde R^*(z,z',s)  \big(\EXP{\frac{-\IU s}{2}M^{-1/2}
 \left(
 -\IU\omega_x(\nabla_p+\nabla_{p'})+\IU\omega_p(\nabla_x+\nabla_{x'})\right) }
   \tilde R(z,z',s)\big) \\
&\qquad\times \EXP{\IU\omega_x\cdot(x-x')} \EXP{\IU\omega_p\cdot(p-p')} \Rd \omega_x\Rd \omega_p
 \Rd  z\Rd z'\\
 &= 
  \frac{1}{\pi^{2N} (2\pi)^{2N}}
 \int_{\rset^{6N}} \hat r^*(\xi,v) \EXP{\frac{-\IU s}{4}M^{-1/2}(
  J\omega\cdot \xi) + 2\IU \omega\cdot v}\ 
   \hat r(\xi,v)  \Rd \omega 
 \Rd  \xi\Rd v\\
 &= \frac{1}{\pi^{2N}}
 \int_{\rset^{4N}} \hat r^*(\xi,v) 
   \hat r(\xi,v)
   \delta(2v- \frac{s}{4}M^{-1/2} J \xi)
 \Rd  \xi\Rd v\\
 &= \frac{1}{(2\pi)^{2N}}
 \int_{\rset^{2N}} |\hat r(\xi,\frac{s}{8}M^{-1/2}J\xi) |^2
 \Rd \xi\COMMA \\
\end{split}
\end{equation}
using definition \eqref{J_def} of the matrix $J$.

The next step is to determine 
\[
\OPER{r}(\xi,\gamma J\xi)=:\sum_{n=1}^{m+1}\OPER{r}_n(\xi,\gamma J\xi)\]
 for $\gamma:= \frac{s}{8}M^{-1/2}$ by using the convolution
\[
\begin{split}
|\hat r_n(\xi, v)| &=\frac{1}{(2\pi)^{2N}}|\int_{\rset^{2N}} \FT  C_n(\zeta)\EXP{i\zeta\cdot v}\FT  D_n(\xi -\zeta) \EXP{-i(\xi-\zeta)\cdot v}\Rd \zeta|\\
&\le \frac{1}{(2\pi)^{2N}}\int_{\rset^{2N}} |\FT  C_n(\zeta)||\FT  D_n(\xi -\zeta)| \Rd \zeta\\
\end{split}
\]
which  by Young's inequality, namely $\|f*g\|_{L^2}\le \|f\|_{L^2}\|g\|_{L^1}$, implies
\[
\begin{split}
& \frac{1}{(2\pi)^{2N}}
 \int_{\rset^{2N}} |\hat r_n(\xi,\frac{s}{8}M^{-1/2}J\xi) |^2\Rd \xi\\
 &=   \frac{1}{(2\pi)^{6N}}
 \min( \|\FT \tilde C_n\|_{L^2(\rset^{2N})}^2\|\FT \tilde D_n\|_{L^1(\rset^{2N})}^2,
  \|\FT \tilde C_n\|_{L^1(\rset^{2N})}^2\|\FT \tilde D_n\|_{L^2(\rset^{2N})}^2)\\
  &= \frac{1}{(2\pi)^{4N}}\min( \|\tilde C_n\|_{L^2(\rset^{2N})}^2\|\FT \tilde D_n\|_{L^1(\rset^{2N})}^2,
  \|\FT \tilde C_n\|_{L^1(\rset^{2N})}^2\|\tilde D_n\|_{L^2(\rset^{2N})}^2)\\
\end{split}
\]
that proves \eqref{CD_fourier} in Lemma \ref{comp_lemma} and combined with \eqref{r*_est}, \eqref{isometry} and 
\eqref{hat_r2} it also establishes  Lemma \ref{moyal_lemma}.

The estimate \eqref{CD} of $\|A(x)\MP  D(x,p)\|_{L^2(\rset^{2N})}$ follows similarly from \eqref{rxp} by 
integration by parts 
\[
\begin{split}
&\int_{\rset^{2N}} \big|\EXP{\frac{\IU}{M^{1/2}}\nabla_{x}\cdot\nabla_{p}}
{A(x)D(x',p)} \big|_{x'=x}^2 \Rd  x\Rd  p\\
&=\int_{\rset^{3N}} \big(\EXP{\frac{-\IU}{M^{1/2}}\nabla_{x}\cdot\nabla_{p}}
{A^*(x)D^*(x',p)}\big) \big(\EXP{\frac{\IU}{M^{1/2}}\nabla_{x}\cdot\nabla_{p}}
{A(x)D(x',p)}\big) \\
&\qquad\times \delta(x-x')\Rd  x\Rd x'\Rd  p\\
&=\frac{1}{(2\pi)^N}\int_{\rset^{4N}} 
{A^*(x)D^*(x',p)} 
\\
&\qquad\times
\EXP{\frac{-\IU}{M^{1/2}}\nabla_{x}\cdot\nabla_{p}}\Big(\big(\EXP{\frac{\IU}{M^{1/2}}\nabla_{x}\cdot\nabla_{p}}
{A(x)D(x',p)} \big)\EXP{\IU\omega(x-x')}\Big)
\Rd  x\Rd x'\Rd  p\Rd \omega\\
\end{split}
\]
and the Fouriertransform, $\FT_p$, in
the $p$-direction
\[
\begin{split}
&\int_{\rset^{2N}} \big|\EXP{\frac{\IU}{M^{1/2}}\nabla_{x}\cdot\nabla_{p}}
{A(x)D(x',p)} \big|_{x'=x}^2 \Rd  x\Rd  p\\
&=\frac{1}{(2\pi)^N}\int_{\rset^{4N}} 
{A^*(x)D^*(x',p)} \Big(\EXP{\frac{-\IU}{M^{1/2}}(\nabla_{x}+\IU\omega)\cdot\nabla_{p}}\big(\EXP{\frac{\IU}{M^{1/2}}\nabla_{x}\cdot\nabla_{p}}
{A(x)D(x',p)} \big)\Big)\\
&\qquad\times \EXP{\IU\omega(x-x')}\Rd  x\Rd x'\Rd  p\Rd \omega\\
&=\frac{1}{(2\pi)^{2N}}\int_{\rset^{4N}} 
{A^*(x)(\mathcal F_p D)^*(x',\xi)} \EXP{\frac{\IU}{M^{1/2}}\omega\cdot \xi }
{A(x)\FT_p D(x',\xi)} \\
&\qquad\times \EXP{\IU\omega(x-x')} \Rd  x\Rd x'\Rd  \xi\Rd \omega\\
&=\frac{1}{(2\pi)^{N}}\int_{\rset^{3N}} 
{A^*(x)(\mathcal F_p D)^*(x',\xi)} 
{A(x)\FT_p D(x',\xi)} \delta(x-x'+\frac{\xi}{\sqrt M}) \Rd  x\Rd x'\Rd  \xi\\
&=\frac{1}{(2\pi)^{N}}\int_{\rset^{2N}} 
{A^*(x'-\frac{\xi}{\sqrt M})(\mathcal F_p D)^*(x',\xi)} 
{A(x'-\frac{\xi}{\sqrt M})\FT_p D(x',\xi)}  \Rd  x'\Rd  \xi\\
&= \frac{1}{(2\pi)^{N}}\int_{\rset^{2N}} 
{|A(x'-\frac{\xi}{\sqrt M})|^2 |\FT_p D(x',\xi)|^2} 
 \Rd  x'\Rd  \xi\\
 &\le \|A\|^2_{L^\infty(\rset^{N})} \|D\|^2_{L^2(\rset^{2N})}\PERIOD
\end{split}
\]
The case $A(p)\MP  D(x,p)$ is analogous to $A(x)\MP  D(x,p)$.

In conclusion, these  %
estimates bound the remainder symbols
$r_0$ and $r_2$ used in Lemmas \ref{lemma1}--\ref{lemma13}.
Finally, using that Schwartz functions are dense in  $L^p(\rset^{2N})$, for any $p\ge 1$, we have proved
Lemmas \ref{moyal_lemma} and \ref{comp_lemma}.


\appendix
%
%
\section{Which density operator?}\label{Gibbs_section}
If the density operators $\EXP{-\beta\OPER{\bH}}$ and $\OPERW{\EXP{-\beta \bH_0}}$ would differ only little
it would not matter which one we use as a reference. The proof of Theorem~\ref{gibbs_corr_thm} shows that
observables based on these two operators differ by the small amount of order $\BIGO(M^{-1/2})$ when the number of 
particles, $N$, is small compared to $M$. Since we do not know if this difference is small for larger 
number of particles, we may ask which density operator to use.  The density operator $\OPER{\rho}_q:=\EXP{-\beta\OPER{\bH}}$
is a time-independent solution to the quantum Liouville-von Neumann equation
\begin{equation}\label{liouvilleVN}
\partial_t \hat\rho_t= \IU M^{1/2}[\hat\rho_t, \OPER{\bH}]
\end{equation}
while the classical Gibbs density $\EXP{-\beta \bH_0}$ is a time-independent solution to the classical Liouville  equation
\[
\partial_t \rho_t= -\{\rho_t, \bH_0\}\PERIOD
\]
The corresponding density matrix symbol $\rho_q$ is not a time-independent solution to the classical Liouville equation, since $0=\IU M^{1/2}(\rho_q\MP \bH-\bH\MP \rho_q)\ne \{\rho_q,\bH\}$, and the classical Gibbs density
is not a time-independent solution to the quantum Liouville-von Neumann equation, since
$\IU M^{1/2}(\EXP{-\beta \bH_0}\MP \bH_0-\bH_0\MP \EXP{-\beta \bH_0})\ne \{ \EXP{-\beta \bH_0},\bH_0\}=0$.
We are lead to the question which Gibbs density to use and why use any Gibbs measure. This question is analyzed regarding the time-dependence and the classical behavior in the following two subsections.


\subsection{Why should we use the Gibbs density?}\label{gibbs_why} 
In Statistical Mechanics books
the Gibbs density is often derived as the marginal distribution of 
a subsystem weakly coupled to a heat bath, where the
composite system is assumed to have the microcanonical distribution, see \cite{feynman}.
Here we give a variant of this derivation, assuming instead that the marginal distribution of the subsystem
is determined by the subsystem Hamiltonian.

 In molecular dynamics simulations one often wants to determine properties of a large macroscopic system with many particles, say $N\sim 10^{23}$. Such large particle systems cannot yet be simulated in a computer and one may then ask for a setting where a smaller system has similar properties as the large. Therefore, we seek an equilibrium density matrix that has the property that the marginal distribution for a subsystem has the same density as the whole system. We will below motivate how this assumption together with an assumption of weak coupling between the subsystem and the whole system leads to the Gibbs measure;
in fact it is enough to assume that the marginal distribution for the subsystem has an equilibrium density which is a function of the Hamiltonian for the subsystem.

%
%
%


Assume that the Hamiltonian symbol has been diagonalized and consider one component so that the Hamiltonians are scalar valued, with coordinates $z_s\in \rset^{2n}$ and Hamiltonian $H_s(z_s)$ in the subsystem and coordinates $z_b\in \rset^{2(N-n)}$ and Hamiltonian $H_b(z_b,z_s)$ for a large heat bath environment system, i.e. $N\gg n$. The whole system has the Hamiltonian $H(z)=H_s(z_s)+H_b(z_b,z_s)$ with $z=(z_s,z_b)\in\rset^{2N}$.
On the classical side, $F(H)$ for any
differentiable function $F:\rset\to(0,\infty)$ 
yields a non normalized solution to the time independent Liouville equation.
We assume that 
\begin{assumption}\label{Gassump1}
\begin{equation}\label{gibbs_assump_1}
\begin{split}
&\mbox{$F:\rset\to (0,\infty)$ is continuously differentiable,}\\
&\mbox{$\int_{\rset^{2N}} F(H(z)){\rm d} z$ and 
$\int_{\rset^{2(N-n)}} F(H_b(z_b,z_s)){\rm d} z_b$ are finite.}
\end{split}
\end{equation}
\end{assumption}
The non normalized marginal distribution for the subsystem is then
\[
\int_{\rset^{2(N-n)}} \frac{F\big(H_b(z_b,z_s)+H_s(z_s)\big)}{\frac{F(H_b(z_b,z_s))}{\int_{\rset^{2(N-n)}}F(H_b(\bar z_b,z_s))\Rd  \bar z_b}}\, \frac{F\big(H_b(z_b,z_s)\big)}{\int_{\rset^{2(N-n)}}F\big(H_b(\bar z_b,z_s)\big)\Rd  \bar z_b}\, \Rd  z_b
\]
which by the mean value theorem is equal to 
\[
\frac{F\big(H_b(\ZBSTAR ,z_s)+H_s(z_s)\big)}{\frac{F(H_b(\ZBSTAR ,z_s))}{\int_{\rset^{2(N-n)}}F(H_b(\bar z_b,z_s))\Rd  \bar z_b}}
\]
for some $\ZBSTAR \in \rset^{2(N-n)}$,
that depends on $H_s(z_s)$. We note that $\ZBSTAR $
 may be non unique. 
For given $F$ and $H_b$ we introduce the notation 
\[
\begin{split}
\ZBSTAR  & := \ZBSTAR \big(H_s(z_s)\big)\, ,\\
 \HBSTAR \big(z_s,H_s(z_s)\big)&:= H_b(\ZBSTAR \big(H_s(z_s)\big),z_s)\PERIOD
\end{split}
\] 
If the heat bath is uncoupled to the system, the value $\ZBSTAR $ does not depend on $H_s$, i.e. $\HBSTAR \big(z_s,H_s(z_s)\big) = \HBSTAR (z_s,0)$. We make two additional
assumptions.
\begin{assumption}\label{Gassump2}
The coupling between the heat bath and the system is weak
in the sense that
\begin{equation}\label{z_assump}
\HBSTAR \big(z_s,H_s(z_s)\big) = \HBSTAR (z_s,0) + o\big(H_s(z_s)\big)\COMMA
\mbox{ as $H_s(z_s)\to 0$.}
\end{equation}
\end{assumption}
This assumption expresses that the coupling energy between the subsystem
and the heat bath is  much smaller than the subsystem energy $H_s(z_s)$.
The second assumption is that the marginal distribution for the subsystem
is a function of the subsystem Hamiltonian.
\begin{assumption}\label{Gassump3}
For any continuous $H_b$
there is a function $f$, a constant $C$ and a point $z_s\in\mathbb{R}^{2n}$ such that
\begin{equation}\label{F_assump}
f\big(H_s(z_s)\big)=C\frac{F\Big(\HBSTAR \big(z_s,H_s(z_s)\big)+H_s(z_s)\Big)}{F\Big(\HBSTAR \big(z_s,H_s(z_s)\big)\Big)}
\end{equation}
for a set of continuous subsystem Hamiltonians $H_s$, containing a sequence $H_s^{(n)}$ such that  $H^{(n)}_s(z_s)\to 0$, not identically zero, 
with $\int_{\rset^{2n}} f(H^{(n)}_s(\bar z_s)){\rm d} \bar z_s$ finite.
\end{assumption}
We obtain with the definitions
\begin{equation}\label{LF_def}
\begin{split}
L(H) &:=\log F(H)\COMMA\\
\ell(H) &:= \log f(H)\COMMA
\end{split}
\end{equation}
and $\HSSTAR :=H_s(z_s)$ that
\[
\ell(\HSSTAR )= L\big(\HBSTAR (z_s,\HSSTAR )+ \HSSTAR \big) -  L\big(\HBSTAR (z_s,\HSSTAR )\big) +\log C\COMMA
\]
and note that the assumption that $F> 0$ is differentiable shows that $L$ and $\ell$ also are differentiable. 
The continuity of $F$ and $L$ yields \[\log C=\ell(0)\COMMA\] by choosing $\HSSTAR =0$, 
 so that
\[
\begin{split}
&\frac{\ell(\HSSTAR )-\ell(0)}{\HSSTAR } \\
&= \frac{L\big(\HBSTAR (z_s,\HSSTAR ) + \HSSTAR \big)-L\big(\HBSTAR (z_s,\HSSTAR )\big)}{\HSSTAR }\\
&=\int_0^1 L'\big(\HBSTAR (z_s,\HSSTAR ) + t\HSSTAR \big)\Rd t\COMMA
\end{split}
\]
which as $\HSSTAR \to 0$ 
combined with \eqref{z_assump} and the differentiability of $\ell$  establishes 
\[
\ell'(0)=L'\big(\HBSTAR (z_s,0)\big)\COMMA\quad\mbox{for all $\HBSTAR (z_s,0)$}\,.
\]
Consequently the function $L'$ is constant, since $\HBSTAR (z_s,0)=H_b(\ZBSTAR (0),z_s)$ can take any value in $\rset$
by varying $H_b$.
Let the constant be $-1/T$ so that
\[
L(H)= -\frac{H}{T} + \mbox{ constant,} \quad \mbox{ for any $H\in \rset$,}
\]
and we have by the definition \eqref{LF_def} obtained the Gibbs density
\[
F(H)= \EXP{-\beta H}\COMMA 
\]
which by \eqref{F_assump} shows that also the marginal is the Gibbs distribution
\begin{proposition}\label{gibbs_deriv11}
Assumptions \ref{Gassump1}, \ref{Gassump2} and \ref{Gassump3} imply
\[
f(H_s)= \frac{\EXP{-\beta H_s}}{\int_{\rset^{2n}} \EXP{-H_s(z_s)}\Rd  z_s}\PERIOD
\]
\end{proposition}

In conclusion, molecular dynamics simulations often seek
the property that the classical equilibrium for the subsystem is the same as for the larger environment, since the subsystem is supposed to model a larger system. We have shown that the classical Gibbs density is the only differentiable function with this property,
in the sense of the derivation based on Assumption~\ref{Gassump1}, \ref{Gassump2} and \ref{Gassump3}, while we do not know if the symbol for the quantum density matrix $\rho_q$ has the same property for a large number of particles. Therefore, we prefer to use the Weyl quantized classical Gibbs density $\OPERW{\EXP{-\beta \bH_0}}$
as our reference density, as in Theorem~\ref{gibbs_corr_thm_unif} and \ref{gibbs_corr_thm_analytic}, since its drawback of  being non time-independent solution to the quantum Liouville-von Neumann equation is mild, in the sense that the time dependent perturbation is small for long time,
as shown in Section \ref{time_dep_sec}.

%
\begin{remark}
Let us finally informally motivate why it seems difficult to
find time independent solutions to the quantum Liouville equation that are not  functions of the Hamiltonian.
An equilibrium density operator must commute with the Hamiltonian operator, by the Liouville-von Neumann equation, and  consequently it is diagonalized by the same transformation as the Hamiltonian. The diagonalized density operator is then a function of the eigenvalues of the Hamiltonian operator if they are distinct, by mapping the eigenvalues of the Hamiltonian to the eigenvalues of the density operator. 
Assume that this mapping can be extended to a continuous mapping $F:\rset\to\rset_+$.
We can then write the density operator as a function of the Hamiltonian, namely $\hat\rho=F(\hat H)$. 
\end{remark}

\subsection{Time-dependent density operators}\label{time_dep_sec}
One criterion for a density operator is  that it is a time independent or approximately time--independent 
solution to the quantum Liouville-von Neumann equation, so that measurements of the observable at different times  vary little.
Here we show that the time dependent perturbation of observables based on the initial density matrix $\OPERW{\EXP{-\beta \bar H}}$, namely the quantization of the classical Gibbs density, is small up to time $t\ll M$,
which  in some sense justifies to use the density matrix $\OPERW{\EXP{-\beta \bar H}}$.
 We include a motivation for an estimate which is uniform in the total number of particles $N$, 
 based on a small system weakly coupled to a large heat bath.

Let $\hat\rho_t$ be the solution to the quantum Liouville-von Neumann equation \eqref{liouvilleVN}, with initial data $\rho_0=\EXP{-\beta \bar H}$. Introduce the change of variables $\hat \rho_t=\hat\rho_0 + \hat v_t$,
then the perturbation $v:[0,\infty)\times \rset^{2N}\to \mathbb C^{d\times d}$ satisfies
\[
\partial_t \hat v_t = \IU M^{1/2} [\hat v_t,\OPER{\bar H}] + \IU M^{1/2}[\hat \rho_0,\OPER{\bar H}]\, , t>0,\quad \hat v_0=0\, .
\]
As in \eqref{diag_A} we see that $v$ is diagonal since $\rho_0$ is diagonal.
By Duhamel's principle we have
\[
\hat v_t= \IU  M^{1/2} \int_0^t \EXP{-\IU  M^{1/2}\sigma\OPER{\bar H}} [\hat \rho_0,\OPER{\bar H}]
\EXP{\IU  M^{1/2}\sigma\OPER{\bar H}}\Rd  \sigma\, .
\]
An observable based on this density matrix $\hat \rho_t$ can therefore be written
\[
\TR (\OPER{\bar A}_\tau\OPER{\bar B}_0\hat \rho_t)
= \TR (\OPER{\bar A}_\tau\OPER{\bar B}_0\hat \rho_0)
+\TR (\OPER{\bar A}_\tau\OPER{\bar B}_0\hat v_t)
\]
where the time dependent perturbation takes the form
\[
\begin{split}
&\TR (\OPER{\bar A}_\tau\OPER{\bar B}_0\hat v_t)
=\IU  M^{1/2} \int_0^t \TR \Big(\OPER{\bar A}_\tau\OPER{\bar B}_0 \EXP{-\IU  M^{1/2}\sigma\OPER{\bar H}} [\hat \rho_0,\OPER{\bar H}]
\EXP{\IU  M^{1/2}\sigma\OPER{\bar H}}\Big)\Rd  \sigma\\
&=\IU  M^{1/2} \int_0^t \TR \Big(\EXP{\IU  M^{1/2}\sigma\OPER{\bar H}}\OPER{\bar A}_\tau \EXP{-\IU  M^{1/2}\sigma\OPER{\bar H}}\EXP{\IU  M^{1/2}\sigma\OPER{\bar H}}\OPER{\bar B}_0 \EXP{-\IU  M^{1/2}\sigma\OPER{\bar H}} [\hat \rho_0,\OPER{\bar H}]
\Big)\Rd  \sigma\\
&=\IU  M^{1/2} \int_0^t \TR \Big(\OPER{\bar A}_{\tau+\sigma}\OPER{\bar B}_\sigma  [\hat \rho_0,\OPER{\bar H}]
\Big)\Rd  \sigma\, .\\
\end{split}
\]
Let
\[
 \hat h:=[\hat \rho_0,\OPER{\bar H}]\COMMA
\]
then the Moyal expansion for composition  and 
$\{\bar H,\EXP{-\beta \bar H}\}=0$ imply
\[
\begin{split}
h&= \EXP{-\beta \bar H}\MP\bar H- \bar H\MP \EXP{-\beta \bar H}\\
&= \EXP{-\beta \bar H}\MP \frac{|p|^2}{2} {\rm I}- \frac{|p|^2}{2} {\rm I} \MP \EXP{-\beta \bar H}
+ \EXP{-\beta \bar H}\MP\Lambda(x) - \Lambda(x)\MP \EXP{-\beta \bar H}\\
&= 2\sum_{k=1}^\infty \EXP{-\beta\Lambda(x)}
\frac{{\rm i}^{2k+1}}{(2k+1)! (4M)^{(2k+1)/2}}
(\nabla_x\cdot \nabla_p)^{2k+1} \EXP{-\beta|p|^2}\Lambda(x)\, ,
\end{split}
\]
and we obtain by Lemma \ref{moyal_lemma}
\[
\begin{split}
h &=\EXP{-\beta\Lambda(x)}
\frac{{\rm i}^{3}}{3 (4M)^{3/2}}(\nabla_x\cdot \nabla_p)^{3}\EXP{-\beta|p|^2}\Lambda(x)
+\mathcal O(M^{-5/2})\\
&=\EXP{-\beta\bar H(x)}
\frac{{\rm i}^{3}}{3 (4M)^{3/2}}\Big(  (p\cdot\nabla_x)^3\Lambda(x) + 3(p\cdot\nabla_x)\Delta\Lambda(x)\Big) +\mathcal O(M^{-5/2})
\end{split}
\]
which by Lemma \ref{moyal_lemma} yields
$\|h\|_{L^2(\rset^{2N})}=\mathcal O(M^{-3/2})$.
The symbol $h$ is diagonal, since $\rho_0$ is diagonal.
We conclude that on the time scale $t\ll M$ the time dependent perturbation is
small and we obtain
 \begin{equation*}\label{trace_perturb}
 \begin{split}
\left|\frac{\TR (\OPER{\bar A}_\tau\OPER{\bar B}_0\hat v_t)}{
 \TR (\hat \rho_0)}\right| &=
\left|{\rm i}M^{1/2}
\frac{\int_0^t\int_{\rset^{2N}}
\TR(\bar A_{\tau+\sigma}\MP\bar B_\sigma h) {\rm d} z{\rm d}\sigma}{\int_{\rset^{2N}}
\TR\rho_0 {\rm d}z}\right| 
 \le C\frac{t}{M}\, .
 \end{split}
 \end{equation*}


The remainder $h$ includes error terms, as $(p\cdot\nabla_x)^3 \Lambda$ and $(p\cdot\nabla)\Delta \Lambda$, that typically are large proportional to $N$. 
To obtain estimates that do not increase with $N$, we need a new setting: we
consider as in Section~\ref{gibbs_why} a smaller system, with coordinates $(x_s,p_s)=z_s\in\rset^{2n}$, weakly coupled to a large heat bath, with coordinates $(x_b,p_b)=z_b\in\rset^{2(N-n)}$, where $n\ll N$. We assume weak coupling, as in \eqref{z_assump}, which here means that  $\Lambda(x_s,x_b)=\Lambda_s(x_s)+\Lambda_b(x_s,x_b)$  satisfies
\begin{equation}\label{coupling_small}
|\nabla_{x_s} \Lambda_b(x_s,x_b)|\ll 1\, .
\end{equation}
Let $z=(z_s,z_b)\in\rset^{2N}$.
To understand the heat bath perturbation with respect to the system Hamiltonian, $\bar H_s$,
we can, on the operator level, study 
$\EXP{-\gamma\hat{\bar H}_s}
\EXP{-\gamma\hat{\bar H}_b}=\EXP{-\gamma(\hat{\bar H}_s+\hat{\bar H}_b)-\gamma \hat H_\delta}$
based on the Baker-Campel-Hausdorff expansion for $\hat H_\delta$, which to leading order satisfies
\[
\hat H_\delta\simeq \frac{1}{2}[\hat{\bar H}_s,\hat{\bar H}_b]= -\Delta_{x_s}\Lambda_b(x_s,x_b)-2\nabla_{x_s}\Lambda_b(x_s,x_b)\cdot\nabla_{x_s}\, .
\]This perturbation is small in the sense that it is bounded with respect to $N$
provided the coupling \eqref{coupling_small} is weak.
Here $\gamma=\beta$ or $\gamma={\rm i}M^{1/2} t$ for $\beta,t\in\rset$.
Does the remainder $\widehat{\EXP{-\beta\hat{\bar H}}}$ also give a
perturbation that is bounded with respect to $N$?
Integration by parts as follows in fact shows that the perturbation caused by $h$ is bounded with respect to $N$.
Introduce the notation
$D_\sigma:=\bar A_{\tau+\sigma}\MP\bar B_\sigma(z)\, .
$
We have by Lemma \ref{moyal_lemma}
\begin{equation*}\label{trace_perturb20}
\begin{split}
&\int_{\rset^{2N}}\TR(\bar A_{\tau+\sigma}\MP\bar B_{\sigma} h){\rm d} z\\
&=\frac{1}{16M^{3/2}}\int_0^1\int_{\rset^{2N}}
 \TR\Big(D_\sigma(x',p) \EXP{\frac{{\rm i}r}{2M^{1/2}}\nabla_x\cdot\nabla_p} (\nabla_x\cdot\nabla_p)^3
 \big(\Lambda(x)\EXP{-\beta|p|^2/2}\big)\times\\
 &\qquad\times \EXP{-\beta\Lambda(x')}\Big)\Big|_{x=x'}
 {\rm d} x{\rm d}p(1-r)^2{\rm d}r\\
&=\frac{1}{16M^{3/2}}\int_0^1\int_{\rset^{2N}}\int_{\rset^{2N}}
 \TR\Big( \EXP{-\frac{{\rm i}r}{2M^{1/2}}\nabla_x\cdot\nabla_p} (-\nabla_x\cdot\nabla_p)^3
 \big(D_\sigma(x',p) \Lambda(x)\big)\times\\
 &\qquad\times\EXP{-\beta\Lambda(x')} \EXP{-\beta |p|^2/2}\Big)\Big|_{x=x'}
 {\rm d} x{\rm d}p(1-r)^2{\rm d}r
 \, .\\
 \end{split}
 \end{equation*}
We can write
\[
(\nabla_x\cdot\nabla_p)^3
\big(D_\sigma(x',p) \Lambda(x)\big)
=\sum_{|\alpha|=3}\partial^\alpha_p D_\sigma(x',p)\partial^\alpha_x\Lambda(x)
\]
and we assume that there is a constant $C$ such that
\begin{equation}\label{D_assump22}
\begin{split}
\|\partial^\alpha_x\Lambda\|_{L^\infty(\rset^N)}&\le C\, ,\quad \mbox{ for } |\alpha|\le 3,\\
\sum_{|\alpha|=3}\|\partial^\alpha_p D_\sigma\|_{L^2(\rset^{2N})} &\le C\, ,\\
\end{split}
\end{equation}
uniformly in $N$.
The first assumption means that three derivatives of the eigenvalue is bounded
and the second assumption is based on  the observable $D_\sigma$  defined by the initial observables $\bar A_0$ and $\bar B_0$. These two initial observables only depend on the system coordinates. The assumption of weak coupling is then related to the assumption for $D_\sigma$.
The system dynamics $z_s(t,z_0) \in\rset^{2n}$, with initial data $z_0=(x_s,p_s,x_b,p_b)$, depends only weakly on the heat bath coordinate $x_b$ through $\Lambda_b$. A motivation for the last assumption in \eqref{D_assump22} is the following: to leading order we have
\[D_\sigma(z_0)\simeq \bar A\big(0,z_s(\tau+\sigma,z_0)\big)\bar B\big(0,z_s(\sigma,z_0)\big)\]
and the derivatives $\partial_p^\alpha$  of the right hand side, with $z_0=(x,p)$,
are determined by the flows $z',z'',z'''$ as in \eqref{A_z_t}, based on the assumption
\eqref{R7}.

Lemma \ref{comp_lemma} implies that there is a constant $C$ such that
the time dependent perturbation has the bound
 \begin{equation*}\label{trace_perturb200}
 \begin{split}
&\left|\frac{\TR (\OPER{\bar A}_\tau\OPER{\bar B}_0\hat v_t)}{
 \TR (\hat \rho_0)} \right|\\
&\le \frac{\|\EXP{-\frac{{\rm i}}{2M^{1/2}}\nabla_x\cdot\nabla_p} (-\nabla_x\cdot\nabla_p)^3
\big(D_\sigma(x',p) \Lambda(x)\big)\|_{L^2(\rset^{2N})}\|\EXP{-\beta\bar H}\|_{L^2(\rset^{2N})}}{
\frac{M}{t}\|\EXP{-\beta\bar H}\|_{L^1(\rset^{2N})}}\\
&\le \frac{t\|\EXP{-\beta\bar H}\|_{L^2(\rset^{2N})}}{M
\|\EXP{-\beta\bar H}\|_{L^1(\rset^{2N})}}\|\partial^\alpha_x\Lambda\|_{L^\infty(\rset^N)}
\sum_{|\alpha|=3}\|\partial^\alpha_p D_\sigma\|_{L^2(\rset^{2N})}\\
&\le \frac{C}{M}\, ,
 \end{split}
 \end{equation*}
which holds uniformly in $N$.

\section{Numerical tests}\label{numerical_tests}
In this section we present precise formulations of the
numerical demonstrations in Section \ref{sec_comp_demo}.
We formulate a simple model problem in order to compare the Schr\"odinger density with the molecular dynamics density in \eqref{qc_density} and \eqref{cl_density}. Then we present a similar model for comparing the Schr\"odinger position correlation observable with its molecular dynamics approximation. In this section we describe the respective numerical methods that were used in numerical approximation of these quantities.
\subsection{Model problem formulation}\label{sec:equilibriumposobs}
We define the model Hamiltonian $\hat H = -M^{-1}\Id\Delta + V(x)$, where $V: \mathbb{R}\rightarrow \mathbb{R}^{2\times 2}$ and $\Id$ is the $2\times 2$ identity matrix. 
We choose to construct the potential $V$ that yields an avoided crossing with a variable spectral gap depending on the parameter $\delta\in \mathbb{R}$. The construction is done in such a way that the smallest energy gap, $2\delta$, appears at a single position, $x = 0$. We define  the matrix
\begin{equation}\label{eqn:potMatNoOsc}
\bar V(x) =
\begin{bmatrix}
  x + x^2 & \delta \\
  \delta & -x + x^2 \\
\end{bmatrix}, \quad x\in \mathbb{R},
\end{equation}
with the eigenvalues
\begin{equation*}
\bar{\lambda}_1(x) = x^2-\sqrt{\delta^2 + x^2} \text{, }\quad \bar{\lambda}_2(x) = x^2+\sqrt{\delta^2 + x^2}
\end{equation*}
and the normalized eigenvectors ${\psi}_1 := \bar{{\psi}}_1/\|\bar{{\psi}}_1\|_2$, ${\psi}_2 := \bar{{\psi}}_2/\|\bar{{\psi}}_2\|_2$ where
\begin{equation}\label{eqn:eigenvectorsVBar}
{\bar{\psi}}_1 = \begin{bmatrix}
  \frac{x-\sqrt{\delta^2 + x^2}}{\delta}  \\
  1 \\
\end{bmatrix}
\text{, }\quad {\bar{\psi}}_2 = \begin{bmatrix}
  \frac{x+\sqrt{\delta^2 + x^2}}{\delta}  \\
  1 \\
\end{bmatrix}.
\end{equation}
The derivatives with respect to $x$ of the eigenvectors ${\psi}_1$ and ${\psi}_2$ are of order $1/\delta$. We study how  the size of the spatial derivative of the eigenvalues impacts the approximation of the observables. Therefore we construct a family of matrices with the eigenvalues
\begin{equation}\label{eqn:eigenvaluesV}
\lambda_1(x) := \bar{\lambda}_1(x) + a\cos(bx)-1 \text{, }\quad \lambda_2(x) := \bar{\lambda}_2(x)\COMMA
\end{equation}
illustrated in Figure \ref{image:examplePotentialMatrix}.
Then we define the matrix-valued potential function $V(x)$
\begin{equation*}
V := \Psi D \Psi^*\COMMA\;\mbox{ where}\;\;
D := \left[\begin{matrix}
  \lambda_1 & 0  \\
  0 & \lambda_2 \\
\end{matrix}\right]\COMMA\;\mbox{ and}\;\;
\Psi := \left[\begin{matrix}
  {\psi}_1 \\
  {\psi}_2
\end{matrix}\right]\PERIOD
\end{equation*}

\subsection{Numerical approximation}\label{sec:numDeMo}
To solve the Schr\"odinger equation with the Hamiltonian $\OPER{H}$ we discretize the computational domain $\Omega:= [-6,6]$ using the uniform mesh
$x_k = -6+ k\Delta x$, $k = 0, 1, 2, ..., K$, 
$\Delta x = \frac{12}{K}$. The computational domain is chosen such that the quantum density $\mu_\mathrm{qc}(x)$, \eqref{qc_density}, approximately vanishes outside $\Omega$.
In the presented numerical results we use $K = 751$. 
The eigenvalue problem $\OPER{H}\Phi_n = E_n\Phi_n$ is approximated using the 2nd-order finite difference approximation of the Laplacian with the zero boundary conditions on the computational domain $\Omega$ resulting in the algebraic eigenvalue problem 
\begin{equation}\label{algebraic_eig}
\boldsymbol{H}_d \boldsymbol{\phi}_n = e_n\boldsymbol{\phi}_n
\end{equation}
with the matrix
\begin{equation*}
\begin{aligned}
&\boldsymbol{H}_d :=
	\frac{1}{2M(\Delta x)^2}\left[\begin{smallmatrix}
		h_{11,0} + 2 & h_{12,0} & -1 & & & &   & \\
		h_{21,0} & h_{22,0} + 2 & 0 & -1 & & & \mathbf{0}  & \\
		-1 & 0 & h_{11,1} + 2 & h_{12,1} & -1 & &   & \\
		 & -1 & h_{21,1} & h_{22,1} + 2 & 0 & -1 &   & \\
		 & & \cdot & \cdot & \cdot & \cdot & \cdot  & \\
		 & \mathbf{0}&  & & -1 & 0 & h_{11,K} + 2 & h_{12,K}\\
		 & & & & & -1 & h_{21,K} & h_{22,K} + 2
	\end{smallmatrix}\right]\COMMA
\end{aligned}
\end{equation*}
and approximation of the eigen-functions
\[
\boldsymbol{\phi}_{n} := \left[
\phi_{n,0,1},
\phi_{n,0,2},
\phi_{n,1,1},
\phi_{n,1,2}, \ldots,
\phi_{n,K,1},
\phi_{n,K,2}\right]^T\COMMA
\]
with the boundary condition $\phi_{n,0,i} = \phi_{n,K,i} = 0$ for $i = 1,2$.
The entries of the matrix are 
\[
h_{ij,k} := 2M(\Delta x)^2 V_{ij}(x_k)\COMMA\;\;i,j=1,2\;\;
k=0,\dots,K\PERIOD
\]
The algebraic eigenvalue problem has been solved using the \texttt{Matlab} function \texttt{eig}. In the numerical experiments reported here we used the parameters $a = 1$ and $b = 10$ in \eqref{eqn:eigenvaluesV}.

The quantum density $\mu_\mathrm{qc}(x)$, \eqref{qc_density}, is approximated by 
\begin{equation}\label{equation:schApprox}
\frac{\sum_{n}\left(\left|\phi_{n,k,1}\right|^2 + \left|\phi_{n,k,2}\right|^2\right) \EXP{-\beta e_n}}{\sum_k\sum_{n}\left(\left|\phi_{n,k,1}\right|^2 + \left|\phi_{n,k,2}\right|^2\right) \EXP{-\beta e_n}\Delta x} , \;\;  k = 0,1,...,N\PERIOD
\end{equation}
Since the computational domain $\Omega$ is chosen large enough so that $\mu_\mathrm{qc}(x)\approx 0$ for $x\not\in \Omega$ the imposed Dirichlet boundary condition does not introduce a significant numerical error.
To approximate the integrals in the quotients of the molecular density $\mu_\mathrm{cl}(x)$, and the weight $q_k$,  \eqref{cl_density} 
we use the trapezoidal rule.

\medskip
\noindent{\it Correlation dependent observables.}
For numerical simulation of the position correlation observable 
we deal with the expression $\TR(\OPER{x}_{\tau}\OPER{x}_{0}\EXP{-\beta\OPER{H}})$ where $\OPER{x}_{\tau} = \EXP{\IU\tau\sqrt{M}\OPER{H}}\OPER{x}_0 \EXP{-\IU\tau\sqrt{M}\OPER{H}}$ for which we introduce the approximations
\begin{equation*}
\EXP{\IU\tau\sqrt{M}\OPER{H}} \simeq \EXP{\IU\tau\sqrt{M}\boldsymbol{H}_d}\COMMA\;\;\;
\EXP{-\beta\OPER{H}} \simeq \EXP{-\beta \boldsymbol{H}_d},
\end{equation*}
and
\begin{equation*}
\OPER{x}_0 \simeq \left[
\begin{smallmatrix}
x_0 & & & & & & &\mathbf{0}\\
& x_0 & & & & & &\\
& & x_1 & & & & &\\
& & & x_1 & & & &\\
& & & & \cdot & & &\\
& & & & & \cdot & &\\
& & & & & & x_K & \\
\mathbf{0} & & & & & & & x_K 
\end{smallmatrix}
\right] =: X.
\end{equation*}
Next we define the matrices
\begin{equation*}
P := \begin{bmatrix}\boldsymbol{\phi}_{1} & \boldsymbol{\phi}_{2} & \boldsymbol{\phi}_{3} & \cdots & \boldsymbol{\phi}_{2K}\end{bmatrix} \text{, } \quad D := \left[\begin{smallmatrix}
e_1 & & & & \mathbf{0}\\
& e_2 & & &\\
& & \cdot & &\\
& & & \cdot &\\
\mathbf{0} & & & & e_{2K}
\end{smallmatrix}\right]\COMMA
\end{equation*}
where the column vectors $\boldsymbol{\phi}_{1},...,\boldsymbol{\phi}_{K}$ and the numbers  $e_1,...,e_K$ solve the algebraic eigenvalue problem  \eqref{algebraic_eig} with the mesh $x_k = -4.5 +k\Delta x$,  on the computational domain  $[-4.5, 4.5]$ for $\Delta x = \frac{9}{K}$ where $K = 2048$. Note that the orthogonal matrix $P$ diagonalizes the real symmetric matrix $\boldsymbol{H}_d$, so $\boldsymbol{H}_d = PDP^T$. 
Thus 
\begin{equation*}
\EXP{\IU\tau\sqrt{M}\boldsymbol{H}_d} = P\EXP{\IU\tau\sqrt{M}D}P^T\COMMA\;\;\mbox{and}\;\;
\EXP{-\beta \boldsymbol{H}_d} = P\EXP{-\beta D}P^T.
\end{equation*}
We approximate the left hand side of \eqref{eqn:analyticalCorrelation} by
\begin{equation*}
\frac{\TR \big(P\EXP{\IU\tau\sqrt{M}D}P^T X P\EXP{-\IU\tau\sqrt{M}D}P^T (XP\EXP{-\beta D}P^T + P\EXP{-\beta D}P^T X)\big)}{\TR(2P\EXP{-\beta D}P^T)}
\end{equation*}
and perform the calculations in \texttt{Matlab}.

For the right hand side of the estimate \eqref{eqn:analyticalCorrelation} we 
solve the Hamiltonian dynamics \eqref{eqn:HamiltonDynamics} using a position Verlet scheme, see \cite{ben}.
More specifically let $(x_k, p_l) = (k\Delta x, l\Delta p)$ be a partition of $\left[-4.5,4.5\right]\times \left[-4.5,4.5\right]$ where 
\[
\mbox{$\Delta x = \frac{9}{K_x - 1},  \Delta p = \frac{9}{K_p - 1}$ and $K_x=K_p=1000$.}
\] 
We compute for each $k,l\in \left\{-500, -499, ..., 499, 500\right\}$ the path from the dynamics \eqref{eqn:HamiltonDynamics} 
and we approximate the integrals on the right hand side of the estimate \eqref{eqn:analyticalCorrelation} with the 2-dimensional trapezoidal method.

\medskip
\noindent{\it Choice of parameters.} 
The model problem allows us to control the spectral gap by changing the parameter $\delta$ in the definition of $V(x)$. In the numerical simulations for the equilibrium densities we have chosen $M=1000$. Note that in the atomic units used here the mass ratio for proton-electron system is approximately $M\approx 1836$.  The temperature $T$ and the parameter $\delta$ were chosen such that the weight $q_1$ is kept the same and set to $q_1=0.8$. This choice guarantees that the contribution to the observable averaging from the excited state is not negligible. 



The simulation evaluating the error for the position correlation observable was done for the correlation time $\tau=0.2$ and the mass ratios up to $M=100$. Nonetheless, even this relatively small value of $M$ confirmed that the $L^{\infty}$-error is inverse proportional to the mass ratio $M$.

\end{document}